\documentclass[12pt,a4paper]{article}

\usepackage[a4paper,margin=2cm, marginparwidth=2cm]{geometry}

\usepackage{amsmath,amssymb,amsthm}

\usepackage[utf8x]{inputenc}

\usepackage{hyperref}

\usepackage{graphicx}

\usepackage{array}

\newcolumntype{+}{!{\vrule width 2pt}}

\newlength\savedwidth

\usepackage{caption}
\usepackage{subcaption}

\newtheorem{theorem}{Theorem}[section]
\newtheorem{lemma}[theorem]{Lemma}
\newtheorem{mydef}{Definition}[section]
\newtheorem{mydefs}[mydef]{Definitions}

\usepackage{multirow}

\newcommand{\bea}{\begin{eqnarray}}
\newcommand{\eea}{\end{eqnarray}}
\newcommand{\beq}{\begin{equation}}
\newcommand{\eeq}{\end{equation}}

\providecommand{\eqref}[1]{(\ref{#1})}

\newcommand{\figref}[1]{Fig.\ \ref{#1}}
\newcommand{\Figref}[1]{Fig.\ \ref{#1}} 
\newcommand{\tabref}[1]{Table~\ref{#1}}
\newcommand{\Tabref}[1]{Table~\ref{#1}}
\newcommand{\secref}[1]{Section~\ref{#1}}
\newcommand{\Secref}[1]{Section~\ref{#1}}
\newcommand{\appref}[1]{Appendix~\ref{#1}}

\newcommand{\lemmaref}[1]{Lemma~\ref{#1}}

\newcommand{\tsedef}[1]{\textsc{#1}}
\newcommand{\urlname}[1]{\href{http://#1}{\texttt{#1}}}

\newcommand{\equationref}[1]{Eq.~\ref{#1}}

\newcommand{\hmax}{h_\mathrm{max}}
\newcommand{\kin}{k^{\mathrm{(in)}}}
\newcommand{\kout}{k^{\mathrm{(out)}}}

\newcommand{\Dcal}{\mathcal{D}}
\newcommand{\Ecal}{\mathcal{E}}
\newcommand{\Gcal}{\mathcal{G}}
\newcommand{\Ncal}{\mathcal{N}}
\newcommand{\Pcal}{\mathcal{P}}
\newcommand{\Vcal}{\mathcal{V}}

\begin{document}
 \begin{flushright}
 \end{flushright}
 \vspace*{0.5cm}

\begin{center}
	{\Large\textbf{Order in Innovation}	}
	\\[0.5\baselineskip]
	{\large Martin Ho\textsuperscript{1,2\S},
	Henry CW Price\textsuperscript{3,4\S},
	Tim S Evans\textsuperscript{3,4\ddag},
	Eoin O\textsc{\char13}Sullivan\textsuperscript{1,2\ddag}
	}
	\\[0.5\baselineskip]
	\textbf{1} Centre for Science Technology \& Innovation Policy, University of Cambridge, Cambridge CB3 0HU, United Kingdom
	\\
	\textbf{2} Institute for Manufacturing, Department of Engineering,  University of Cambridge, Cambridge CB3 0HU, United Kingdom
	\\
	\textbf{3} Centre for Complexity Science, Imperial College London, London SW7 2AZ, United Kingdom
	\\
	\textbf{4} Theoretical Physics group, Department of Physics, Imperial College London, London SW7 2AZ, United Kingdom
	\\[0.5\baselineskip]
	
	%
	%
	\S These authors contributed equally to this work.
	
	\ddag These authors also contributed equally to this work.
	
	* Corresponding author: \texttt{wtmh3@eng.cam.ac.uk}
	
\end{center}

\begin{abstract}
\noindent
Is calendar time the true clock of innovation? By combining complexity science with innovation economics and using vaccine datasets containing over three million citations and eight regulatory authorisations, we discover that calendar time and network order describe innovation progress at varying accuracy. First, we present a method to establish a mathematical link between technological evolution and complex networks. The result is a path of events that narrates innovation bottlenecks. Next, we quantify the position and proximity of documents to these innovation paths and find that research, by and large, proceed from basic research, applied research, development, to commercialisation. 
By extension, we are able to causally quantify the participation of innovation funders. When it comes to vaccine innovation, diffusion-oriented entities are preoccupied with basic, later-stage research; biopharmaceuticals tend to participate in applied development activities and clinical trials at the later-stage; while mission-oriented entities tend to initiate early-stage research. 
Future innovation programs and funding allocations would benefit from better understanding innovation orders.
\end{abstract}

\section{Introduction: Why do we need to understand the order of innovations?}
\label{intro}

We all know time flows linearly in one direction. Innovation, on the other hand, is historically one-directional but nonlinear~\cite{RN285}. Therefore, studies that present innovation events on a linear calendar timescale alone cannot represent causality, importance, and convergence of innovation intermediaries. In multi-step reactions in chemistry, reactants do not jump straight to products; there are intermediaries with different activation energy and always a rate-determining step that the overall reaction cannot proceed faster than. Chemists often catalyse the rate-determining step to speed up the overall reaction. Likewise, an innovation process contains intermediary outputs, and we further propose that there are bottlenecks whose catalysis would accelerate the overall innovation process. Accordingly, we investigate:

In what \emph{order} did individual technological breakthroughs occur to realise innovation outcomes? And, by extension, in what \emph{order} did innovating entities support the most rate-limiting innovations along an order? To answer these questions, we prototype the use of a multilayer \tsedef{directed acyclic graph} (DAG) to order scientific and technological precursors of innovation breakthroughs.

We propose methods to understand innovation order because contemporary analytical regimes may fall short of systematic causal explanations of complex, multi-phase innovation. Evolutionary economics acknowledges the existence of innovation intermediaries: Science and technology are subject to evolution,
which, in turn, leads to constant changes\footnote{Evolutionary economists state that it is technologies and organisations that evolve; price, quantities, and GDPs are downstream changes.} in the macroeconomy~\cite{RN202}. Sociotechnical transition describes the emergence of new technologies via evolutionary intermediaries and their incorporation in the society. However, this approach has so far been limited to qualitative case studies of longitudinal innovation because there is a tradeoff between scope of the innovation being analysed and depth of causal explanation. As a result, sociotechnical transition theorists call for techniques that can analyse the ``heterogeneity and multi-dimensionality of large scale sociotechnical systems''~\cite{RN1295}.

We argue potential outcome reasoning
``potential outcome reasoning''  
in natural experiments, commonly used in policy evaluation, is not the most appropriate in establishing in evolutionary economics. Rooted in clinical statistics, potential outcome reasoning estimates the causal effect of a treatment variable on an outcome variable by randomly assigning subjects into treatment and control groups so that, on average\footnote{As permitted by central limit theorem and law of large numbers}, the treatment and control groups only differ by their treatment status and is independent of all other factors (\secref{natural_experiment} for details). In innovation ecosystems, however, it is unlikely that (quasi-)random assignment can be achieved because ideas respect no barriers: There is no such thing as a  ``natural control'' in innovation due to low marginal cost to adopt knowledge. Another requirement of natural experiment is that treatment, outcome, and all confounding variables be accounted for unless there is an appropriate instrumental variable. Not only is it impractical to regress all relevant variables in an innovation system, often, the representativeness of innovation variables are sensitive to time. To illustrate, a drug in a Phase I trial is focused on toxicity, whereas the same drug at Phase III relies on efficacy variables. Problems about potential outcome reasoning are not unique to innovation economics: epidemiologists, who chiefly use randomised controlled trials, struggle with different states of a same variable, the specificity of variables (e.g. what variable can exhaustively denote innovation?), the context dependence of causality, and using different types of evidence to arrive at one overall verdict~\cite{RN1636}.

Interestingly, graph theory is used alongside, rather than as an alternative to, randomised controlled trials in epidemiology~\cite{RN1635}. However, epidemiologists' use of DAGs is limited to non-parametric visual representations of variables in randomised controlled trials and \emph{a priori} exploration of causal variables~\cite{RN1632,RN1633}. Citation networks are a prime example of DAG being analytically applied to causally order innovation events~\cite{RN1427,RN996}. Their ability to support causal inference is, nevertheless, marred by incomplete data. This is firstly because citation networks typically rely on one type of data, either patents only or journal publications only, meaning not all technological maturities are represented. Secondly, citation datasets are typically generated using keyword searches only~\cite{RN1317,RN1316}. Searching for patents using keyword search, for example,  ``biofuel'', would not result in a citation network that captures early scientific advances, in, for example, genetic engineering because the future applications of these advances were unknown. Thirdly, contrary to natural experiments that are restrictive in the dataset being used, it is often a challenge to delimit a specific dataset to construct a citation network. For instance, some citation networks lose specificity by clustering millions of patents ever filed in a country crudely by patent classification codes.

As a DAG contains a causally ordered chain of events, provided the network data is sufficient and relevant, from there we can directly observe the causal path of input A to outcome B along with all causal intermediaries. Network science describes and analyses complex systems through abstraction: with nodes representing entities and edges representing a connection between a node pair. The network approach has been successful in deducing properties of real networks, such as the fat-tailed degree distribution (power law) and community behaviours (e.g. centrality) of entities~\cite{RN1647,RN1648}. However, attempts to deduce causal relations in multilayer networks remain scarce, but this is fundamental to understanding complex systems such as innovation.

This paper shows how to represent innovation order and demonstrates how this can better our understanding of innovation from an evolutionary perspective. \Secref{methods} develops an original method of applying graph theory to innovation; \secref{s:data} introduces the empirical data; \secref{results} interprets results; \secref{discussion} discusses the use of longest path to order innovation; \secref{conclusion} concludes.


\section{Methods} \label{methods}

In this section, we look at how we move from raw data to produce the citation networks which encode the multiplicity of innovation phases.  One important feature is our integration of multiple sources of data. Another key difference to earlier work is that the direction of time in a citation network is fundamental to our approach. We give a formal set of definitions in the Supplementary Material.


\subsection{Data}\label{s:data}

We create a multilayer citation network which is a directed acyclic graph (DAG) in order to observe innovation patterns
and to test the relationship between critical scheduling events and documents on or close to the longest paths in the network, as discussed later in \secref{discussion}.

Data on medical innovations is an excellent source, not only because the concept of translation is most established in medical research,
but also because new therapeutics are required by law to be reported and registered. In particular, we focus on vaccine approvals where there is excellent data available.

Vaccination confers long-lasting and protective immunity by presenting antigens of interest to elicit specific antibody production in recipients. Historically, vaccines present antigen through inactivated or attenuated version of whole or protein subunits of pathogens. Beyond efficacy, to prevent the spread of infectious agents, vaccines are administered to a large proportion of a population. Hence, vaccine must be safe and inexpensive. As a rapid countermeasure to such pathogenic outbreaks, other bottlenecks for vaccine platforms are manufacturability and ease of deployment. We outline the four vaccine platforms covered in this analysis and some technical events we expect to recover from the network in \tabref{tab:empirical_data}. In \appref{a:datasource} we give further details of the data sources used and the innovation events we expect.

Each network we create starts from a single document approving a particular vaccine, and this is the only source node in that DAG.

We obtain our data on clinical approvals from the US Food and Drug Administration (FDA), 
European Medicines Agency (EMA), and the UK Medicines and Healthcare products Regulatory Agency (MHRA). 
When a product is authorised by any of the three entities,
we use the first authorised date to represent novelty and scan for all available references from all three agencies' authorisations~\cite{RN1723,RN1724,RN1722,RN1721,RN1725,RN1726,RN1650,RN1728}.

\begingroup
\renewcommand{\arraystretch}{1.5}
\begin{table}[!ht]
\centering
\small
\caption{\label{tab:empirical_data}Information on the vaccines analysed here. The number of nodes and edges are those present in the multilayer citation network created from a multi-step snowball sample starting from the vaccine approval document.
}

\begin{tabular}{lllllll}
 \hline
\textbf{Vaccine}          &\textbf{Technology}  & \textbf{Disease} & \textbf{Developer} & \textbf{Year first} & \textbf{Source} & \textbf{Data}\\
\textbf{network}           & \textbf{platform}   & \textbf{targeted} & & \textbf{approved} &\textbf{node} & \textbf{source}\\
 \hline
Spikevax  & \multirow{2}{*}{mRNA} & \multirow{3}{*}{COVID-19}& Moderna &2020 & \cite{RN1723} & \cite{RN1723,RN1857,RN1858}\\
\cline{1-1}  \cline{4-7}
Comirnaty & &&BioNTech& 2020& \cite{RN1724} & \cite{RN1724, RN1859, RN1860} \\\cline{1-2}  \cline{4-7}

Vaxzeria   & \multirow{2}{*}{Viral Vector}  && AstraZeneca& 2020&\cite{RN1722} & \cite{RN1722, RN1861}\\ \cline{1-1} \cline{3-7}

Zabdeno  & & Ebola & Janssen &2020& \cite{RN1721} & \cite{RN1721}\\
\hline
Dengvaxia  & \multirow{2}{*}{Live Attenuated} & Dengue & Sanofi Pasteur&2019&\cite{RN1725} & \cite{RN1725,RN1862} \\
  \cline{1-1} \cline{3-7}
Imvanex &&Smallpox&Bavarian Nordic&2013& \cite{RN1726} & \cite{RN1726,RN1863,RN1864}\\
\hline
Nuvaxovid & \multirow{2}{*}{Subunit}  &COVID-19&Novavax&2022&\cite{RN1650} & \cite{RN1650, RN1865, RN1866}  \\
\cline{1-1} \cline{3-7}
Shingrix  && Shingles &GSK&2017&\cite{RN1728} &\cite{RN1728,RN1867}\\
\hline

\end{tabular}

\end{table}

\endgroup

\subsection{Innovation network}

We start by defining the key properties of the innovation networks used in our work. Formally, a network (or graph) is a set of nodes, and pairs of nodes can be connected by an edge. In our networks, each node represents a single document which is one of four types: an innovation outcome represented by regulatory authorisation,
a clinical trial, a patent, or an academic publication. So, our networks are examples of what are called \tsedef{multilayer networks}, for example see \cite{C21}, as each type of node can be visualised as placed on a different \tsedef{layer}, see \figref{fig:multilayer_framework}.
Our edges, written as $(u,v)$, are citations from one node $u$ to another node $v$
so our networks are examples of \tsedef{citation networks}. Note that edges in citation networks have a sense of direction as $(u,v)$ represents an entry listing document $v$ in the bibliography of a document $u$, not the other way round. So, citation networks are examples of what are known as directed networks.

Citation networks also have a sense of order since a document cannot cite a later document so for an edge $(u,v)$, document $v$ must have been published before\footnote{Our data gives a single date for each document but in reality one can associate several different `publication' dates: application and grant dates for patents, date first appeared online as opposed to the official publication date written in the text of a journal publication \cite{HBC15}, etc. So, the data used to build a citation network can have edges that go from an earlier to a later document at least according to any single date we assign to each document, something seen in any work with citation networks such as \cite{CGLE14}. To portray novelty consistently, we use the first published date for publications, priority date for patents, and start date for clinical trials.}
document $u$.
As a result, there should be no cycles (loops) in our networks.  That is, if we move from one node to a neighbour, respecting the direction of the edge, and then repeat these steps as often as we want (this defines what is called a `walk' in a network \cite{C21}), we will never return to the same node twice. Thus, our citation networks are examples of what are called \tsedef{directed acyclic graphs} (DAG).
The direction and the lack of cycles in a DAG are a direct result of a sense of order
that is present in all DAGs. In a citation network, the order is the arrow-of-time implicit in a citation network. This order in a DAG leads to several special properties, which we exploit in our work.

In practice, we find that our data initially gives networks where 0.07\% of all edges are part of a cycle, for example, due to authors citing each others' paper during journal submission or mislabelling. We always remove these cycles (as described below) to ensure reduce our the networks we analyse are always DAGs.

\subsection{Growing an innovation network}

The networks we use all start from a single seed node, known as the \tsedef{source node}, representing the regulatory marketing authorisation for one vaccine. This is because the regulatory decision represents the first time a therapeutic product is marketed and thus marks an innovation breakthrough. This regulatory authorisation node will be the newest node in each network we consider and so the only node in that network with no citations, that is, no incoming edges. This is the only node in our initial set of nodes denoted $\Vcal(0)$.

In the second step, we scan the regulatory authorisation for any publications, clinical trials, and patents. We denote these documents as part of the set of nodes $\Vcal(1)$ at a `depth' of one from the source vertex. An edge is added from the source regulatory document to each of these document nodes at depth one. 

In addition, since regulatory documents do not normally contain patent ids, we also locate precise patents $p$ associated with vaccines through supplementary information on drug manufacturer inserts and websites. We add a node $p$ and a link $(f,p)$ from the regulatory document $f$ to each associated patent.

Once we have all the patents $p$ associated with the regulatory document $f$, directly and indirectly through drug manufacturer inserts and websites and through the Intervention sections of clinical trial documents, we finish by looking at \tsedef{patent families}. Each of the patents we have found is part of a patent family, something mentioned in the patent information, giving us further patents, say $p_{a}$ where label $a$ identifies a patent in the same family as $p$. However, we do not add new nodes for each of these patents $p_a$. We do find all references from any associated patent $p_{a}$ to any further document, say $d$. However, all of these references are represented as links from the single $p$ node to document $d$, a node in the second level set $\Vcal(2)$, that is we add a link $(p,d)$. In some sense, the patent nodes $p$ at this first level represent all patents in the same family. We do not do this for patents at higher levels.

Another way we expand the patents in the early parts of our citation network is that we look at the clinical trials $t$ in the regulatory document $f$, where there is already a link $(f,t)$. We then search patent databases for therapeutic names recorded in the ``Intervention'' sections of each clinical trial document. New patents $p_t$ found this way are also added as nodes at the next level, part of $\Vcal(2)$, with links $(t,p_t)$. We also perform the same search for documents cited by patents in the same patent family as $p_t$ as noted above. 

Thirdly, we perform snowball sampling. That is, at the $\ell$-th step of sampling, we have a set of documents $\Vcal(\ell)$ which form the nodes most recently added to our DAG. We start the process from the set of documents $\Vcal(1)$, those one step away from the regulatory approval document. We follow the references in these $\Vcal(\ell)$ documents to create new edges, say edge $(v,r)$ from document $v \in \Vcal(\ell)$ to a document $r$ listed in the bibliography of $v$. If we have not encountered a document $r$ so far, then we add $r$ to the next set of documents to be considered, namely $\Vcal(\ell+1)$.  
This process produces an exponential growth in the number of documents, so we have to terminate this process at some stage.  We do choose to do this after three steps because of computational limitations. This leaves us with vertices defined by the four distinct sets $\Vcal(0)$, $\Vcal(1)$, $\Vcal(2)$ and $\Vcal(3)$. We also have all the edges defined by the references of the documents in the first three sets. The final step of this part of our process is to look at the references given in the last set of documents found, $\Vcal(3)$. If any document $v$ in this last set of documents found refers to a document $r$ we have already added to our network, then we add an edge $(v,r)$. Should this reference be to a document $r$ not currently in our data set, we do not add this document $r$ as a new node to our network, and neither do we add an edge to $r$.  

We limit the growth of the network to three iterations for two reasons: (i) the graphs would have grown exponentially in the first iteration of the algorithm, such that any further network growth will capture innovations so distant in the past that it would be unclear whether we can attribute them to the innovation outcome; and (ii) such a growth will result in an unmanageable graph due to computational and data-sourcing limitations. At the end of the process, our network will have several nodes with no outgoing edges, and these nodes are known as \tsedef{sink nodes}.

In doing this snowball sampling, we have access to the citation data on three types of document: clinical trials from \urlname{ClinicalTrials.gov}; patents from \urlname{Lens.org}~\cite{RN1697}; and publication data from \urlname{Dimensions.ai}~\cite{RN1696}. The exhaustiveness of the first two data sources rests on the fact that all drugs and biological products conducted under regulatory investigational new drug application must be registered on, \urlname{ClinicalTrials.gov} and that the FDA maintains a database of all drugs and biological products it has approved. Similar requirements exist for the European Medicines Agency, but we do not include any results based on approvals from. 

As with any data on citations, our results will be incomplete. References in the original documents may not be captured in our data sets for a number of reasons, such as errors in the original document, incorrect transcription from the primary source to our electronic sources for citations, or simply some documents are not in our databases such as press releases or preprints (`grey' literature). The legal framework behind vaccines means that our data on regulatory approval, clinical trials and patents is likely to be better than journal citations, but no data is perfect. We have not attempted to study the effect of errors in our data, rather relying on the large scale of our data to provide some statistical safety net.

Our process up to this point has provided a directed network where the nodes always have two additional labels. First, we record which of the four types of documents a node represents. The second node label gives a single date which we call the publication date: the priority date for patents, the start date for clinical trials, and the official publication date for an academic article. Additional node information about the funding of the research reported in any document is discussed later.

However, we also require that our network is acyclic. 
While in principle one document only ever refers to older documents, which guarantees an acyclic network, in practice there are always cycles in raw citation networks. These arise because documents are never created in a single moment of time. In practice, documents have a range of dates, from the first formal submission of a document (such as the application to hold a clinical trial, filing of a patent, depositing a paper on a preprint server) through to a final version of a document (the end of a clinical trial, the award of a patent or the physical publication date assigned to a journal article). For journal articles, there are many possible dates we could use \cite{HBC15} but they usually differ by the smallest amount, typically less than a year. For clinical trials and patents, the range of dates associated with these documents can be over several years.
Across the eight networks studied, the mean and standard deviation of the number of cycles per edge was $0.0022 \pm  0.0008$.
Our final step is to remove one edge from every cycle to produce a true directed acyclic graph\footnote{In practice, we use the \texttt{find\_cycle()} function in the \texttt{NetworkX} package~\cite{hagberg2008exploring}. This is used iteratively to remove all cycles. When a cycle is found, the first edge of the cycle is removed from the graph and the function is run again until no cycles can be found. Alternative approaches can be used to produce an acyclic graph \cite{CGLE14}.}.

\begin{figure}
  \centering
    \includegraphics[width=1\textwidth]{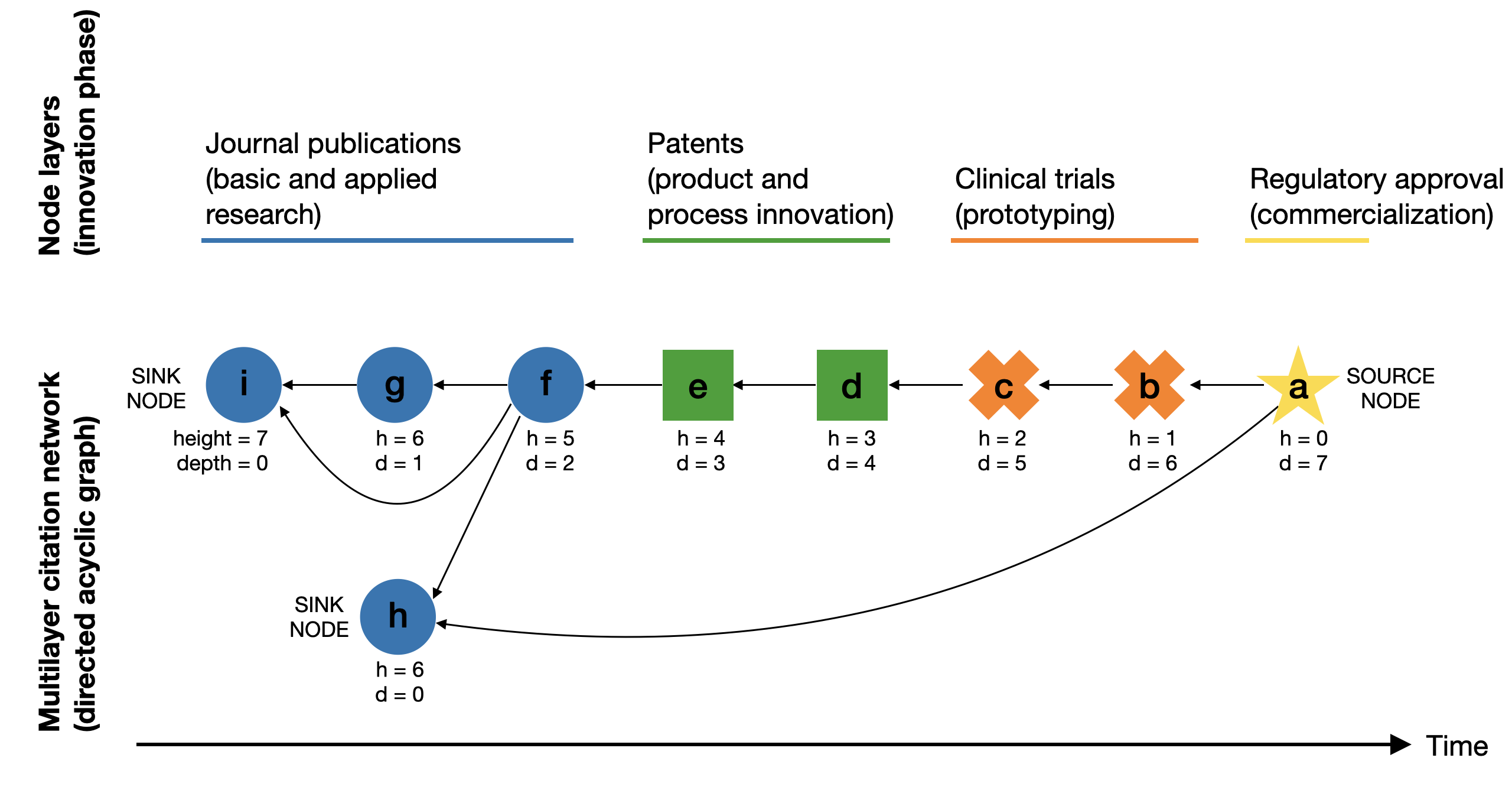}
  \caption{\textbf{Conceptual framework for multilayer innovation network.} Arrows represent direct citations from newer to older document, embedding causality and the flow of time. Colours represent different data sources and approximate innovation phases. Only one source node with height 0 exists, but multiple sink nodes with depth 0 may exist. In this illustrative figure, the critical path is formed by nodes \emph{abcdefgi} and represents the maximum distance between any two nodes in the graph. We further propose that the critical path in a multilayer innovation network represents a series of cumulative knowledge used for an innovation outcome. In contrast, any edge in a graph can be a shortest path, which might miss information on knowledge inheritance.}
  \label{fig:multilayer_framework}
\end{figure}

\subsection{Longest path in a citation network}

A path in a network is a sequence of distinct nodes, $\{u_0,u_1,\ldots,u_\ell\}$, where each consecutive pair of nodes forms an edge so $(u_i,u_{i+1})$ is an edge, from $u_i$ to $u_{i+1}$ if the edges are directed as here. In our case, paths are always moving backwards in time, as each document in a path can only cite an older document as the next step on a path. We will define the \tsedef{length of a path}, $\ell$, to be the number of edges in the path (one less than the number of nodes). In particular, we will focus on the \emph{longest} paths, not on the shortest paths normally encountered in network science, e.g.\ as in \cite{C21}. It is one of the special properties of a DAG that the longest paths are typically of a reasonable length, making them useful measures. See \secref{discussion} for a more detailed discussion of why we work with the longest path. We will define the \tsedef{distance} between pairs of nodes in a DAG to be equal to the length of the longest path between two nodes.

A key assumption in our work is that the most important steps for an innovation lie on or close to the longest path in an innovation citation network. We argue that this is because knowledge is built up incrementally.
Even when there are leaps in development, they are built on the success or failure of the most recent attempts to develop science.
Our longest paths contain many documents that made a small contribution to the final vaccine but we suggest that all the key documents will be there.
By way of contrast, had we used the shortest paths to study our innovation networks, the most widely used path in Network Science \cite{C21}, the shortest paths do not contain cumulative information of knowledge inheritance.
The shortest path would miss information because a document may cite important but old documents and so the shortest path will miss more recent critical developments, see \figref{fig:multilayer_framework}.
For a longer discussion of our choice and possible alternatives, including the differences between the longest path in a network and critical path in a schedule, see the discussion in \secref{discussion}.

In order to study the longest paths, it is convenient to define two standard properties of nodes in a DAG.
The \tsedef{height} $h(v)$ of a node $v$ is the maximum distance from the source node (the seed authorisation document) to the node $v$ while the \tsedef{depth} $d(v)$  is the maximum distance from node $v$ to any of the sink nodes.
The height of the DAG $\hmax$ is equal to the largest possible value of the height, $\hmax = \max \{h(v)\}$. The height of a DAG is also always equal to the largest depth of any node, which in our case is the depth of the seed node, the regulatory approval node and the only source node in our DAGs.

It is important to note that because our distance is integer valued, there can be many longest paths between any two nodes, not just one.
Further, while we argue that critical developments will lie on a longest path,
this is not something we can prove rigorously and, in any case, we can expect data used to form our citation network to be imperfect.
Therefore, it is extremely useful to be able to look at documents that are not on one of the longest paths to the source node
but instead lie on a path from source to sink that is one or two steps shorter than the longest path in the DAG.
That is we will also consider documents that are \emph{close} to a longest path. To quantify what we mean by `close' in this context, we define \tsedef{criticality} $c(v)$ for a node $v$ as:
\begin{equation}
	c(v)= \hmax - h(v) - d(v).
	\label{eq:criticality}
\end{equation}
Criticality $c$ takes integer values between zero and the largest possible value of height or depth,  the height of the DAG $\hmax$.
Any node which lies on a longest path of the DAG will have zero criticality. Equally, nodes which lie on a path from the source node to a sink node which is $c$ steps shorter than the longest path of the DAG will have a criticality value of $c$. Thus, the criticality value of a node can be thought of as the distance of a node to one of the critical paths down which the key innovations flow.
Applying \eqref{eq:criticality} to \figref{fig:multilayer_framework}, nodes $a$ to $i$  have a criticality of $0$, indicating they are on the longest path, whereas node $h$ has a criticality of 1.

In other words, for a given node in the innovation network, the node's height from the source node and depth from a sink node are uniquely defined. The novelty of our analysis is that we derive the longest path in the network by taking the criticality using height and depth. The criticality values $c(v)$ of nodes not only shows which nodes lie on longest paths (nodes with zero criticality) but also associated nodes lying on ``near-longest paths'' (small values of criticality). It is easy therefore for us to find other critical innovations which may have been missed by any method based on a single path, c.f. conventional main path analysis~\cite{RN1245,RN1231} which always returns a single path (see \secref{main_path_analysis} for further discussion of main path analysis).

\subsection{Measuring funder activity as a function of time}\label{funder_attribution}
For each node, there is a possibility that grants and funders linked to the research are reported in the associated document. We also look for specific entities in the acknowledgements for increased coverage. On \emph{Dimensions}, some publications, patents, and clinical trials are connected to grant nodes, providing additional details such as the value of the grant\footnote{However, we do not use monetary information in this study as we do not know how grants are split up by several publications, patents, or trials}, associated funder, and funding period. When measuring the effect of funders, we look at nodes and their associated funders, either directly via document-funder edges or indirectly via document-grant-funder edges, at one citation step from the grant attached to that project.

This information on the relationship between documents recorded in our multilayer citation DAG and funders means that every node $n$  
can be associated with a subset of funders\footnote{We could think of this as a new layer forming a bipartite network between document nodes and funder nodes. 
In our work, we only look at simple measures relating to funders, so such a network description of the funding landscape is unnecessary here. 
All the networks we discuss here are multilayer citation networks, no funding or grants are encoded in the network structures we analyse.}. 
We can now look at the properties of those nodes linked to any one funder, such as height and depth, and use various summary statistics, 
such as the median document height, to understand the different roles played by different funders in the innovation process.    

\section{Results}\label{results}

\subsection{Descriptive statistics}

The eight citation networks we study contain a  total of 569,660 nodes and 4,384,502 edges as shown in \tabref{tab:descriptive_stats}.
What is interesting is that the two vaccine platforms which were commercialised after 2020 contain roughly two-fold more publications
and five-fold more patents than the four new vaccines using more established vaccine platforms.
Similarly, the citation networks of the two new vaccine platforms contain more edges than that of the two older platforms.
More remarkable is the fact that the newer vaccine platforms contain twice the proportion of inter-layer edges\footnote{edges that involve two node types, e.g. publication-to-patent edges, as opposed to singular node type, e.g. patent-to-patent edges.} (5.7\%) than that of the older vaccine platforms (2.9\%).
This indicates that more translation is needed to commercialize the new vaccines.

\begin{table}[!ht]
\centering
\small
\caption{\label{tab:descriptive_stats}Basic network properties of the eight vaccine networks.
}

\begin{tabular}{lllllll}
 \hline
\textbf{Vaccine}        && \multicolumn{2}{c}{\textbf{Nodes}}  &&  & \multirow{2}{*}{\textbf{Edges}}\\
\textbf{network}           & \textbf{Publication}   & \textbf{Patent} & \textbf{Clinical trials}  & \textbf{Funders} & \textbf{Grants} &\textbf{} \\
 \hline
Spikevax  & 62,112 & 24,407 &10  & 1,286 & 25,043 & 786,563 \\

Comirnaty & 37,383 & 8,127 & 76  & 1,289 & 18,744 & 340,161  \\

Vaxzeria   &  58,210 & 32,367 & 5  & 1,274 & 21,528 & 648,877 \\

Zabdeno  & 77,359 & 47,145 & 9 & 1,371 & 27,561 & 953,002 \\

Dengvaxia  & 9,986  & 2,681 & 30  & 505 & 2,079 & 81,716 \\

Imvanex & 38,979 & 5,298 & 24  & 922 & 13,129 & 357,320\\

Nuvaxovid & 13,855 &  1,348 &  4  & 924 & 7,547 & 104,182 \\

Shingrix  & 12,987 & 6,993 & 22 & 753 & 6,288 & 174,881 \\
\hline

\end{tabular}

\end{table}

\subsection{Critical innovation path narrates causality in innovation}

Graphically, we follow \equationref{eq:criticality} and plot depth as a function of height.
We plot the data from the mRNA vaccine graph in Figure~\ref{fig:height_v_depth} to illustrate:
the  bottom left node represents regulatory authorization and the top right nodes represent the earliest nodes in the network.
The diagonal represents nodes which lie on at least one longest path of the DAG, our critical innovation paths while numerous sub-critical nodes populate the region above the diagonal.
From the hue of the diagram, we also observe a cluster of non-critical innovations at the region with low height and low depth.
We observe the same pattern in all eight DAGs as shown in \appref{all_figures}.
We set forth to inspect:
 (i) nodes that are critical,
 (ii) the order of critical nodes from oldest to newest,  and
 (iii) nodes that are of lowest criticalities to test the theoretical equivalence between critical schedule path and longest network path.

\begin{figure}
  \centering

    \includegraphics[width=0.7\textwidth]{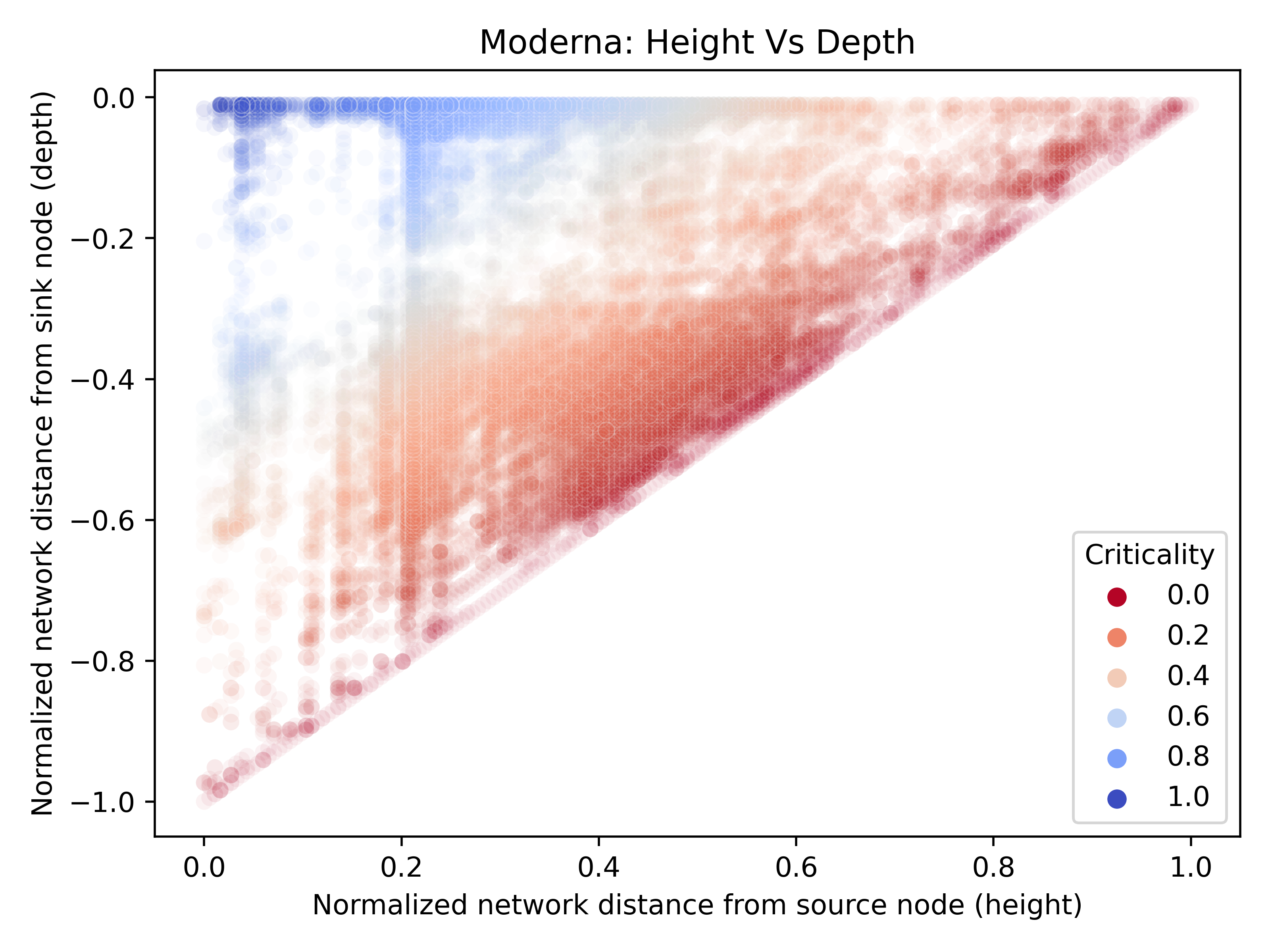}
  \caption{\textbf{Critical innovation path represented by low criticality nodes.}
  	Illustrative data from Moderna COVID mRNA vaccine DAG showing the depth of nodes plotted as a function of height, both normalised by the height of the DAG $\hmax$. The colour represents different values of criticality, again normalised by $\hmax$, with $0.0$ being the most critical and $1.0$ the least critical. The diagonal of the plot are documents lying on one of the longest paths where $c = 0$. The complete set of nodes with assigned criticalities are available at~\cite{figshare2}.
}
  \label{fig:height_v_depth}
\end{figure}

Looking at nodes whose criticality is strictly zero (i.e.\ most critical), in each DAG in \figref{fig:calendar_date_v_height} we see a mix of nodes representing publication, clinical trials, and regulatory authorisation.
If we relax the criticality threshold to consider nodes whose criticality is below 19.5\% of the maximum height, we begin to see many more publications, a few more clinical trials, and a few patents in this relaxed critical path region.
The order of the critical path, moving from high to low height nodes, always proceeds from publications, intertwined with a much smaller number of patents if in the version with the 19.5\% threshold, followed by phase 1, 2, and 3 clinical trials, before ending with the regulatory authorisation.
This sequence generally proceeds from basic research (publications), applied research (patents), development (clinical trials), to commercialization (regulatory authorisations).

A closer look at the critical path nodes unveils a logical sequence of technical progression.
For instance, the Moderna mRNA vaccine DAG has its longest paths formed by early attempts to apply mRNA as influenza vaccine platform~\cite{RN1679,RN1680},
using liposomal delivery system to enhance the expression kinetics of mRNA vaccine~\cite{RN1694,RN1695,RN1682}, methylation to enhance \emph{in vivo} antigen expression~\cite{RN1683},
the phases 1-3 clinical trials of mRNA COVID vaccines (NCT04283461, NCT04796896, NCT04847050, NCT04470427),
and finally the FDA emergency use authorisation letters~\cite{RN1684}  events that the scientific literature is well aware of~\cite{RN1371,RN998}.
In addition, the longest path of the same DAG also identified critical discoveries that may have been overlooked:
mRNA post-transcriptional modification mechanisms~\cite{RN1685,RN1686,RN1687,RN1688,RN1689,RN1690}
and early basic research about the potential to modify RNA to evade detection by toll-like receptors~\cite{RN1691,RN1692,RN1693}.

We are also interested in the identity of non-critical nodes. Having low criticality in a DAG does not mean the innovation is unimportant; it means events are not rate-limiting and can be perhaps parallelised.
Empirically, in the BioNTech/Pfizer COVID vaccine DAG, for example, nearly all reviewed nodes with low criticality are either clinical research about prevalence and risk factors for diseases non-specific to COVID. Low criticality events are likely non-critical to the approval of the vaccine by regulatory agency and, in this example, used to facilitate the design of clinical protocols.

\subsection{Calendar time against height reveals innovation speed }

The order inherent in a DAG gives a natural `clock' for the innovation process captured by our citation network.
It is interesting to see how this network order compares against calendar time.
To see this, we plotted the number of days between a document's date and the final regulatory authorisation against the height of that document in \figref{fig:calendar_date_v_height}. This shows that calendar date is strongly correlated with network order, but the relationship is non-linear.
Broadly speaking, the smallest calendar day at every height are nodes on the longest path
(i.e.\ they are nodes with 0 criticality).
Why do the slopes in \figref{fig:calendar_date_v_height} differ despite both describing the longest path? What is the difference between calendar days and depth? Time and network order in a citation network proceed in the same direction. This is because new documents can only cite older documents and, similarly, innovation is cumulative~\cite{RN1373}. However, their unit of progression differ: time proceeds in evenly spaced seconds or days, whereas network order proceeds in citation steps that are non-equidistant\footnote{An analogy is a clock where the ticks are spaced out differently.}. The latter means that the time gap and frequency of citations can increase and decrease over the course of an innovation lifecycle. This is possibly due to the cumulativeness of knowledge, entries and exits, consumer demand, and innovation policy. The different time gaps are observed across all eight vaccine networks (\figref{fig:all_date_v_height}).

\begin{figure}[!ht]
 \begin{center}
	\begin{tabular}{cc}
			\includegraphics[width=0.49\textwidth]{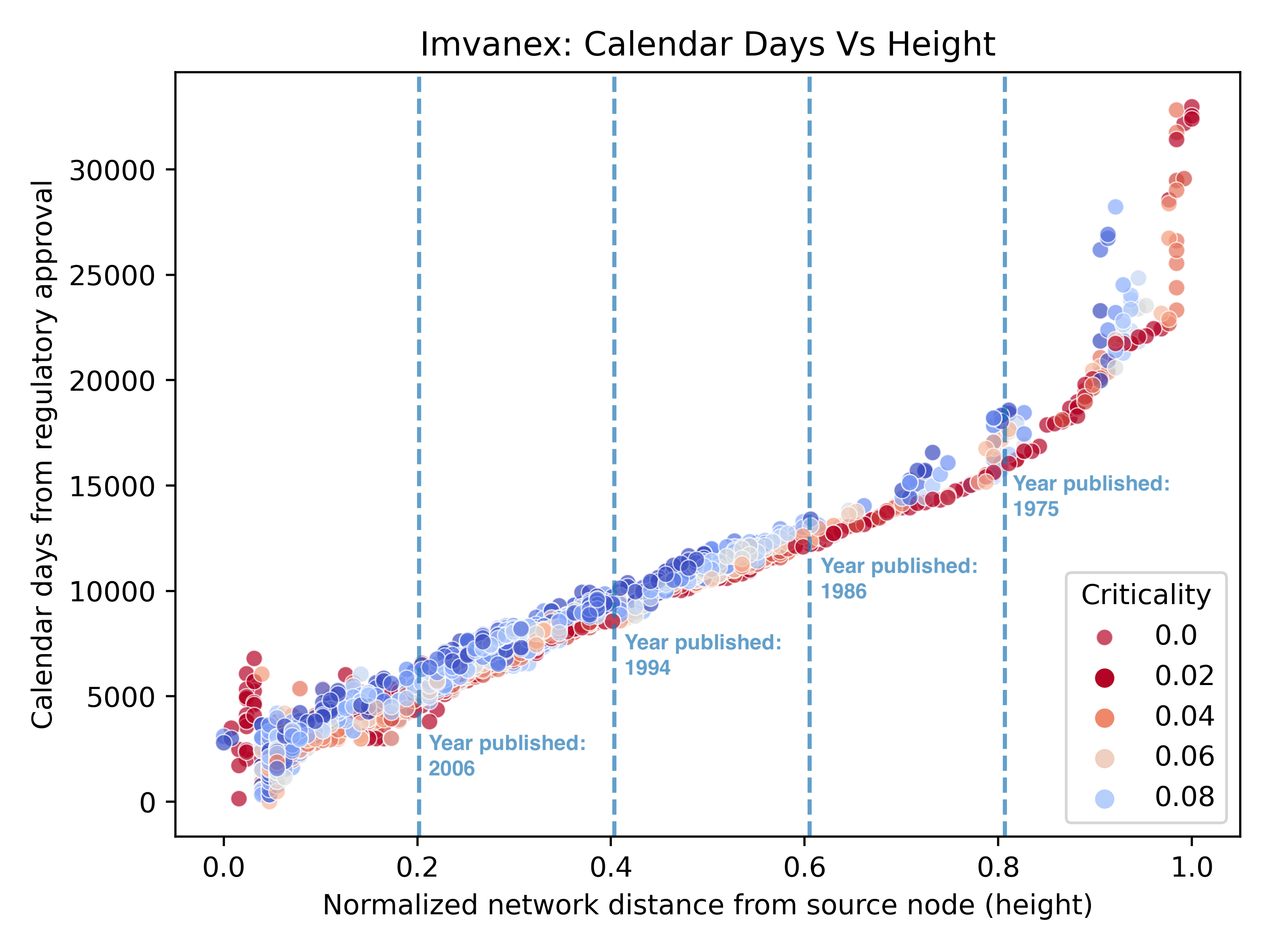}
			&
			\includegraphics[width=0.49\textwidth]{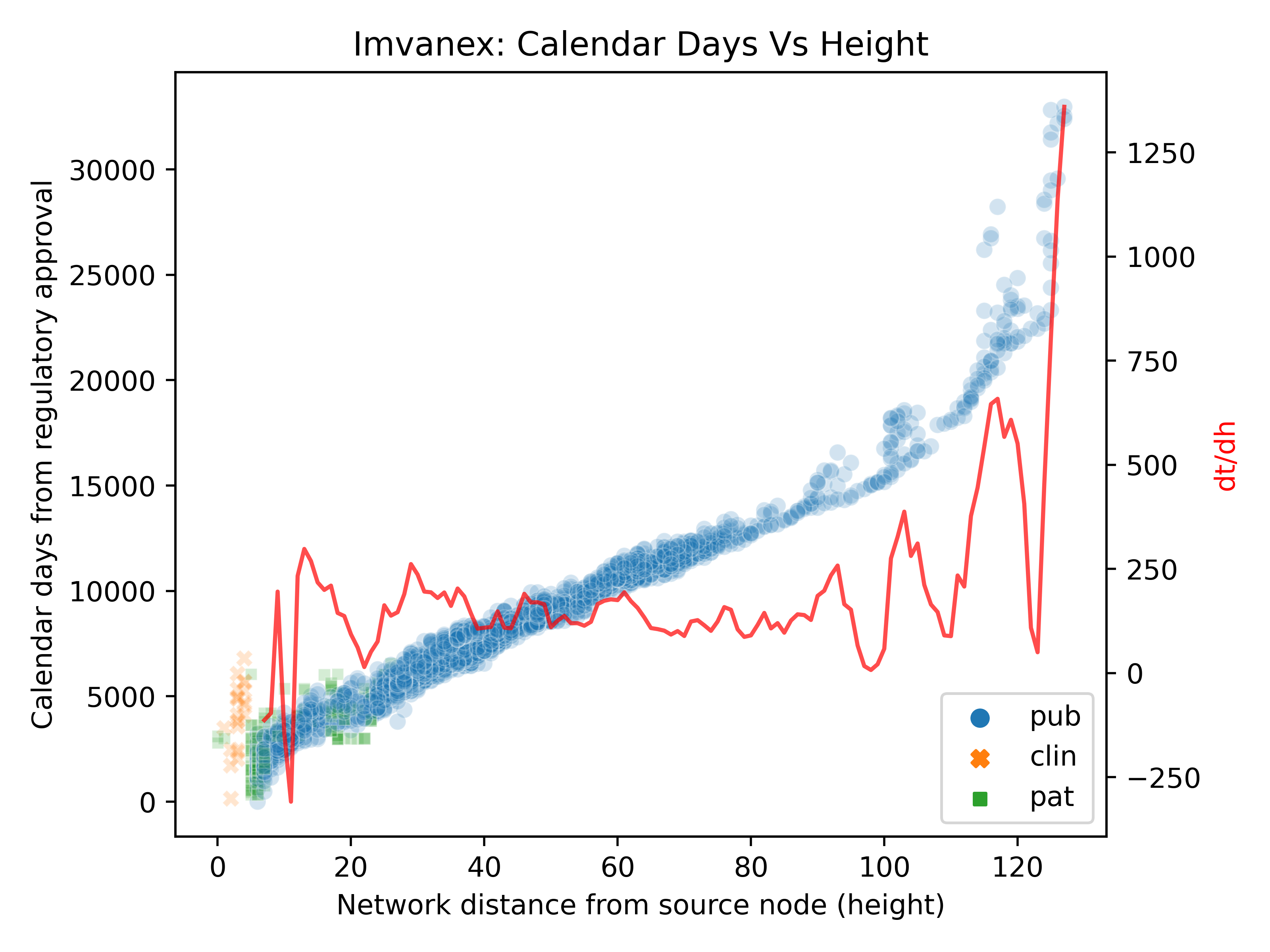}
			\\
			(a) \parbox[t]{0.45\textwidth}{Each point represents one document for nodes with low criticality (less than $0.1$) giving the time difference  in calendar days against height. Colour of the nodes gives the criticality value of the node.}
            &
            (b)\parbox[t]{0.45\textwidth}{Each point represents one document on a longest path (criticality zero) giving the time difference in calendar days against height.   Colour and shape of node indicate the type of document.
            	\\
            The red line shows an estimate of the derivative $d t/d h$ vs.\ $h$ where $t$ is publication date in calendar days and $h$ is height of a node. }
	\end{tabular}
	\caption{\textbf{Time as a function of height for the Imvanex network.}
				Height is normalised by the largest value so $0.0$ is the regulatory approval of the Imvanex vaccine \cite{RN1726} while $1.0$ is for nodes at the largest network distance from regulatory approval node. The time difference between the document publication date and the regulatory approval date is given in calendar days.
	}
		\label{fig:calendar_date_v_height}
 \end{center}
\end{figure}

Visually, \figref{fig:calendar_date_v_height}a suggests the publication date is rising at a constant rate for most critical nodes, 
but the slope increases for documents with a normalised height close to $1.0$.
We have tried to estimate the rate of change of publication date against height in \figref{fig:calendar_date_v_height}b 
by smoothing the data for those nodes on a longest path. 
On small scales, the change in height with calendar time fluctuates as seen in \figref{fig:calendar_date_v_height}b. 
On a larger scale, the trend overall shows that height and time are reasonably correlated. This could show that network order provides an alternative measure of innovation progress compared to calendar time (\secref{patent_prosecution}).

Vaccines have both forward (future) and backward (past) citation, and because the patent process often spans several years. We found that due to interactions between patent applicants and examiners during patent prosecution, the patent document may be updated with new references. We use the initial patent submission date as our patent publication date. A year or two into the patent process, a recent paper can be added to the application, one that was published after the patent was submitted. As a result, a patent may cite forward in time as well as the logically acceptable backwards in time. 
We could use the patent award date as our patent publication date, which would solve the problem with the example just given.  
However, now we run into problems with documents that cite a patent that is not yet approved and is a critical part of the innovation process. 
This again illustrates why our using the height of a node in our citation network can be a more consistent record of the logical order in the innovation process.

Second, the order of node types along the critical path in \figref{fig:calendar_date_v_height}b shows a clear progression\footnote{If we consider non-zero criticality nodes, we start to see more overlaps between node types.} of publications (basic research) to patents (applied research) to clinical trials (development).

Third, the rate of change fluctuates within each node type. The rates of change for critical patents and clinical trials fluctuate between 300 and -300 days, with the negative values indicating the problems of using a single publication date for patents, as these are revised over the several years it takes for a patent to be approved. On the other hand, it takes 50-1300 days for height to increase by one in the early critical journal publications, whereas more recent critical publications, those closer to the regulatory approval, have one year for a height increase of one, indicating an increasing rate of innovation.  A plausible explanation is towards late-stage (low calendar days from regulatory approval), the purpose of innovation activities are better known and focused towards the vaccine; more complete knowledge about and greater participation in the vaccine may have led to increasing frequency of critical innovations. In future studies, \figref{fig:calendar_date_v_height}b
could serve as a measurable interpretation of Utterback and Abernathy's~\cite{RN837} industry lifecycle model, which hypothesises that the ``rates of'' product and process innovation over time are convex and concave respectively; as well as the linear innovation model~\cite{RN958}, which prescribes that basic research, applied research, development, and production be carried out by different sets of actors one stage after another.

While the primary aim of this section is to propose new methods to measure the order and rate of innovation, we cannot help but observe some interesting differences across the vaccine data, see \figref{fig:all_day_per_height_v_height} in the Appendix: it always takes less time to make critical progress at later innovation phases. Future studies may compare data across sectors to reposition science policy's role in accelerating innovation.

\subsection{Division of innovation labour is quantifiable via network height}

Using the findings above, we demonstrate another real-world utility of innovation order. We portray the frequency of innovator funding as a function of network height to discern the innovation phases entities are supporting (\secref{funder_attribution} for methodological details. \Figref{fig:funder_network_order} shows illustrative data from the Novavax COVID protein subunits vaccine where we show the top five funders by number of nodes funded, three mission-oriented innovation agencies\footnote{Mission-oriented innovation agencies are entities that specifically fund frontier innovations to attain specific goals~\cite{RN960}. These entities are hypothesised to behave differently to diffusion-oriented agencies~\cite{RN445, RN1}, but this difference was awaiting quantification.}, and the top five pharmaceuticals by number of nodes funded. In the same table across the eight vaccines, we observe that the largest funders tend to occupy lower height, or late-stage; pharmaceuticals fund mostly late-stage documents; whereas mission-oriented innovation agencies are in the early- to mid-stage. Looking at calendar time, the median days of mission-oriented agencies and pharmaceuticals (2-19 years) are much closer to the regulatory approval than large funders are (10-27 years). This may indicate the strategies and division of labour among innovation entities: Larger funders fund basic and risk-averse research, mission-oriented agencies initiate high-risk research and translate discoveries to other funders, and pharmaceuticals playing their obvious commercialization role at the late-stage. However, we do not know whether this division of labour is deliberate or a result of their funding agenda\footnote{For instance, the mission statement of the National Institutes of Health is to ``to seek fundamental knowledge about the nature and behaviour of living systems...'' whereas that of the Biomedical Advanced Research and Development Authority is to ``develop and procure medical countermeasures...''}.

\begin{figure}
  \centering
    \includegraphics[width=1\textwidth]{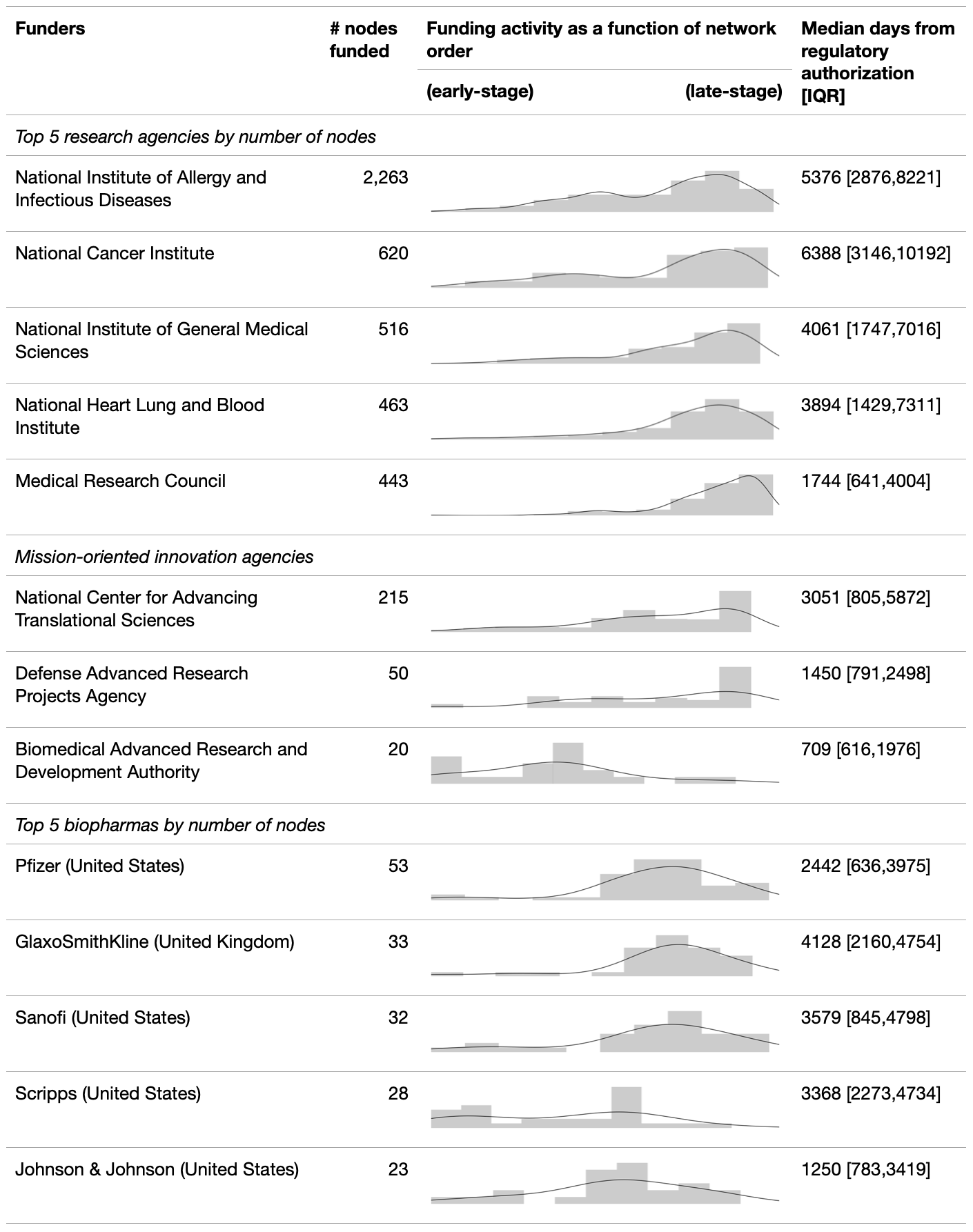}
  \caption{\textbf{Funding activities as a function of network height.} Illustrative data from Novavax vaccine. Funding activity is derived from kernel density estimations of heights of nodes associated with the funders; all entries are normalised on the same scale. Days are the  calendar days between the regulatory approval and documents funded by the funders. Some funders were involved in the manufacturing and procurement of vaccines, but these data are not available in the network; it is therefore likely that their actual funding activity curves are more skewed to the right.}
  \label{fig:funder_network_order}
\end{figure}

\subsection{Criticality of innovation funders is quantifiable via longest path}

We take the definition of ``critical'' from operations research to mean an event that delays the global project schedule when locally delayed (see \secref{discussion} for details). With the criticality information from the DAG, we also compute the criticality of funders, measured by the number of critical nodes funded by a particular entity divided by the total number of nodes funded by the entity in the DAG. We compare this performance metric with the citations received by documents the funder funds within the network. 

\Tabref{tab:funder_criticality} shows that funders who fund research that, in turn, leads to large number of citations are not necessarily the funders of critical research. Entities that fund a high proportion of critical nodes are avid removers of innovation bottlenecks, who, in turn, allow progression along the technological trajectory. One caveat is that we do not know whether these critical funders deliberately removed innovation hurdles or unintentionally produced innovations that were applied to advance a technology by other entities. Another caveat is that some funders specialise in advancing basic science without thought of practical ends, while others specialise in translating innovation.

\begin{table}[!ht]
\centering
\small
\caption{\label{tab:funder_criticality}\textbf{Funding performance measured by critical path hit rate.}
Illustrative data from Shingrix shingles subunits vaccine by GlaxoSmithKline. Critical path hit rate is the number of nodes funded on the critical innovation path divided by the number of nodes funded in the entire vaccine network (\secref{critical_path_hit_rate} for details). Total citation counts the number of in-degrees of nodes funded by a funder within the DAG; citations can be by publications, patents, clinical trials, or regulatory authroisation within the DAG. Defense-related and commercial funders may not disclose certain critical innovations. The true critical path hit rate of these funders are likely to be higher than observed.
}
\begin{tabular}{llllll}
 \hline
Funders on longest path & Funded   &  Citations       & Citations  & Critical & Critical \\
                         &nodes in &  from   & per funded & path     & path hit  \\
                         & entire   &  funded & node        & nodes     & rate (\%) \\
                         & network  & nodes         &           &           &           \\

 \hline
 \emph{Top 5 by critical path hit rate}               &  &  &  &  \\ \hline
  Defense Advanced Research  & 38 & 315 & 8  & 5  & 13.16 \\
 Projects Agency           &  &  &  &  &\\ \hline
Swedish Research Council     &  59  & 246   & 4 & 4  & 6.78 \\\hline
 GlaxoSmithKline (UK)      & 61 & 330  & 5 & 3  & 4.92 \\ \hline
 United States Public Health      & 217  & 1,139 & 5  & 6  & 2.76 \\
Service   &  &  &  & &  \\ \hline
   National Institute of Allergy and      & 2,803 & 18,138 & 6  & 67  & 2.39 \\
  Infectious Diseases   &  &  &  &  & \\ \hline

  \emph{Top 5 by citations in DAG}            &   &  &  &  &  \\ \hline
   National Institute of Allergy and      & 2,803 & 18,138 & 6  & 67  & 2.39 \\
  Infectious Diseases   &  &  &  &  & \\ \hline
   National Cancer Institute      & 1,366 & 9,099 & 7  & 23  & 1.68 \\\hline
   National Institute of General      & 753 & 5,208  & 7 & 11  & 1.46 \\
  Medical Sciences   &  &  &  &  & \\ \hline
   National Heart Lung and     & 451  & 3,193 & 7  & 8  & 1.77 \\
  Blood Institute  &   &  &  &  & \\ \hline
   National Institute of Diabetes and       & 382 & 2,555  & 7 & 0 & 0.00 \\
  Digestive and Kidney Diseases  &   &  &  &  \\ \hline

\end{tabular}

\end{table}

\subsection{Validation} 

We validate the critical path by checking for documents that also appear in literature reviews published by subject-matter experts\footnote{This is unlike most main path analyses (\secref{main_path_analysis}) that do not validate their results or only validate the keywords they use to generate the citation network.}. 
\Figref{fig:validation} shows the height versus depth diagrams (as described in \figref{fig:height_v_depth}) for the Moderna and Pfizer/BioNTech vaccines but with additional annotations showing 352 documents found in three literature reviews on mRNA vaccines~\cite{RN1738,RN998,RN1738}. 
We found that the critical path (the hypotenuses) are heavily populated by documents referenced in the literature reviews. 
\figref{fig:boxplots} shows that documents found both in the Pfizer/BioNTech vaccine network and literature review have lower median criticalities of 0.142 [0.0885, 0.0230] (where we give 25\% and 75\% in brackets) 
compared to documents found only in the former where the median is 0.504 [0.274, 0.788].
Similarly the figures for the Moderna network are 0.0710 [0.0328, 0.123] versus 0.0333 [0.169,0.574]. 
Kolmogorov-Smirnov tests indicate that criticalities of documents found in literature reviews is significantly different to that of documents not found in literature reviews 
(the p-values are always much less than $0.0001$),
validating the use of the critical path method to identify important innovation events.

\begin{figure}
     \centering
     \begin{subfigure}[b]{0.9\textwidth}
         \centering
         \includegraphics[width=\textwidth]{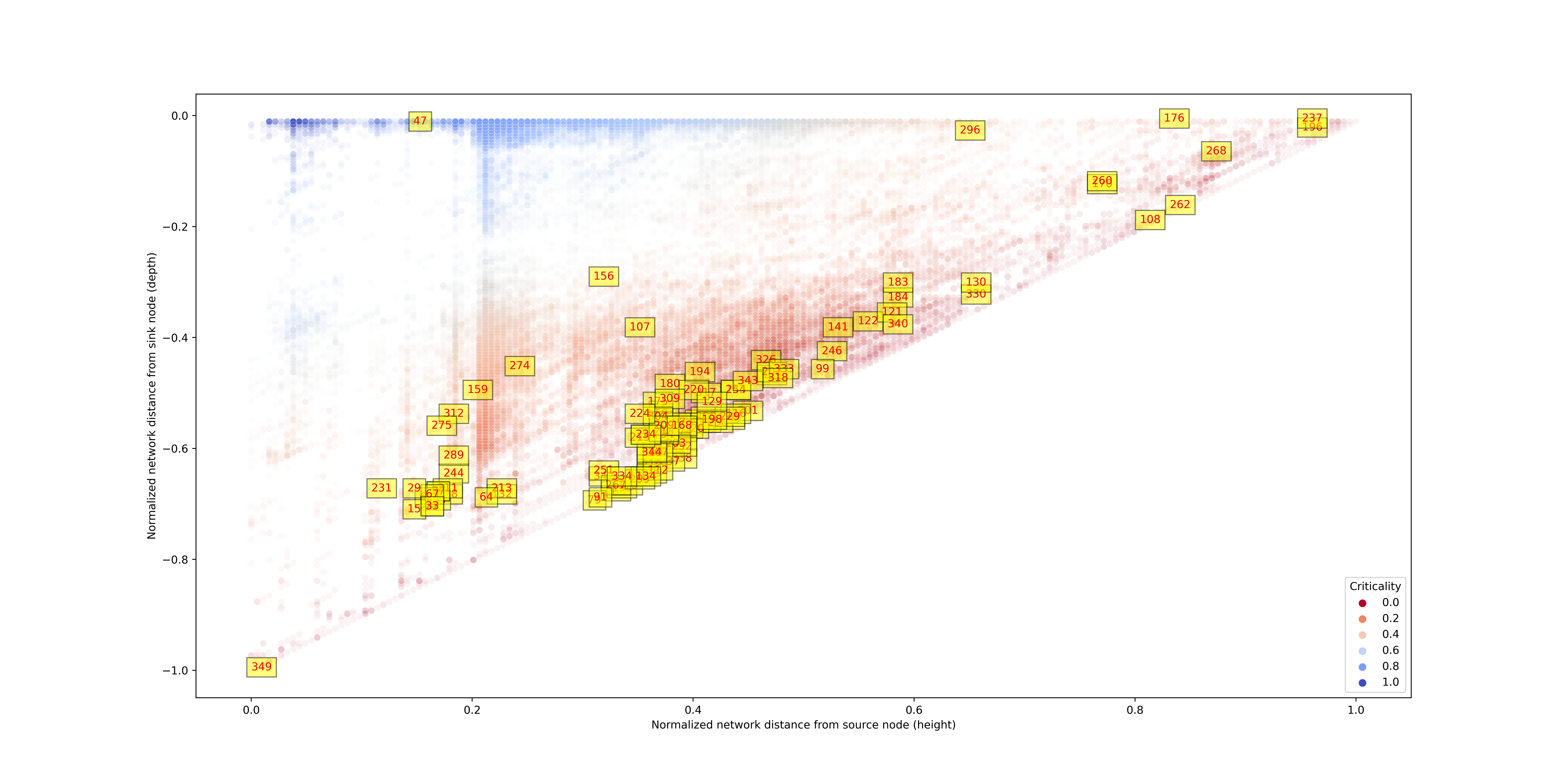}
         \caption{Moderna Spikevax}
     \end{subfigure}
     \hfill
     \begin{subfigure}[b]{0.9\textwidth}
         \centering
         \includegraphics[width=\textwidth]{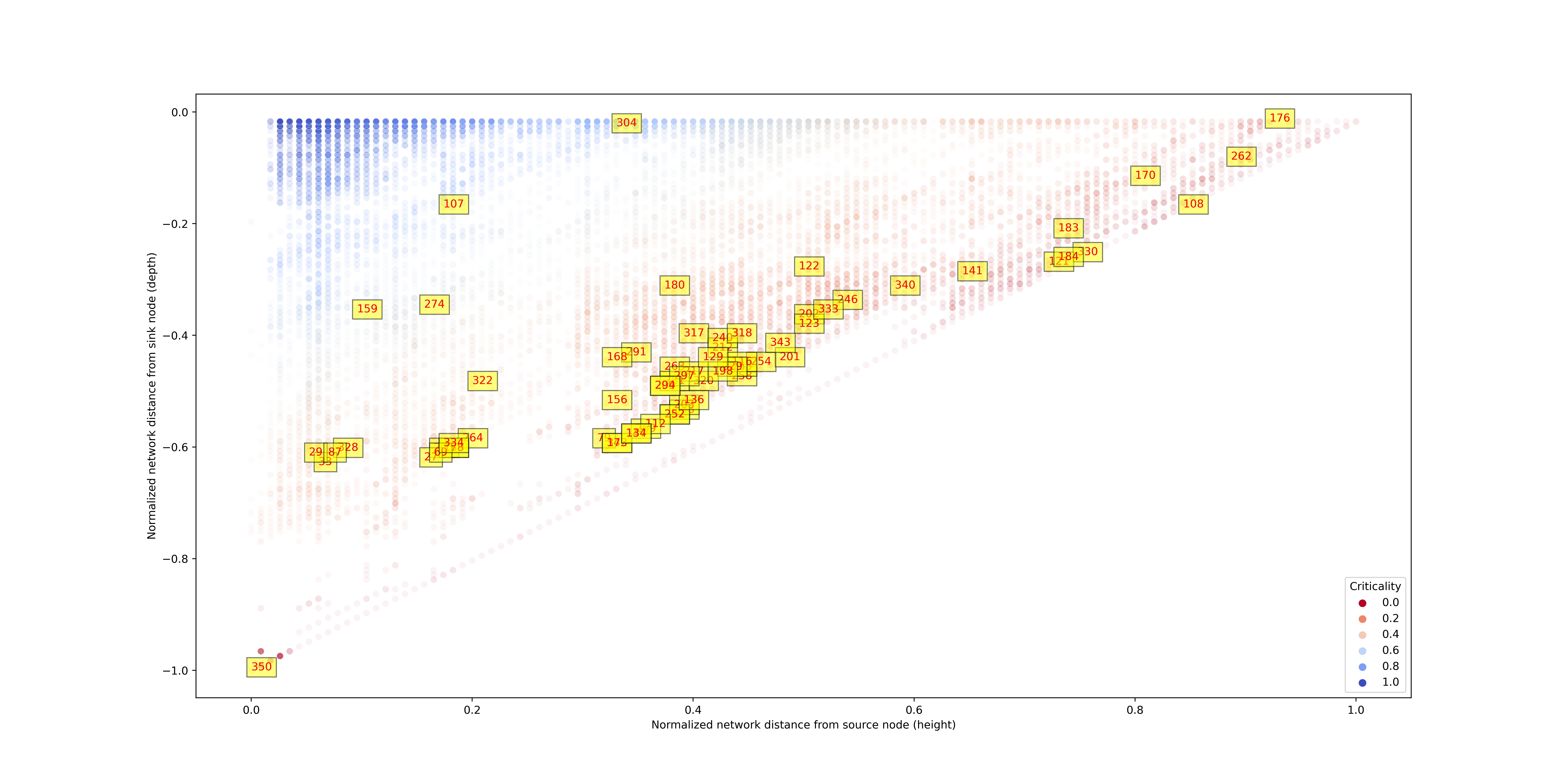}
         \caption{Pfizer/BioNTech Comirnaty}
     \end{subfigure}
     \hfill
 		\caption{\textbf{Critical innovation path of mRNA vaccine platform} represented by zero-criticality nodes. Depth: the maximum network distance from regulatory authorization to any node; Height: the maximum network distance from the earliest innovation events. A low depth or high height represents proximity to therapeutic approval in the citation network and vice versa; critical innovation path (red nodes): composed of nodes whose height ≈ depth meaning they are on the longest path of the network, approximating the importance of a node to the progression of the technology. A low distance from the longest path may indicate bottleneck to technological progress being overcome; a high distance may indicate the innovation event can happen at any time without obstructing technological progress. Documents in yellow are found in literature reviews by~\cite{RN1371,RN998,RN1738}; their presence validates our method. List of labelled documents are available at~\cite{figshare1}.}
		\label{fig:validation}    
\end{figure}

\begin{figure}
     \centering
     \begin{subfigure}[b]{0.45\textwidth}
         \centering
         \includegraphics[width=\textwidth]{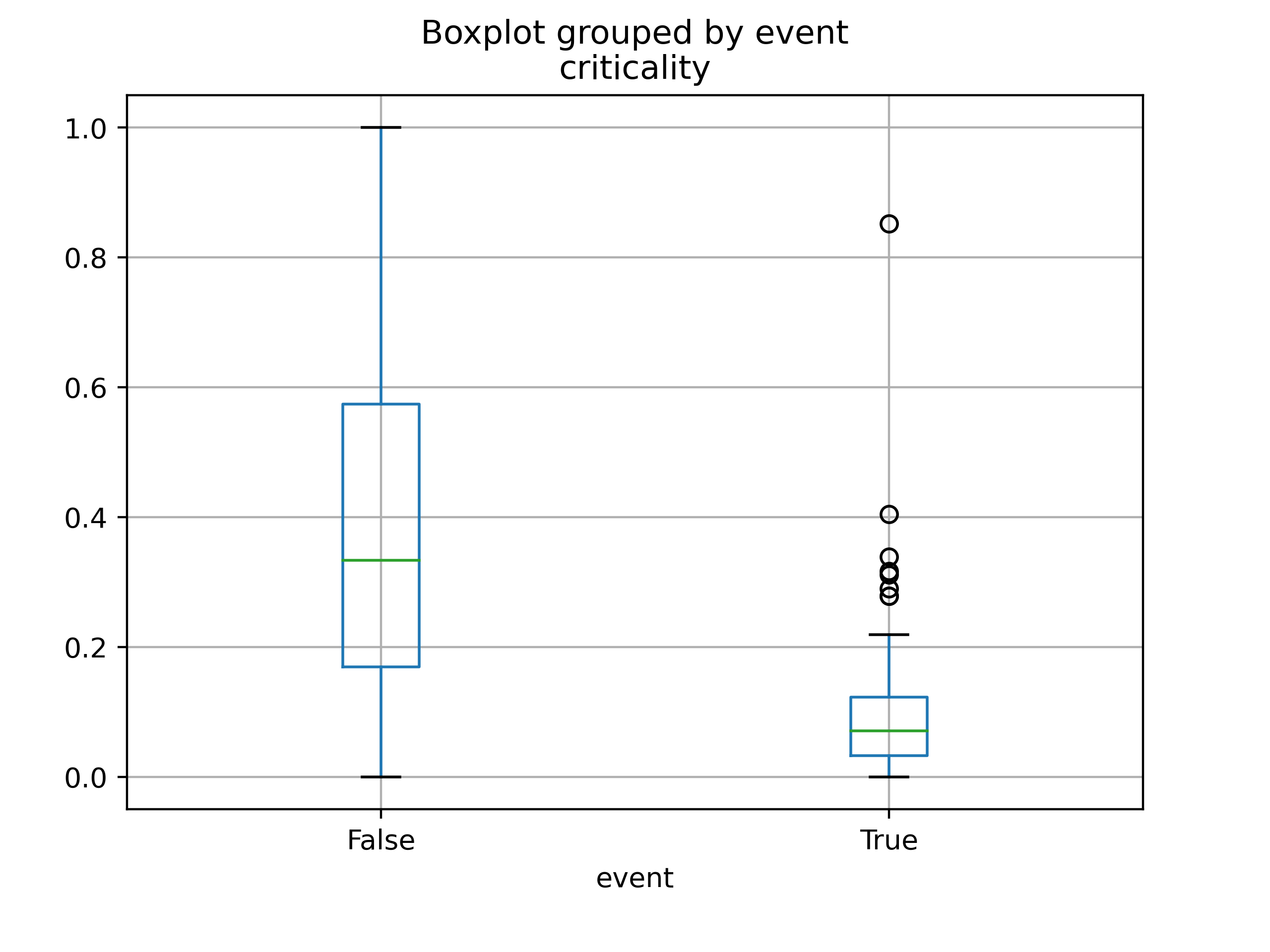}
         \caption{Moderna Spikevax}
     \end{subfigure}
     \hfill
     \begin{subfigure}[b]{0.45\textwidth}
         \centering
         \includegraphics[width=\textwidth]{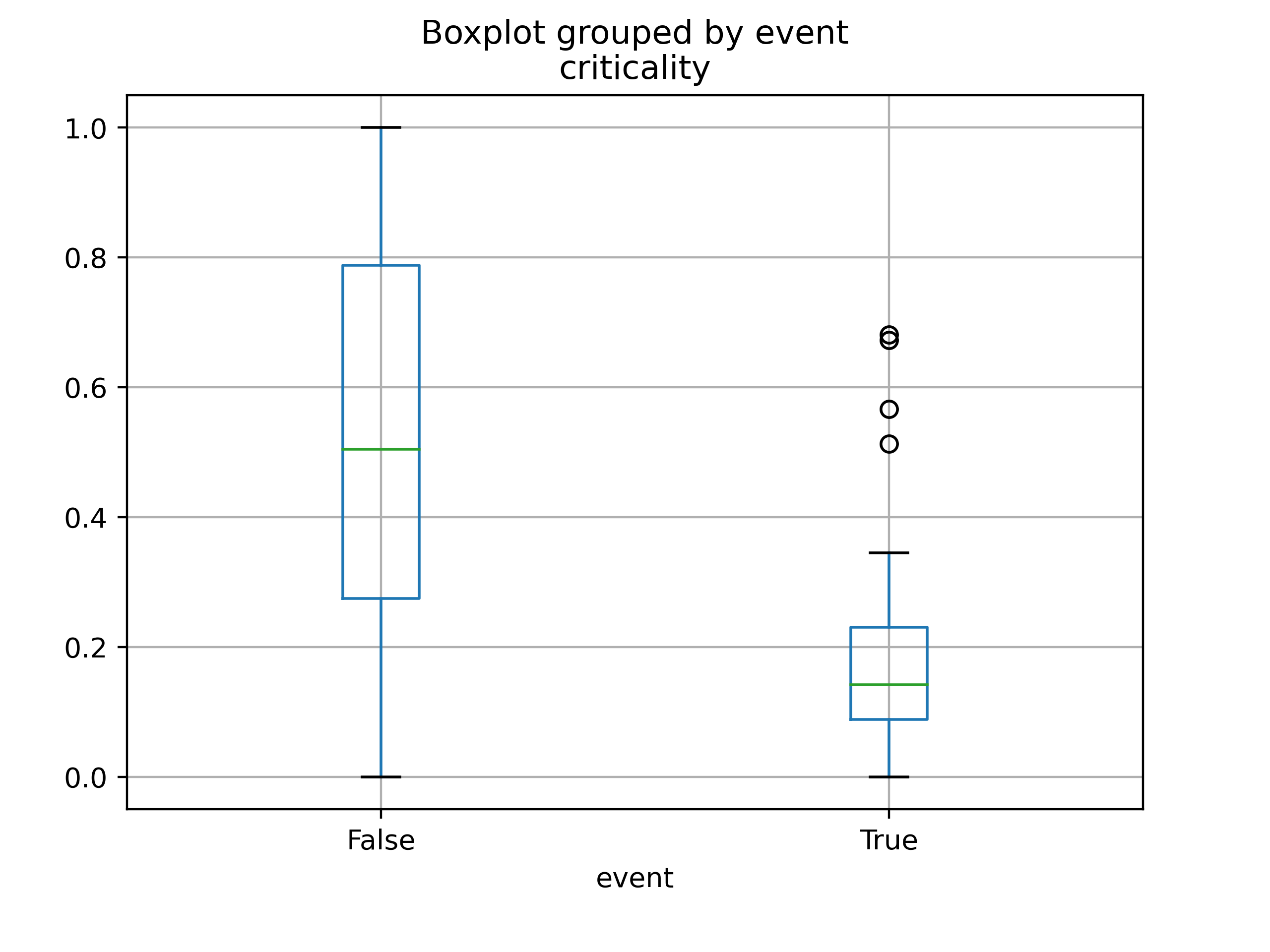}
         \caption{Pfizer/BioNTech Comirnaty}
     \end{subfigure}
     \hfill
 		\caption{\textbf{Criticality of documents cited by literature reviews.} Event: inclusion of network documents in literature review articles identified \emph{ex post}; true event: all documents within the network and identified in literature reviews; false events: all documents within the network and \emph{not} identified in literature reviews. The box extends from the Q1 to Q3 quartile values of the data, with a line at the median (Q2). The whiskers extend from the edges of box to show the range of the data. 
 		}
		\label{fig:boxplots}    
\end{figure}

\clearpage

\section{Discussion}\label{discussion}

One of the most useful measures in network science is the length of the \emph{shortest} path between two nodes as this is used in numerous situations as the distance between two nodes\footnote{The length of the shortest path between two nodes in a network satisfies all the criteria is what is formally defined as a `distance' function in mathematics. Indeed, it also satisfies the mathematical criteria to be a `metric' and the shortest paths are therefore `geodesics'.}. In many cases, these shortest paths are of practical relevance as we often look for the quickest, shortest route between two objects (nodes) in a network. For instance, in a social network, where people are nodes and edges are the connections between friends, the shortest path often represents the quickest way to get information between people. It is the basis for the popular idea of the six degrees of separation \cite{C21}. As a result, this shortest-path measure of distance between nodes is the basis for many other fundamental measures in network science, such as centrality measures \cite{C21}.

In most networks, the longest path between two nodes has little practical relevance. For instance, in a social network, the longest path between two people will usually be a path visiting almost everyone in the network once, and we can think of no practical use for the distance of such a path in that context\footnote{Paths which visit every node are sometimes of interest, e.g.\ in the classic travelling salesman problem.  However, in such cases the length of the paths is defined differently, in terms of the length of time to travel the network, the sums of the travel times associated with each edge. In such cases, there are many paths passing through all nodes, and the problem is still to find the shortest of these paths when distance is measured in terms of the sum of the time taken to travel each link in the path. So, this is still searching for the shortest paths out of a set of options, but using a different measure for distance from the one we are discussing at this point in the text.}.
However, the order encoded in a DAG means that the longest path between two nodes in such networks is not especially long; the length of the longest path in a DAG is rarely similar to the number of edges, as it is in the social network example above. So, now the question is, which path should we use when analysing the flow of information in our multilayer innovation citation networks, the shortest or the longest paths?

At a qualitative level, we can see that the longest path is likely to be more interesting for citation networks. The oldest papers we cite in our bibliography here are over sixty years old. In preparing this paper, it is likely that we did not learn much of direct relevance to the current paper by reading such papers. Such old papers are classics, but their influence is indirect, felt in our present work via more recent publications which apply these classic concepts in a modern context with modern terminology and notation. Conversely, it seems likely that the most recent papers give a much more powerful stimulus to authors.

Most papers in the bibliography of a journal article are recent, and the time difference between the citing document and the documents in one reference list decreases exponentially when the age difference is a couple of years or more, \cite{R04c,MRJU17,GS17,G19}. Further support comes from studies which show that the text for over 70\% of references in the bibliography of a journal article may have been copied from the publications of intermediate age suggesting that a large fraction of these older articles are not read when a new article \cite{SR03,SR05a,GAE14,CGLE14}. That is consistent with the idea that the ideas in these older texts were learnt from intermediary texts. In the case of innovation, the longest path embeds a chronological sequence of key technological advancements. The longest paths in a DAG embed the chronological and causal sequence of key scientific and technological advances contributing to a technological outcome.

The use of longest paths, not shortest paths, to analyse DAGs is common in other areas. The critical path method is used to schedule jobs in a project, such as the Manhattan Project~\cite{KW59,RN1655,RN1654} or independent parts of a numerical simulation running on multiple processors. In this method, the DAG captures the dependency (the edges) of one job (one node) on another. An innovation network path length, considered as a scheduling DAG, is the sum of the time needed to complete each job on the path (so not simply the number of edges in the path). The aim is to find the ``critical path'' which is the path that sets the least time needed to complete the project. The critical path is set by the longest path in the scheduling DAG\footnote{The critical path is of greater relevance to mission-oriented innovation programs than diffusion-oriented ones~\cite{RN960}. Innovation missions involve organizing multiple innovation projects and programs, all with intermediate outputs, to attain specific goals and would benefit from finding the minimum time to attain an innovation goal.}.

Badiru~\cite{RN1654} defines three important aspects of the critical path method that we test using larger and more complex innovation datasets:
\begin{enumerate}
	\item An activity is considered \emph{critical} if changing the start or finish time of the activity will affect the overall project schedule
	\item The series of critical activities connecting the start and end points of a project is known as the \emph{critical path}. Logically, the critical path ``turns out to be the \emph{longest path} in [a] network''.
	\item A delay in any critical activity delays the entire project. Therefore, the ``sum of durations for critical activities represents the \emph{shortest possible time} to complete the project''
\end{enumerate}
This means:
\begin{center}
	\emph{Critical schedule path $\equiv$  longest network path $\equiv$ shortest time path}
\end{center}

There is also a more formal basis for the use of longest paths in DAGs. The order in DAGs is often derived from the flow of time, so it makes most sense to look at embedding a DAG in space-time not simply in a space. Technically, DAGs are best embedded in a Lorentzian space-time, such as Minkowski space used in special relativity, rather than a Riemannian space, such as Euclidean space used for school-level geometry and in most traditional data science methods. For simple models, it is possible to show that the longest path in the DAG is the best approximation to the geodesic in the space-time\footnote{Analytical results are derived in Bollob\'{a}s and Brightwell \cite{BB91} and earlier papers cited therein. These results are applied in Brightwell and Gregory \cite{BG91}. Numerical results are in Rideout and Wallden \cite{RW09}.}
\cite{BB91,BG91,RW09}.
Geodesics represent paths of ``least resistance'', the path a freely moving particle would follow in the space. So, by analogy, it makes sense to think of longest paths in our innovation citation networks as representing the easiest route for knowledge to flow between two documents (nodes) in our citation network. This analogy has been successfully tested in the context of citation networks to show that geometric concepts like the dimension of a DAG can be defined using Minkowski space-time \cite{CE16}  or how to embed a DAG in Minkowski space-time \cite{CE17}.

Unsurprisingly, the study of innovation through citation networks is not a new subject.
The best known approach is ``main path analysis'' proposed by Hummon and Doreian~\cite{RN997} and, with variations, implemented in some popular analysis packages \cite{RN1245} (see \secref{main_path_analysis} for details).
Our critique of the main path analysis method is that there is no formal basis for the method, unlike the work on the relationship between the longest path and geodesics in space-time models. The weights used in main path analysis are formed by looking at all paths from a set of initial nodes (such as the first publications in the data set, sometimes all publications) to a set of dfinal destination nodes, each path equally weighted. However, we know some publications listed in a bibliography are more important than others, so it is unclear why giving all paths are equal weight is a good way to capture the flow of innovation. While popularity of main path analysis is one measure of a successful method, this popularity could be due to other factors such as the easy access to numerical implementations in widely used packages such as \texttt{pajek} \cite{RN1245}.
An alternative view of main path analysis can be found when it fails to identify the backbone of the technological trajectory of the semiconductor industry~\cite{RN1092}. Our interpretation is that the method may fail due to its reliance on a single path rather than the inability to look at good paths for innovation.

\section{Conclusion}\label{conclusion}

When studying innovation, using a citation network DAG as opposed to other econometric approaches allows us to causally trace every intermediate step between a complete set of innovation inputs and an innovation outcome. Rather than regressing a limited set of variables, a citation network is a DAG, so this encodes the geometry of innovation order. Theoretically, the longest paths in a DAG represent the critical causal routes which show the bottlenecks that constrain innovation; it is the most complex route to achieve because it is composed of the largest amount of linear components which cannot be parallelised. We hypothesise that the longest paths in these multiplayer citation networks where order matters are the critical paths of innovation.

To verify the usefulness of the longest path in describing critical innovations, we prototyped two methods: (i) a reproducible way to construct multilayer citation network thus representing basic research (publications), applied research (patents), development (clinical trials), and commercialisation (regulatory authorisations), and (ii) a simple way to quantify a document's closeness to the longest path. These methods allow us to analyse events that turn out to be in the longest paths of eight vaccine citation networks. We were able to observe how basic discoveries in the lab accumulated and got absorbed by clinical researchers, who used these phenomenological observations to hypothesise what could work for a vaccine. Once a prototype vaccine product was available, the technological community further applied basic discoveries to optimise the product, which was eventually validated through clinical trials and approved for marketing.

As seen in \figref{fig:multilayer_framework}, we can define the types of edge, in terms of the labels of the nodes at the end of the edge. This method of defining edges is useful in understanding innovation ``translation'' will be explored in a separate study. The proposed method to quantify criticality of innovation events empowers scientific understandings of technological change, particularly in: (1) comparing criticality patterns across industries and time spans, (2) comparing criticality patterns of funders, (3) measuring linearlity of innovation phases, (4) informing mission-oriented innovation planning, (5) attributing technological outcomes to events and entities and (6) forecasting new innovation outcomes based on intermediate outputs.

By assembling a list of innovation events in the observed technological past, we inform the ingredients needed in similar future technological programs. Using this innovation ruler, we demonstrated the possibility to measure the rate of innovation, division of innovation labour, and criticality of innovation funders.

\bibliography{orderInInnovation}

\clearpage




\begin{center}
\Large\textbf{Appendices}
\end{center}
\appendix
\renewcommand{\thesection}{\Alph{section}}
\renewcommand{\theequation}{\thesection\arabic{equation}}
\renewcommand{\thefigure}{\thesection\arabic{figure}}
\renewcommand{\thetable}{\thesection\arabic{table}}
\numberwithin{equation}{section}
\numberwithin{figure}{section}
\numberwithin{table}{section}
\setcounter{section}{0}
\renewcommand{\theHsection}{\Alph{section}}
\renewcommand{\theHequation}{\thesection\arabic{equation}}
\renewcommand{\theHfigure}{\thesection\arabic{figure}}
\renewcommand{\theHtable}{\thesection\arabic{table}}

\section{Formal Definitions}\label{a:defn}

\noindent Here we give the definitions of the network properties used in the main text using a formal language. The main aim is to arrive at a formal proof that nodes with zero criticality lie on a longest path in the DAG, \lemmaref{l:zerociticality}.

\begin{mydefs}{Network/Graph, Nodes, Edges}

A \tsedef{network}  (here synonymous with \tsedef{graph} ) $\Gcal = (\Vcal,\Ecal)$ is a set of \tsedef{nodes} (also known as vertices) $\Vcal$ and a set of \tsedef{edges} $\Ecal \subseteq \Vcal \times \Vcal$. An edge is denoted $(u,v)$ where $u,v \in \Vcal$.
\end{mydefs}

\begin{mydefs}{Layers}

A \tsedef{multilayer network} is a graph where nodes are connected by different types of edges (for example, see section 4.2 in \cite{C21}). In our case, the \tsedef{layers} are defined by a partition of the nodes into different types $\Vcal_\alpha$ so each layer contains nodes of only one type and each node exists on just one layer (a particular type of multilayer network). In our setting, we have an unweighted graph with four main types (Publication, Patent, Clinical trials, Regulatory approval) of nodes. The label of the node at each end of an edge, say $\alpha$ and $\beta$, leads to a partition of the edges into different types $\Ecal_{\alpha\beta}$ (the definition of a multilayer network), that is $\Ecal_{\alpha\beta} \subseteq \Vcal_{\alpha} \times \Vcal_{\beta}$.

\end{mydefs}

\begin{mydefs}{Directed Graphs and Edges}
	
A \tsedef{directed network} has \tsedef{directed edges} where the edge $(u,v)$ is distinct from the edge $(v,u)$.
\end{mydefs}

\begin{mydefs}{Predecessors and In-Degree}
	
	The \tsedef{predecessors} $N^-(v)$ of a node $v$ is the set of nodes which are connected by an edge \emph{to} $v$, so $\Ncal^-(v) = \{u | (u,v)\in \Ecal\}$
	
	The \tsedef{in-degree} $\kin_v$ of a node $v$ is the number of incoming edges, $\kin_v=|\Ncal^-(v)|$.
\end{mydefs}

\begin{mydefs}{Successors and Out-Degree}
	
	The \tsedef{successors} $N^+(v)$ of a node $v$ is the set of nodes which are connected by an edge \emph{from} $v$, so $\Ncal^+(v) = \{w | (v,w)\in \Ecal\}$.

	The \tsedef{out-degree} $\kout_v$ of a node $v$ is the number of outgoing edges, so $\kout_v=|\Ncal^+(v)|$.
\end{mydefs}

\begin{mydefs}{Walk, Path, Cycles and Path Length}

A \tsedef{walk} from node $u$ to node $v$, denoted $\Pcal(u,v)$, is a sequence of nodes starting at $u$ and finishing with $v$ which are connected sequentially by edges
\begin{equation}
 \Pcal(u,v) = \{ w_n | n=0,\ldots, \ell, w_n \in \Vcal, (w_n,w_{n+1}) \in \Ecal \mbox{ for } 0\leq n<\ell , w_0=u, w_n=v \} \, .
 \label{e:walkdef}
\end{equation}

The \tsedef{length of a walk} is the number of nodes minus one $\ell=|\Pcal|-1$.

A \tsedef{path} is a walk where all the nodes are distinct, so $w_n =w_m$ iff $m=n$ in \eqref{e:walkdef}.

A \tsedef{cycle} is a walk where the first and last node in are identical, $w_0=w_\ell$ in \eqref{e:walkdef}, but all other nodes are distinct.

The \tsedef{concatenation of two walks} is where a walk $\Pcal(u,v)$ from $u$ to $v$ is combined with a walk $\Pcal(v,w)$ from $v$ to $w$ to produce a walk $\Pcal(u,w)$ from $u$ to $w$ via $v$. We simply extend the sequence of nodes in the first walk $\Pcal(u,v)$ with the nodes in the second walk $\Pcal(v,w)$ keeping the nodes in the same order and we only include the common end/start node $v$ once. We denote this as $\Pcal(u,w)=\Pcal(u,v) \cdot \Pcal(v,w)$ and formally we may define this as
\begin{eqnarray}
 \Pcal(u,v) \cdot \Pcal(v,w)
  &=&
  \{ x_n |  x_n       = w_n \in \Pcal(u,v), \;  n=0, \ldots, L_1;
  \nonumber \\
  && \qquad x_{n+L_1} = w_n \in \Pcal(v,w), \;  n=0, \ldots, L_2 \} \, ,
  \nonumber \\
  &&
  \quad L_1=|\Pcal(u,v) |-1 \, ,
  \;\;\;
  L_2=|\Pcal(v,w) |-1       \, .
 \label{e:walkconcatdef}
\end{eqnarray}
\end{mydefs}

\begin{mydefs}{Directed Acyclic Graph --- DAG, Sources and Sinks}

	A \tsedef{Directed Acyclic Graph}, a \tsedef{DAG}, $\Dcal = (\Vcal,\Ecal)$ is a network/graph with directed edges containing no cycles.
	
	A node $s$ with no incoming edges, that is $\nexists \, u \in \Vcal \; \text{s.t.} \; (u,s) \in \Ecal$, is known as a \tsedef{source node}

	A node $t$ with no outgoing edges, that is $\nexists \, w \in \Vcal \; \text{s.t.} \; (t,w) \in \Ecal$, is known as a \tsedef{sink node}
\end{mydefs}

\begin{mydef}{The Partial order of a DAG}

A DAG always defines a unique \tsedef{partial order} on the set of nodes $\Vcal$ in which we have a binary relation between two nodes, denoted $u \prec v$, if and only if there is a path from $u$ to $v$.
\begin{equation}
  u \prec v \; \Leftrightarrow \; \exists \; \Pcal(u,v) \; \; \forall \; u, v \in \Vcal
\end{equation}
\end{mydef}
\noindent Note that under our definition of a path, every node is part of a trivial path of length zero $\Pcal(v,v) = \{ v \}$ so this  binary relation is reflexive, i.e.\ $v \prec v$, as required. Transitivity of a partial order comes from the fact that a concatenation of paths produces a path.

\begin{mydef}{Distance between nodes}

The distance from node $u$ to node $v$ in a DAG is defined to be the length of the longest path from $u$ to $v$, while the distance is left undefined if there is no such path.
\beq
 d(u,v) = \max (|\Pcal(u,v)|-1 | u \prec v ) \,.
 \label{e:lpdist}
\eeq
\end{mydef}
\noindent This definition of a `distance' in \eqref{e:lpdist} is not sufficient for this function $d$ to be a distance in the formal sense used in mathematics because (a) many pairs nodes in a DAG may not be connected and this distance is undefined for such pairs and (b) this definition is not symmetric since if $d(u,v)$ is defined then $d(v,u)$ is not defined unless $u=v$. Note that both of these issues are easy to fix if a distance in the formal mathematical sense is required\footnote{For instance, consider $d'(u,v)$ where $d'(u,v)=d(u,v)$ of \eqref{e:lpdist} when $u \prec v$. We then define $d'(v,u)=d(u,v)$ when $u \prec v$. Finally we set $d(u,v)=\infty$ whenever $u \not\prec v$ and $v \not\prec u$. This function $d'$ is a formal mathematical distance function.}.

Note that there are often many paths between two nodes $u$ and $v$ with the same length. This includes the longest paths, those with length equal to $d(u,v)$.

Also note that many other distances functions can be defined for pairs of nodes on a DAG such as the length of the shortest path between two nodes. We will only use the distance defined here in terms of the length of longest path.

\begin{mydef}{Reverse Triangle Identity}

The distance between two nodes satisfies the \tsedef{reverse triangle identity}
\beq
 d(u,w) \geq  d(u,v) + d(v,w) \quad \text{if} \;\; u \prec v, v \prec w \, .
 \label{e:revtriid}
\eeq
\end{mydef}
\noindent The reverse triangle identity has the opposite inequality from that found in the usual triangle inequality. The other difference is this identity does not apply for other permutations of the three sites $u$, $v$ and $w$.

Transitivity of the partial order guarantees that $u \prec w$. This means that we can consider one of the longest paths from $u$ to $v$, say $\Pcal_L(u,v)$, whose length gives the value to $d(u,v)$. Likewise, we know we have a path $\Pcal_L(v,w)$ which is a longest path from $v$ to $w$ of length $d(v,w)$. If we concatenate these two paths, we produce a path $\Pcal_L(u,v)\cdot \Pcal_L(v,w)$ from $u$ to $w$ which is of length $d(u,v) + d(v,w)$. The concatenated path is a path from $u$ to $w$, but it need not be a longest path between these two, even though it is made by combining two longest paths. Since the distance \eqref{e:lpdist} is set by the largest path length, it means the length of the concatenated path sets a lower bound on the distance from $u$ to $w$. By definition, if there is another path from $u$ to $w$ and if such a path is longer than the concatenated path $\Pcal_L(u,v)\cdot \Pcal_L(v,w)$, then this alternative path will set the distance $d(u,w)$ and that will be larger than $d(u,v) + d(v,w)$. Hence, the simple properties of paths and the maximum function in \eqref{e:lpdist} give us our reverse triangle identity.

\begin{mydefs}{Height, Depth and Criticality}

The \tsedef{height} $h(v)$ of a node $v$ in a DAG $\Dcal$ is the length of the longest path to that node from \emph{any} node. 
\begin{equation}
 h(v) = \max (|\Pcal(u,v)|-1 | u \prec v, u \in \Vcal) \,.
\end{equation}

The \tsedef{depth} $d(v)$ of a node $v$ in a DAG $\Dcal$ is the length of the longest path from that node $v$ to \emph{any} node.
\begin{equation}
 d(v) = \max (|\Pcal(v,w)|-1 | v \prec w, w \in \Vcal) \,.
\end{equation}

The \tsedef{height} $\hmax$ of a DAG $\Dcal$ is the largest height of any node
\begin{equation}
 \hmax = \max (h(v) | v \in \Vcal ) \,.
 \label{e:hmaxdef}
\end{equation}

The \tsedef{criticality} $c(v)$ of a node $v$ in a DAG $\Dcal$ is the height of the DAG minus the height and minus the depth of that node
\begin{equation}
 c(v) = \hmax - h(v) - d(v) \,.
 \label{e:critdef}
\end{equation}

\end{mydefs}

\noindent The terminology `height' and `depth' are common when working with DAGs but `criticality' is our own terminology for $c(v)$ in \eqref{e:critdef}. 

If the distance function $d$ satisfies the properties of a formal mathematical distance then the height and depth are always defined and they will have a value of at least zero, $h(v),d(v) \geq 0$. The only nodes with height zero are source nodes and the only nodes with depth zero are sink nodes.

With height and depth, we are effectively defining a distance between a node $v$ and the `beginning', some source node $s$, and the `end' of our DAG, some sink node $t$. We can formalise this in the following lemmas, which lead to our key result on the bounds for criticality of a node and on the interpretation of zero criticality value for nodes in a DAG.

\begin{lemma}[Path leading to the height of a node] \label{l:sourceheight}

The height of a node $v$ is always the length of a longest path from some source node $s$ to $v$.

\end{lemma}

\begin{proof}
Suppose this were not true and the height of $v$ is based on a path from some node $u$ to $v$. If $u$ is not a source node then there must be a node preceding $u$ connected by an edge $(u',u) \in \Ecal$.
By definition $d(u',u)=1$ and so by the reverse triangle identity \eqref{e:revtriid} we know that $d(u',v) \geq 1+ d(u,v)$ so $d(u',v) > d(u,v)$.
Thus, the node $u$ does give the longest path to $v$ and this node $u$ does not define the height of node $v$. We have a contradiction and so deduce that the height of the path must use a node with no predecessors, i.e.\ a source node.
\end{proof}

Using the same type of argument used to prove \lemmaref{l:sourceheight} we can quickly show the following lemma.
\begin{lemma}[Path leading to the depth of a node] \label{l:sinkdepth}

The depth $d(v)$ of a node $v$ comes from the longest path from $v$ to a sink node.

\end{lemma}

\noindent We show this using the same arguments used in  \lemmaref{l:sourceheight} but now applied to paths from node $v$ to some node $w$. If $w$ has a successor, $w'$, i.e.\ $(w,w')$ is an edge, then node $w$ is not involved in defining the depth.

\begin{lemma}[Path leading to the height of a node]

The height of a DAG, $\hmax$, comes from the length of a longest path from a source to a sink node.

\end{lemma}

\begin{proof}
	We use the same arguments as we did for the last two lemmas. The only paths that can not be extended at either end to produce longer paths are those running from a source node to a target node.  Thus the nodes with the largest heights are the sink nodes. The height of the DAG comes from the largest of all heights \eqref{e:hmaxdef}, so this must run to a source node and, by \lemmaref{l:sourceheight}, run from a sink node.
\end{proof}

\noindent A couple of corollaries follow from this. First that the height of the DAG is the length of a longest path anywhere in the graph. Second that the height of the graph is also equal to the largest value of the depth which has to be for the depth of one (or more) of the source vertices. 

Finally, we can put these ideas together to show the following lemma that is the basis for our analysis.
\begin{lemma}[Zero criticality nodes] \label{l:zerociticality}

A node $v$ with zero criticality, $c(v)=0$, lies on a longest path in the DAG. Nodes $v$ with positive criticality $c(v)>0$ do not lie on a longest path of the DAG.

\end{lemma}

\begin{proof}
Consider a node $v$ with height $h(v)$ based on a longest path $P_1$ from source $s$ to $v$ (from \lemmaref{l:sourceheight}) and depth obtained from some longest path $P_2$ from $v$ to sink node $t$ (from \lemmaref{l:sinkdepth}).

The path $P_1 \cdot P_2$  from source $s$ via $v$ to sink node $t$ obtained by concatenating paths $P_1$  and $P_2$  has length $h(v)+d(v)$ by definition of concatenated paths in \eqref{e:walkconcatdef}.  This concatenated path $P_1 \cdot P_2$ need not be a longest path from $s$ to $t$, which has length $d(s,t)$, but in that case, $P_1 \cdot P_2$ must be shorter than this longest path so $d(s,t)>h(v)+d(v)$. Equality, $d(s,t)=h(v)+d(v)$,  can therefore only happen if the concatenated path $P_1 \cdot P_2$ is also a longest path from $s$ to $t$. 

If this longest path from source $s$ to source $t$ is one of the longest paths in the DAG, then $d(s,t)$ is the height $\hmax$ of the DAG. In this case we have $\hmax = d(s,t)=h(v)+d(v)$ and so $c(v)=0$ which is the first  of the lemma.  

The converse of this statement, the second part of the lemma, follows from the following cases. Case (i) is where the concatenated path $P_1 \cdot P_2$ is a longest path from source $s$ to source $t$ but is \emph{not} one of the longest paths in the DAG so $\hmax > d(s,t) \geq h(v)+d(v)$. Case (ii) is where the concatenated path $P_1 \cdot P_2$ is not a longest path from source $s$ to source $t$ so $d(s,t)>h(v)+d(v)$. In this second case, we know $\hmax \geq d(s,t)$ regardless of the nature of the concatenated path so again $\hmax > d(s,t) = h(v)+d(v)$. In either case, $v$ is not on a longest path in the DAG and $c(t)>0$ proving the second part of the lemma.
\end{proof}

As a corollary of this lemma, we can see that the \tsedef{height} of the DAG $h(\Dcal)$ is also the largest possible value of the depth of any node.

\newpage
\section{Empirical data} \label{a:datasource}

In our context, the sources of our DAGs are the nodes representing the FDA approval on one of eight vaccines. We have chosen to work with vaccines produced from one of four different methods, and we outline these different approaches and associated vaccines in this section.

\subsection{Viral vector (adenovirus vector) vaccine platform}

\textbf{Technical principle and bottlenecks.} A \tsedef{viral vector vaccine} (VVV) is a relatively novel vaccine platform that uses virus to infect host cells with the genes of pathogens; the infected host cells then transcribe and translate the genes into antigens of the pathogen. Following vaccination, T cells and B cells respond against both to the viral vector itself and, more importantly, the antigen encoded viral vector. Viral vector vaccines can rapidly adapt from one pathogen to another because only the gene of interest needs to be exchanged. An \tsedef{adenoviral vector vaccine} (AVV) is a subtype of the viral vector vaccines that exploits the high transduction efficiency and pervasive tropism of adenoviruses to facilitate the expression of target antigen. An adenoviral vector vaccine is produced by deleting the replication genes from an adenovirus serotype and inserting the genetic sequence of interest to the virus. This is followed by viral vector production in manufacturing cells and purification~\cite{RN1379}. Historically, the efficacy of a viral vector vaccine has be challenged by hosts' immunity against the viral vector~\cite{RN1642} and the surveyed vaccines should exhibit mechanism to circumvent this issue: finding rare or non-human adenovirus serotypes and vectorising adenoviruses from non-human primates~\cite{RN1644,RN1643}.

\textbf{Data.} As of October 2022, four viral vector vaccines have been cleared by the FDA (the US Food \& Drugs Administration) for use in humans; we consider two of them which represent the first uses of adenovirus as vaccine vector: Zabdeno\footnote{Mvabea, the second dose of the Zabdeno/Mvabea regiment uses modified vaccinia Ankara (a poxvirus) as vector.} (against Ebola, developed by Janssen, first authorised for use in 2020)~\cite{RN1721} and Vaxzevria (COVID-19, AstraZeneca, 2020)~\cite{RN1722}.

\subsection{Nucleic acid (mRNA) vaccine platform}

\textbf{Technical principle and bottlenecks.} Both DNA and RNA can elicit immune response, but to date, only \tsedef{mRNA} vaccines have been authorised by the FDA. Similar to adenoviral vector vaccines, nucleic vaccines are highly immunogenic, easily adopted, and readily manufactured compared to inactivated/attenuated vaccines. The two currently commercially available nucleic vaccines work by expressing antigens of interest via nucleoside-modified mRNA encapsulated in lipid nanoparticles (LNP)~\cite{RN1370}. To arrive at the mRNA vaccines we have, innovators had to understand how mRNA elicits an immune response, how to control the amount of innate inflammatory reactions to therapeutic mRNA, how to deliver mRNA to transfect cells; the synthesis, purification, and buffering of mRNA~\cite{RN1371,RN1610,RN998,RN1370}.

\textbf{Data.} Similar to adenoviral vector vaccines, only two mRNA vaccines are authorised by the FDA at the time of writing: Spikevax (COVID-19, Moderna, 2020)~\cite{RN1723} and Comirnaty (COVID-19, BioNTech/Pfizer, 2020)~\cite{RN1724}.

\subsection{Whole pathogen (attenuated) vaccine platform}

\textbf{Technical principle and bottlenecks.} The very first vaccine was an example of a \tsedef{whole pathogen vaccine} (WPV) and the word ``vaccine'', from the Latin ``vaccinus'', comes from Jenner's use of cow (``vacca'') pox to prevent smallpox. A typical whole pathogen vaccine contain microbes that are live but attenuated, meaning they are weakened strains, or they are inactivated, meaning they are killed or altered to prevent replication. For instance, the vaccines behind the eradication of polio contain inactivated poliovirus. Although whole pathogen vaccines are a two-century-old innovation, pathogenic attenuation or inactivation does not readily guarantee a viable vaccine due to immunogenicity, safety, and yield issues~\cite{RN1639}. This has led to newer vaccines to experiment with novel techniques such as the use of hydrogen peroxide and gamma irradiation as alternative means of inactivation~\cite{RN1641,RN1640}.

\textbf{Data.} Since whole pathogen vaccines are based on the oldest approach to vaccination, we use two examples as a baseline for comparisons with novel mRNA and viral vector vaccine platforms. We choose two whole pathogen vaccines that were recently cleared by the FDA: Dengvaxia (Dengue, Sanofi, 2019)~\cite{RN1725} which is one of the first vaccines against dengue, and Imvanex (Smallpox, Bavarian Nordic, 2013; aka. Jynneos)~\cite{RN1726} which is a third-generation smallpox vaccine that is being used to control monkeypox outbreak. Interestingly,  the modified vaccinia Ankara used by Imvanex is the same virus that serves as a viral vector vaccine vector for the Mvabea vaccine discussed above.

\subsection{Subunits (recombinant protein) vaccine platform}

\textbf{Technical principle and bottlenecks.} Compared to a whole bacterium or virion, \tsedef{subunit vaccines} contain one or more isolated constituents of a microorganism to stimulate a more targeted immune response. Subtypes of subunit vaccines make use of protein subunits isolated and modified whole proteins or partial peptides from pathogens; toxoid inactivated pathogenic poisons; polysaccharides or glycoproteins to mimic glycoproteins on cell surface of pathogens; or chemical conjugation of low-affinity polysaccharides with a high affinity protein carriers to improve B-cell recognition of the polysaccharides. Usually, subunit vaccines elicit lower immunogenicity than WPV. Strategies to improve subunit vaccine effectiveness include the use of adjuvants, multiple dosage regiments, condon optimisation to improve yield, and amino acid substitution stabilise the introduced peptide chains~\cite{RN1650}.

\textbf{Data}. We use recombinant protein subunit vaccine, which has been widely available since the 1980s~\cite{RN1638}, as another baseline to compare with the novel mRNA and adenoviral vector vaccine platforms, and for similarity with the even more, established WPVs. Following the same logic, we choose two recently FDA-cleared subunit vaccines that both contain recombinant protein subunits and use adjuvants to enhance immune response: Nuvaxovid (COVID-19, Novavax, 2022)~\cite{RN1727} and Shringrix (Shingles, GSK, 2017)~\cite{RN1728}.

\newpage
\section{All figures} \label{all_figures}

\begin{figure}[hbt]
     \centering
     \begin{subfigure}[b]{0.24\textwidth}
         \centering
         \includegraphics[width=\textwidth]{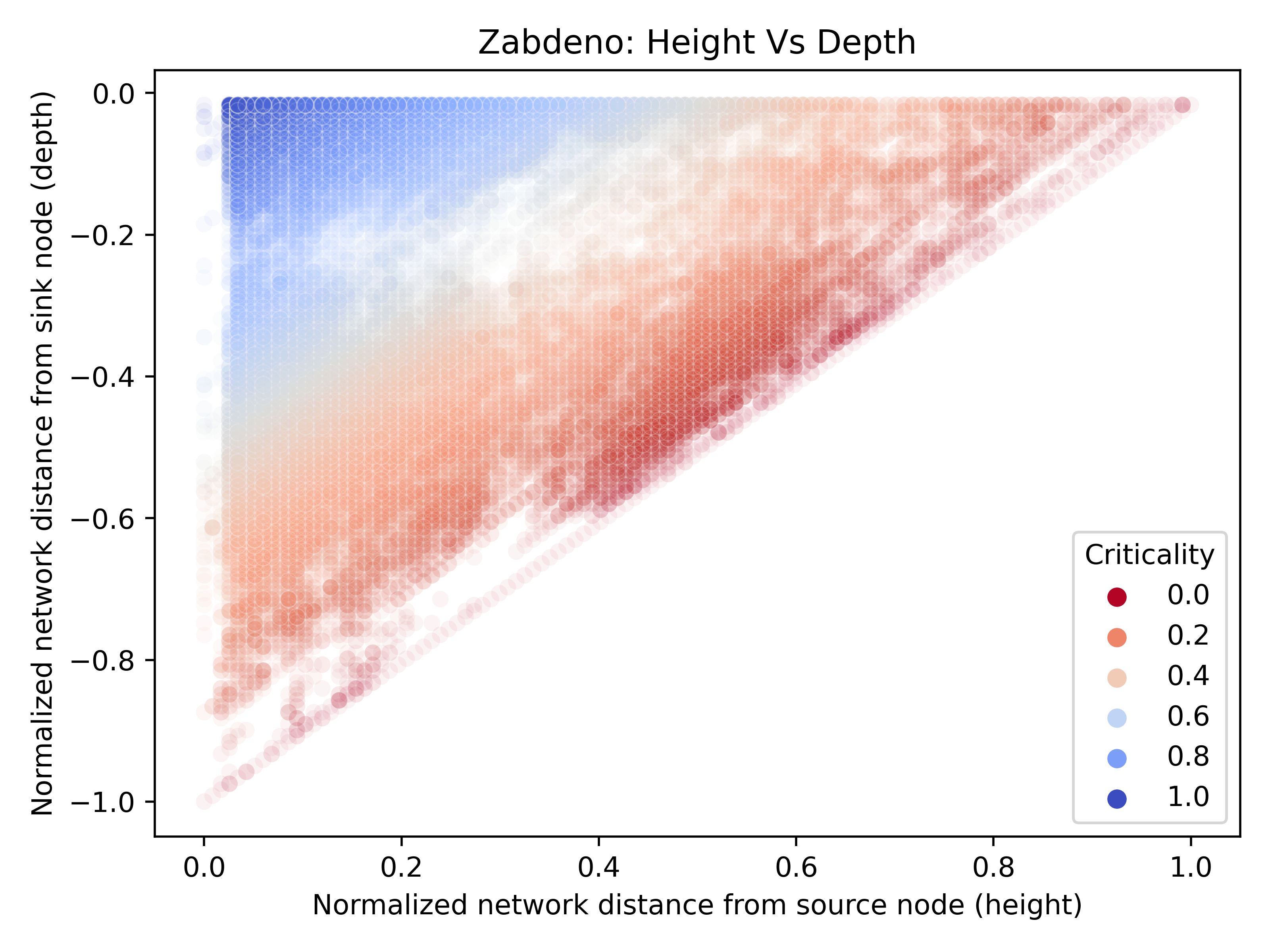}
         \caption{Zabdeno, Ebola, Janssen, 2020,\\AVV}
     \end{subfigure}
     \hfill
     \begin{subfigure}[b]{0.24\textwidth}
         \centering
         \includegraphics[width=\textwidth]{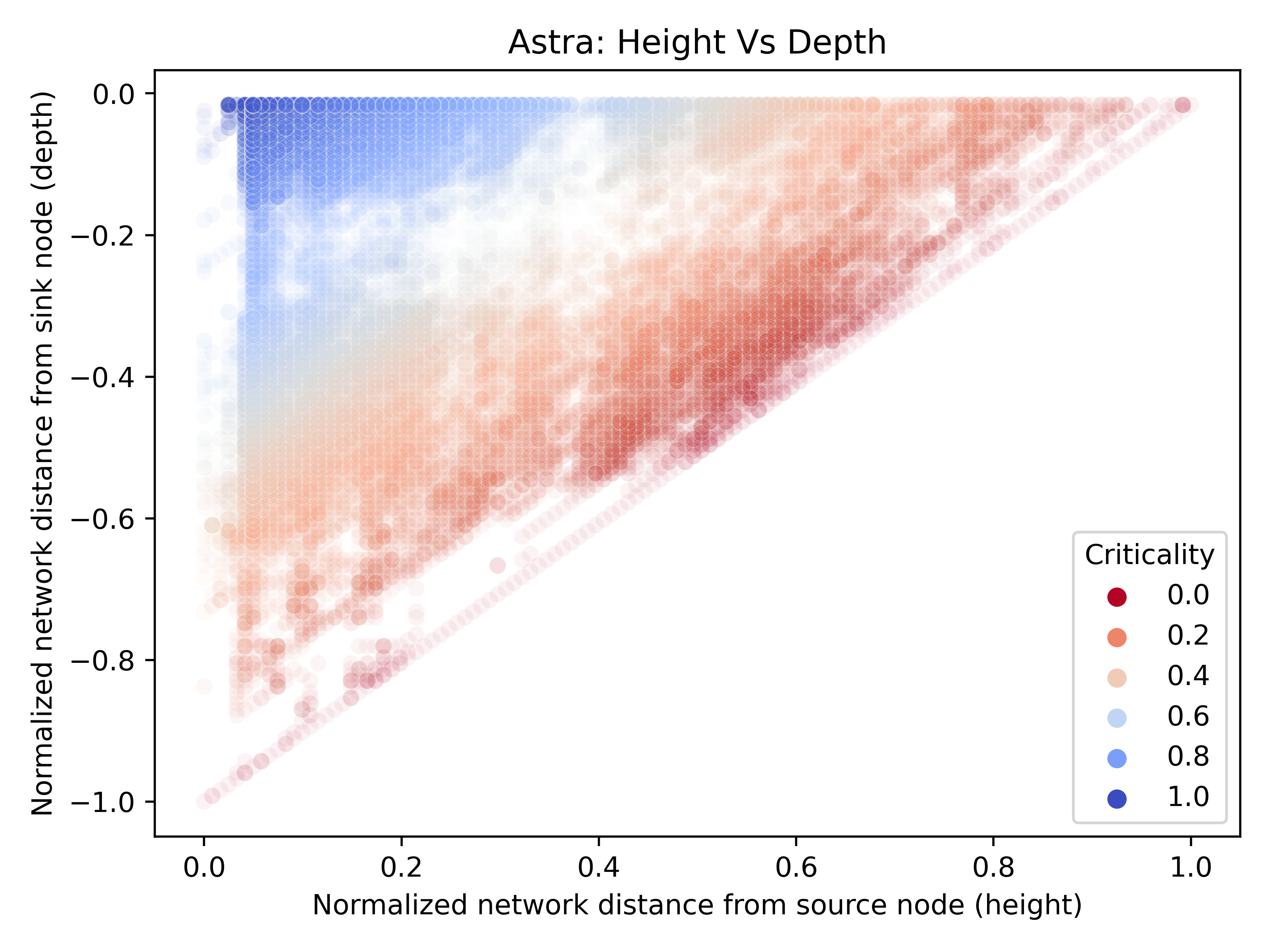}
         \caption{Vaxzevria, COVID-19, AstraZeneca, 2020, AVV}
     \end{subfigure}
     \hfill
     \begin{subfigure}[b]{0.24\textwidth}
         \centering
         \includegraphics[width=\textwidth]{hey_moderna.png}
         \caption{Spikevax, COVID-19, Moderna, 2020, mRNA}
     \end{subfigure}
     \hfill
     \begin{subfigure}[b]{0.24\textwidth}
         \centering
         \includegraphics[width=\textwidth]{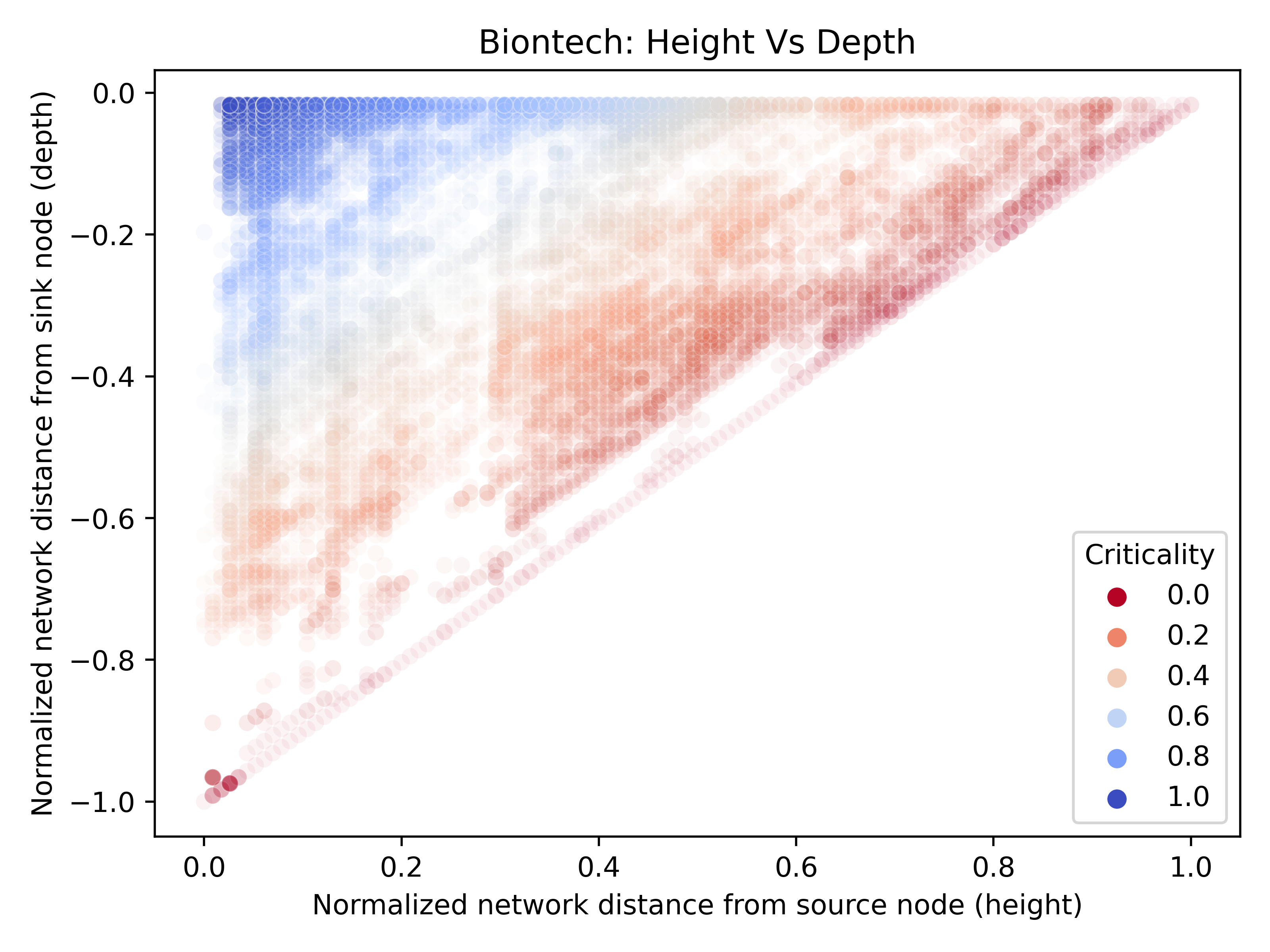}
         \caption{Comirnaty, COVID-19, BioNTech/Pfizer, 2020, mRNA}
     \end{subfigure}
     \hfill
     \begin{subfigure}[b]{0.24\textwidth}
         \centering
         \includegraphics[width=\textwidth]{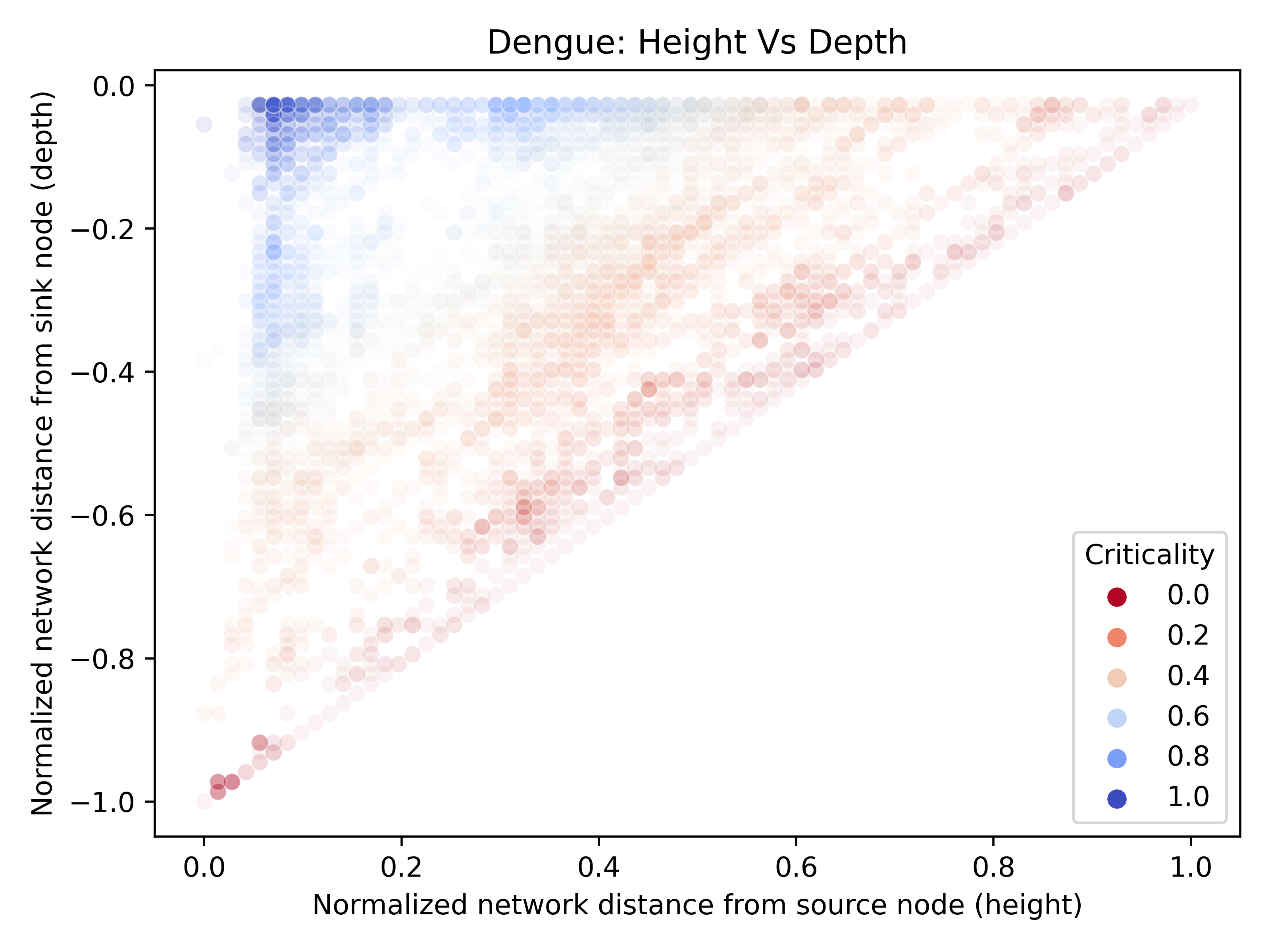}
         \caption{Dengvaxia, Dengue, Sanofi, 2019,\\ WPV}
     \end{subfigure}
     \hfill
     \begin{subfigure}[b]{0.24\textwidth}
         \centering
         \includegraphics[width=\textwidth]{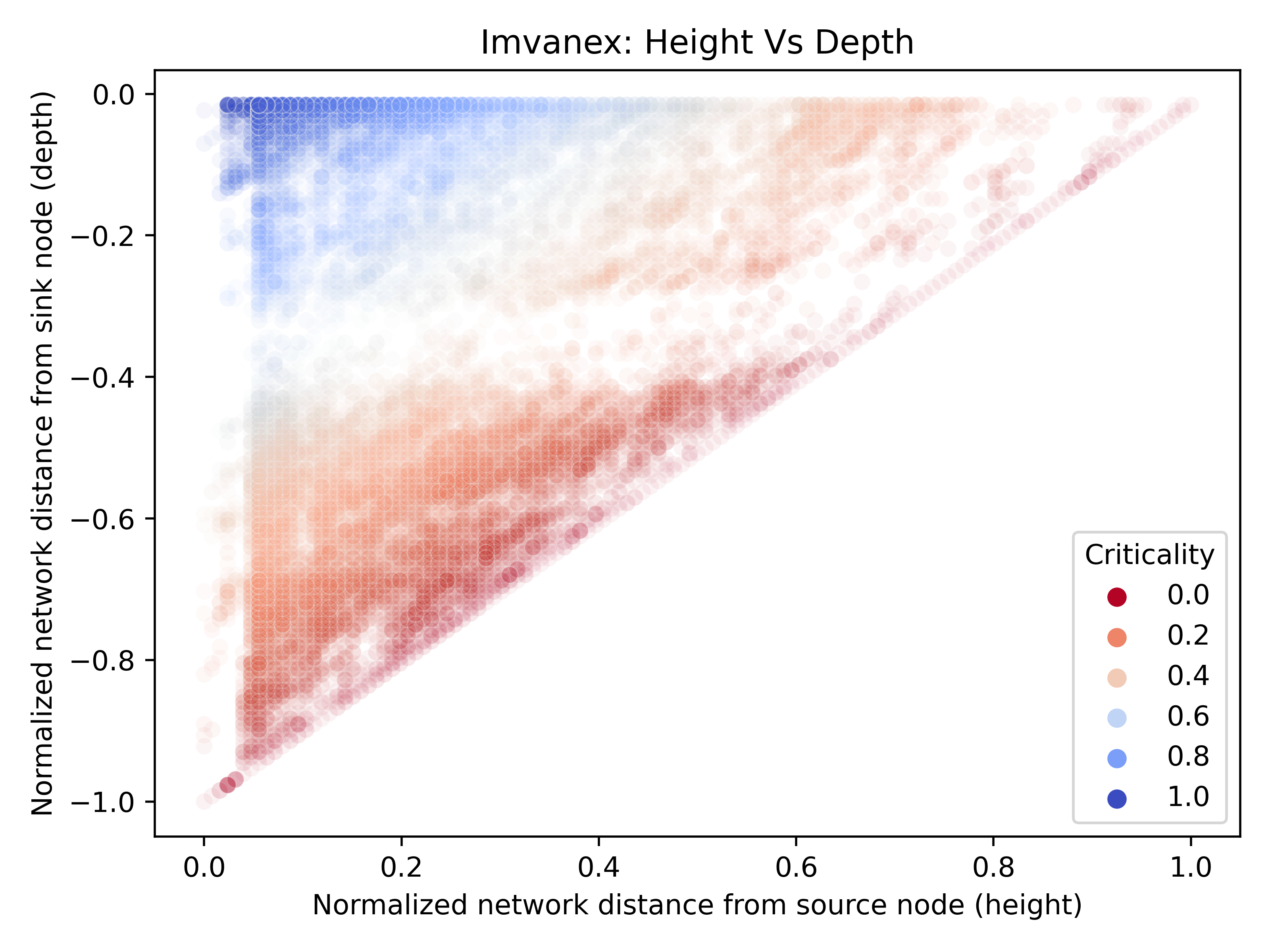}
         \caption{Imvanex, Smallpox, Bavarian Nordic, 2013, WPV}
     \end{subfigure}
     \hfill
     \begin{subfigure}[b]{0.24\textwidth}
         \centering
         \includegraphics[width=\textwidth]{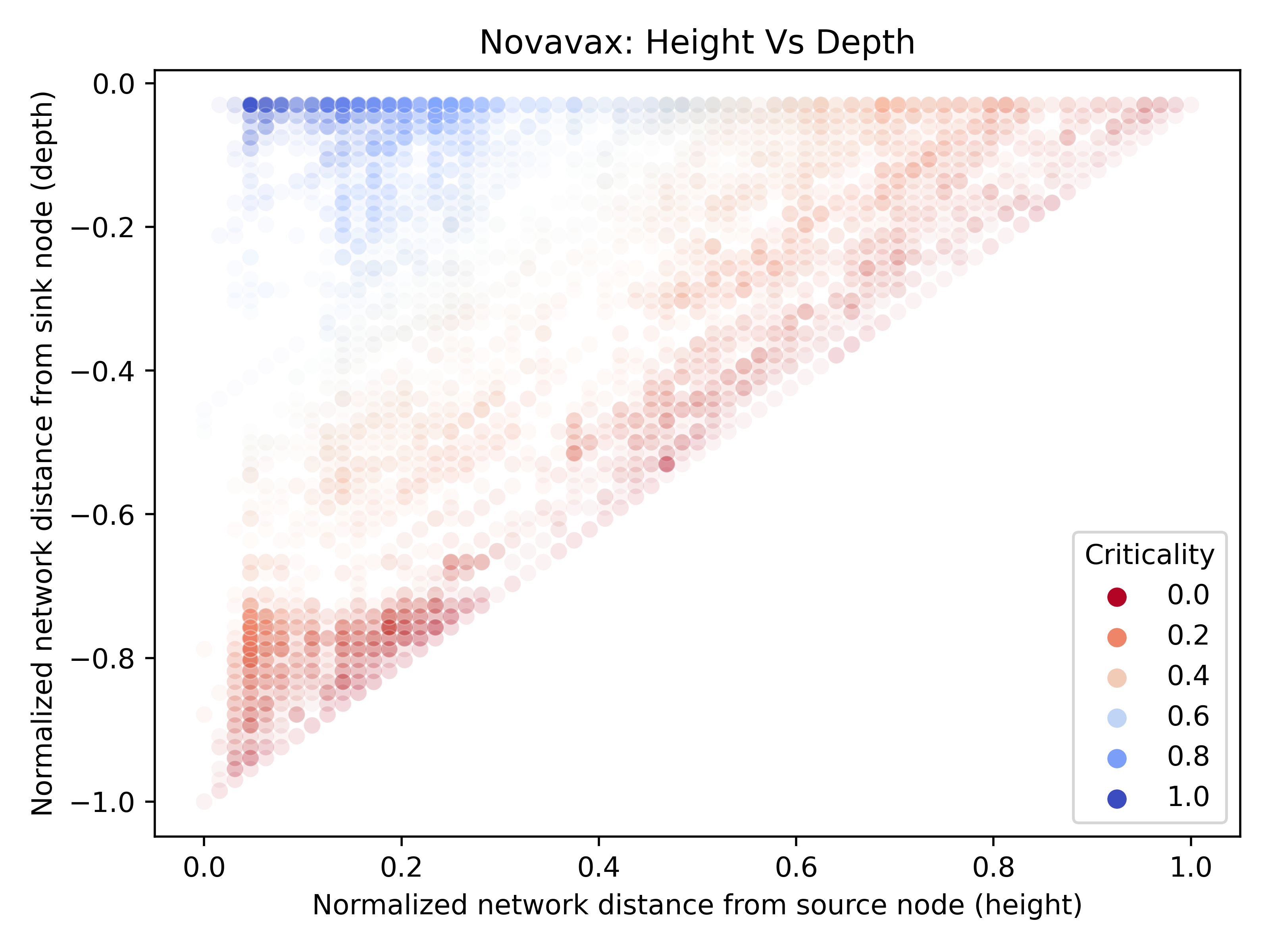}
         \caption{Nuvaxovid, COVID-19, Novavax, 2022, subunits}
     \end{subfigure}
     \hfill
     \begin{subfigure}[b]{0.24\textwidth}
         \centering
         \includegraphics[width=\textwidth]{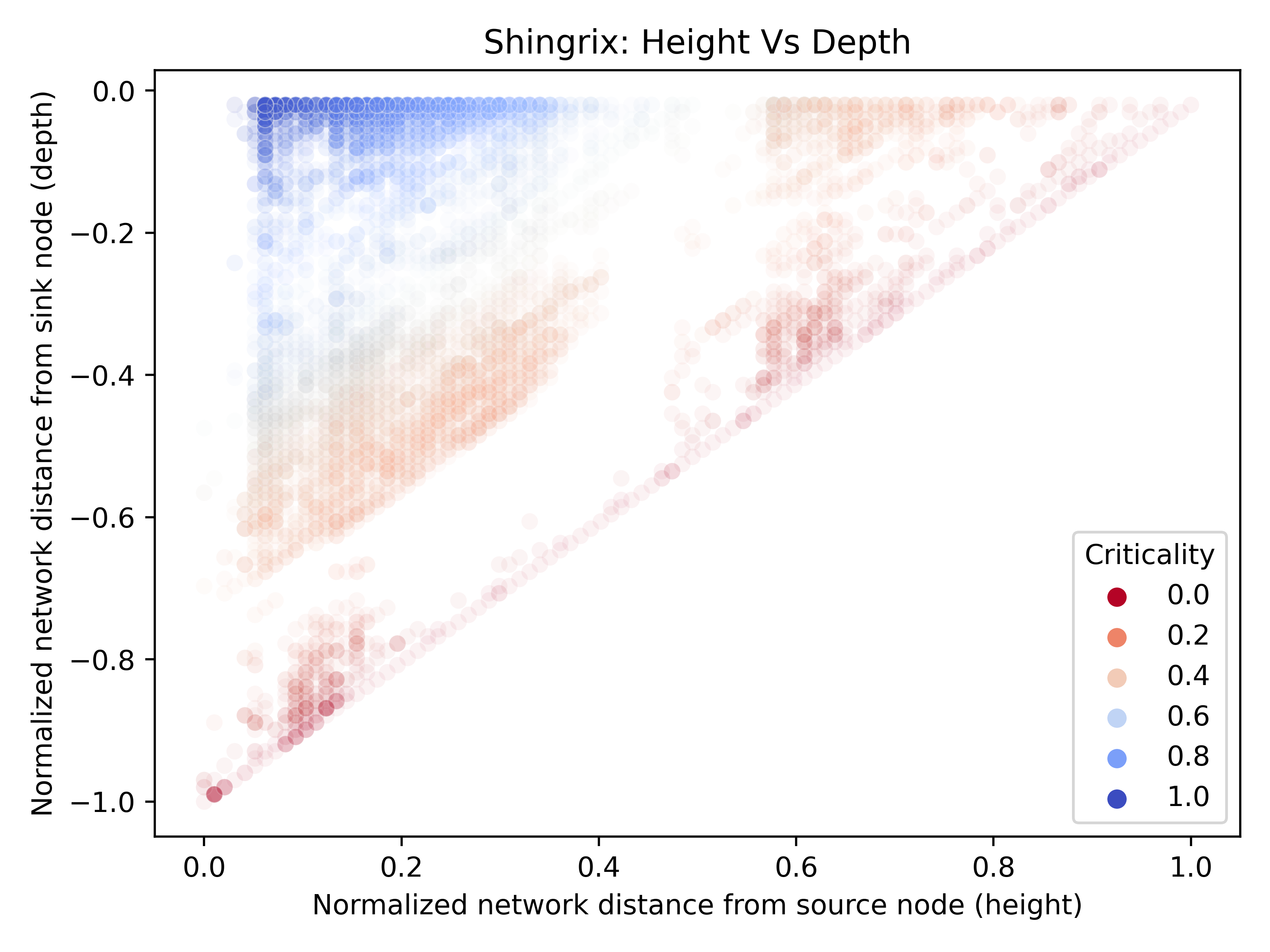}
         \caption{Shringrix, Shingles, GSK, 2017,\\ subunits}
     \end{subfigure}
     \hfill
        \caption{\textbf{Depth as a function of height} with hypotenuse showing longest network path and critical innovation path.}
        \label{fig:all_critical_paths}
\end{figure}

\begin{figure}[htb]
     \centering
     \begin{subfigure}[b]{0.24\textwidth}
         \centering
         \includegraphics[width=\textwidth]{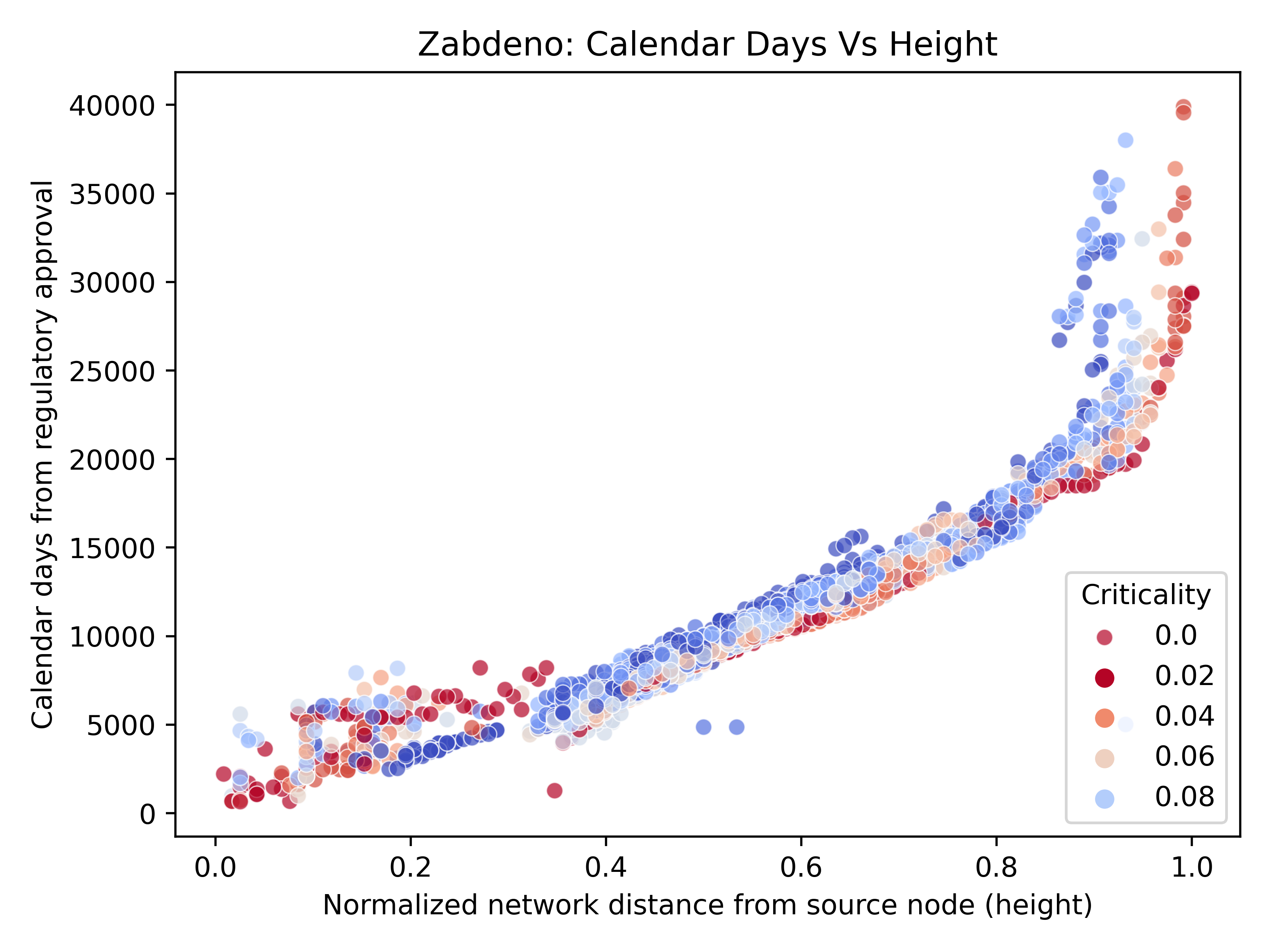}
         \caption{Zabdeno, Ebola, Janssen, 2020,\\AVV}
     \end{subfigure}
     \hfill
     \begin{subfigure}[b]{0.24\textwidth}
         \centering
         \includegraphics[width=\textwidth]{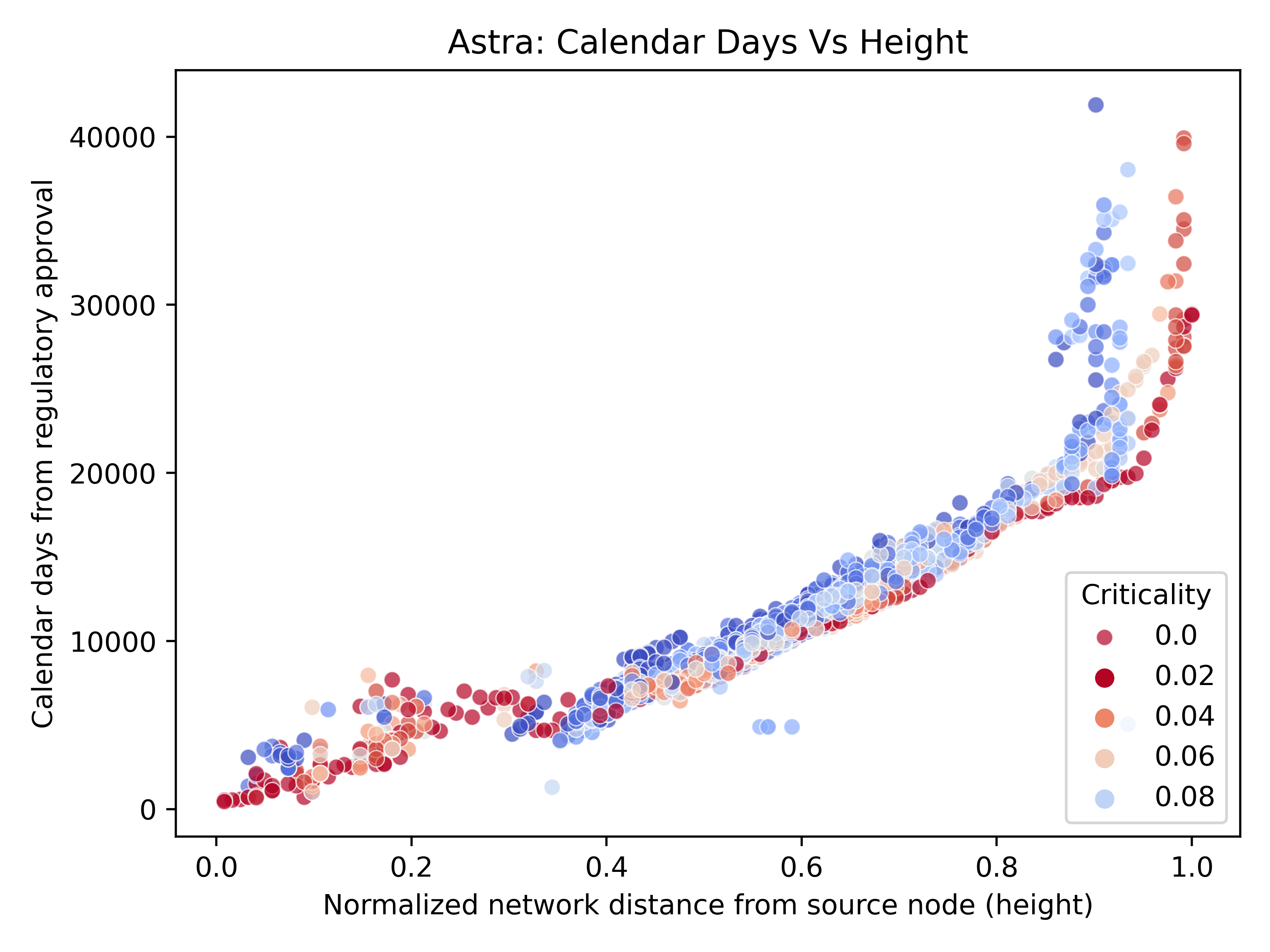}
         \caption{Vaxzevria, COVID-19, AstraZeneca, 2020, AVV}
     \end{subfigure}
     \hfill
     \begin{subfigure}[b]{0.24\textwidth}
         \centering
         \includegraphics[width=\textwidth]{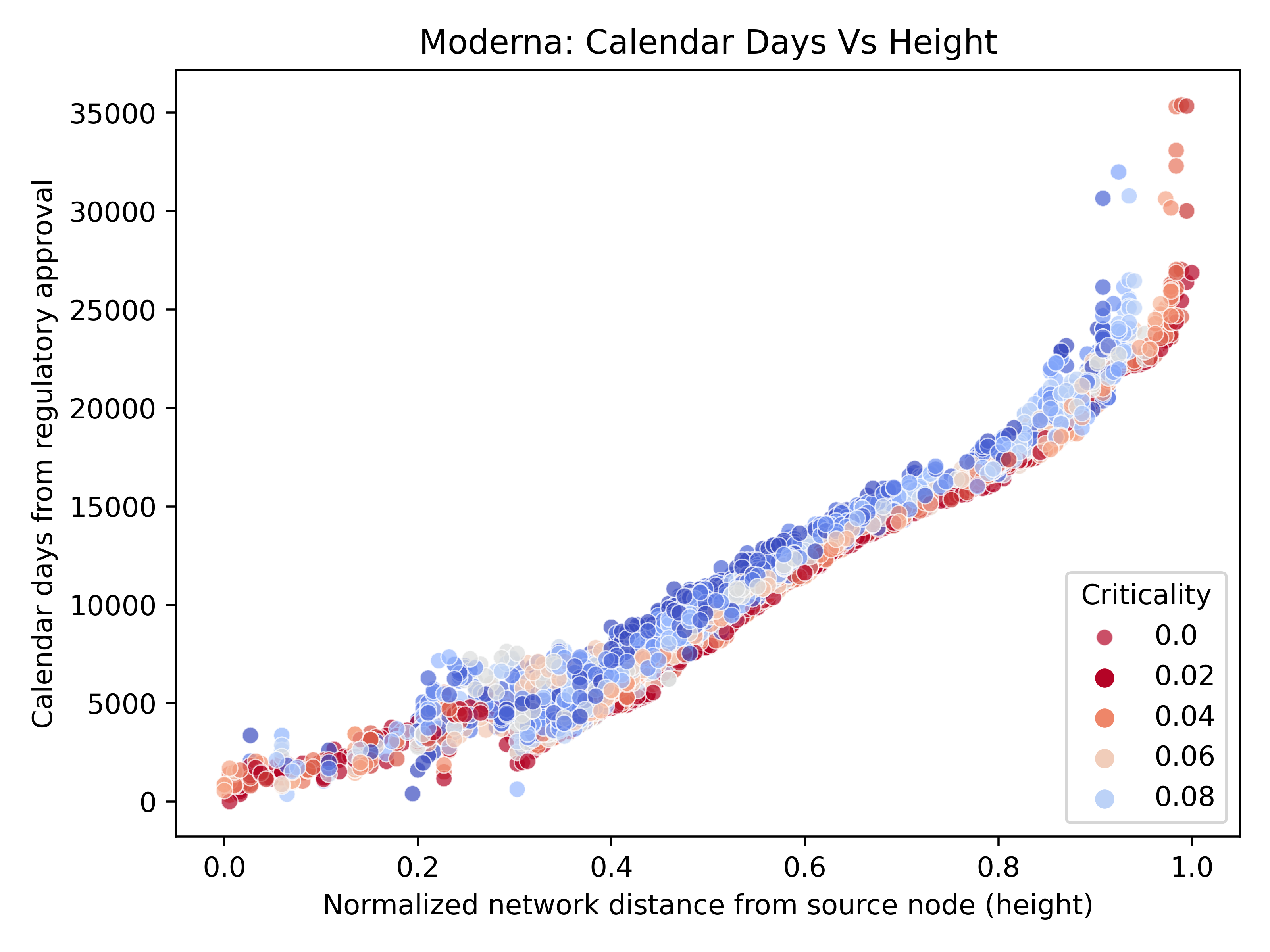}
         \caption{Spikevax, COVID-19, Moderna, 2020, mRNA}
     \end{subfigure}
     \hfill
     \begin{subfigure}[b]{0.24\textwidth}
         \centering
         \includegraphics[width=\textwidth]{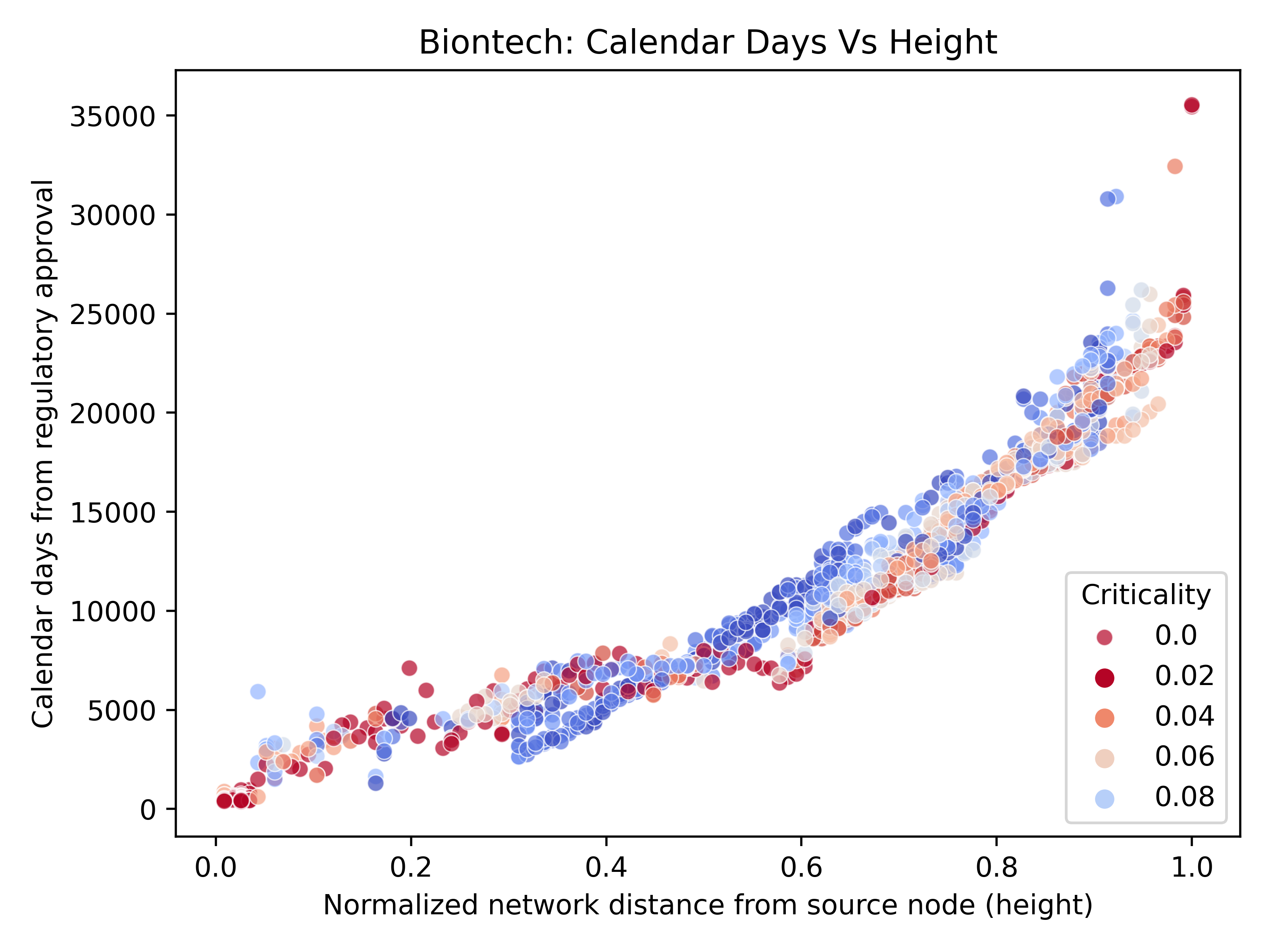}
         \caption{Comirnaty, COVID-19, BioNTech/Pfizer, 2020, mRNA}
     \end{subfigure}
     \hfill
     \begin{subfigure}[b]{0.24\textwidth}
         \centering
         \includegraphics[width=\textwidth]{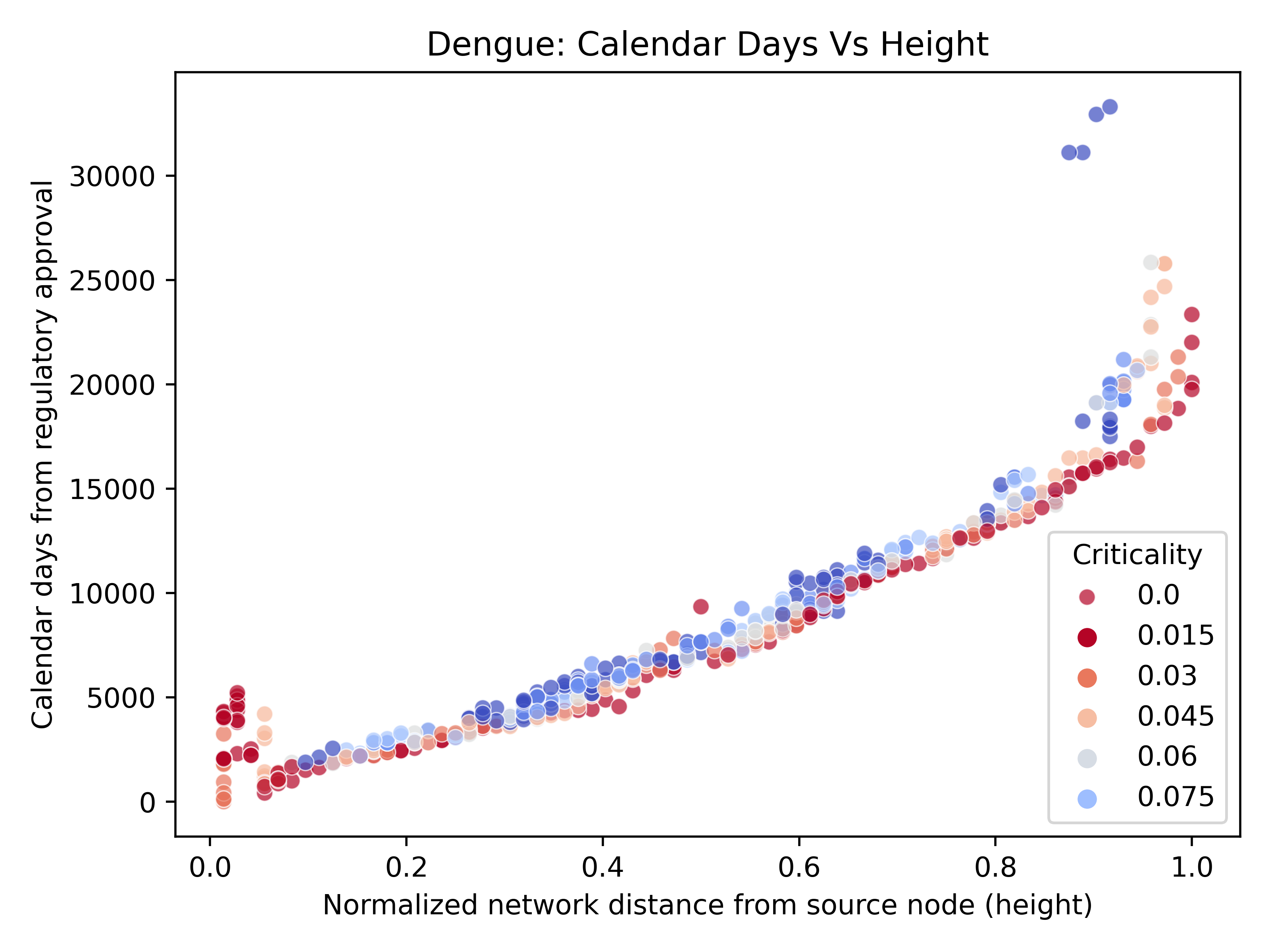}
         \caption{Dengvaxia, Dengue, Sanofi, 2019,\\ WPV}
     \end{subfigure}
     \hfill
     \begin{subfigure}[b]{0.24\textwidth}
         \centering
         \includegraphics[width=\textwidth]{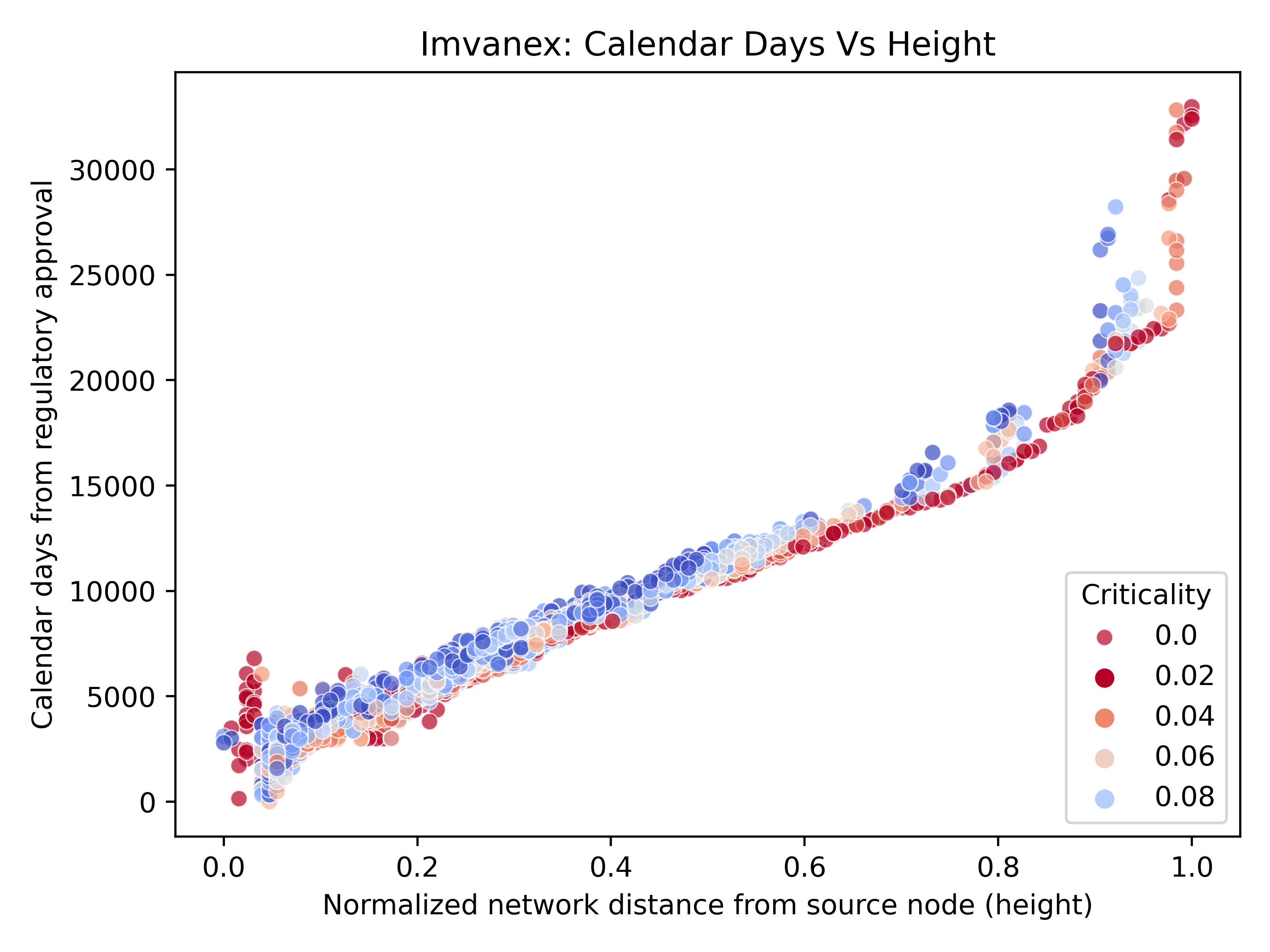}
         \caption{Imvanex, Smallpox, Bavarian Nordic, 2013, WPV}
     \end{subfigure}
     \hfill
     \begin{subfigure}[b]{0.24\textwidth}
         \centering
         \includegraphics[width=\textwidth]{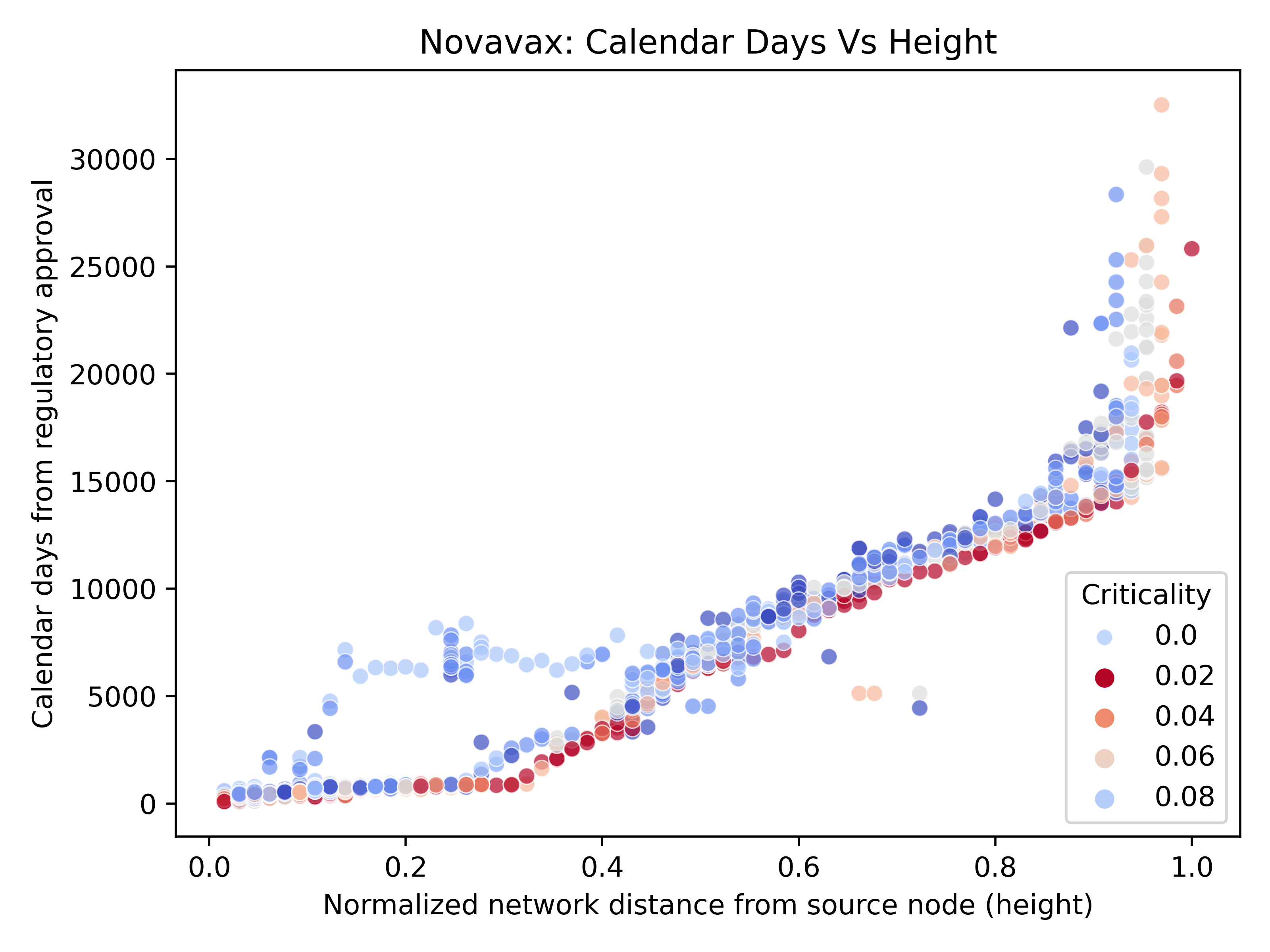}
         \caption{Nuvaxovid, COVID-19, Novavax, 2022, subunits}
     \end{subfigure}
     \hfill
     \begin{subfigure}[b]{0.24\textwidth}
         \centering
         \includegraphics[width=\textwidth]{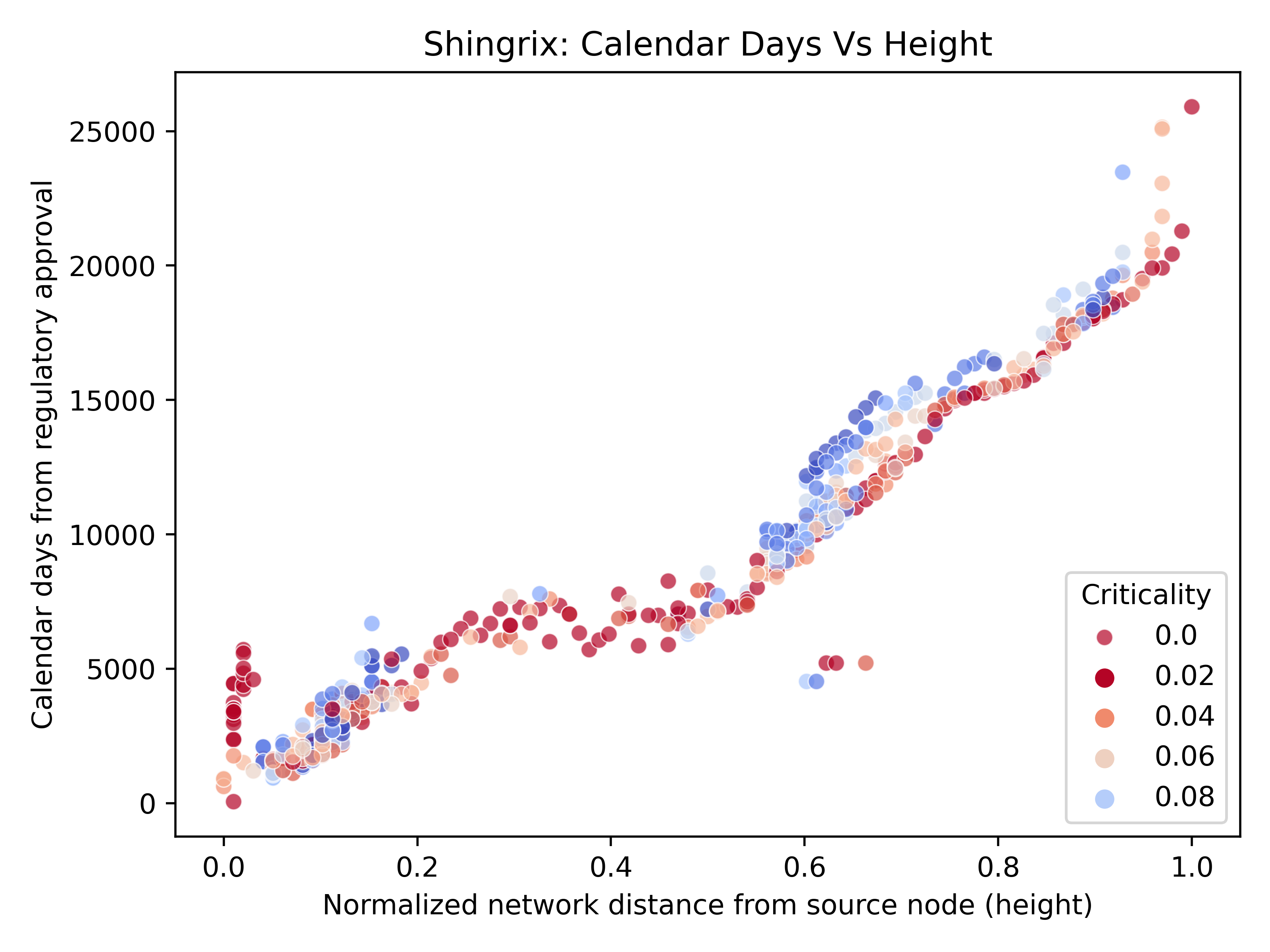}
         \caption{Shringrix, Shingles, GSK, 2017,\\ subunits}
     \end{subfigure}
     \hfill
        \caption{\textbf{Calendar days as a function of network height.}}
        \label{fig:all_date_v_height}
\end{figure}

\begin{figure}[htb]
     \centering
     \begin{subfigure}[b]{0.24\textwidth}
         \centering
         \includegraphics[width=\textwidth]{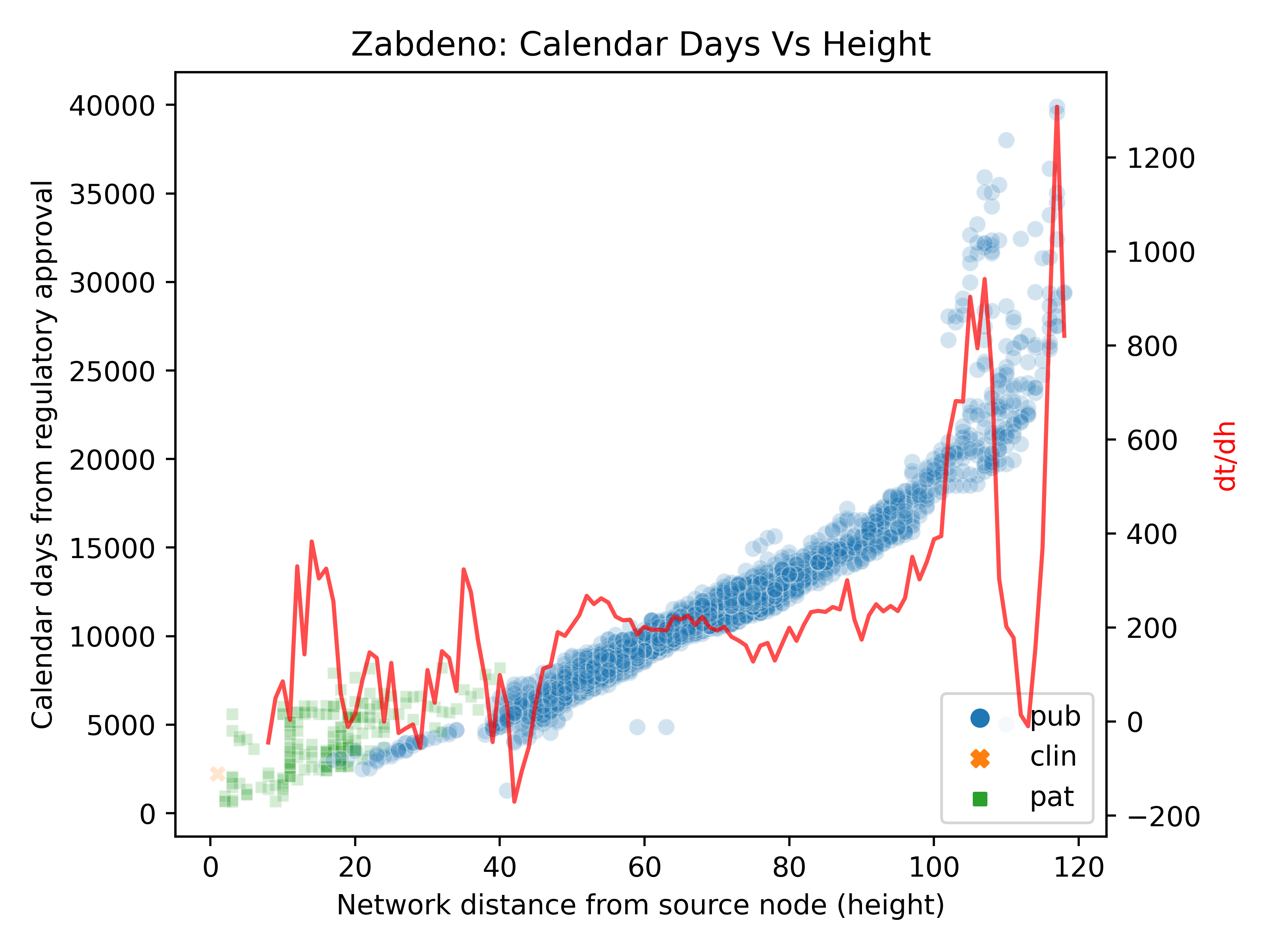}
         \caption{Zabdeno, Ebola, Janssen, 2020,\\AVV}
     \end{subfigure}
     \hfill
     \begin{subfigure}[b]{0.24\textwidth}
         \centering
         \includegraphics[width=\textwidth]{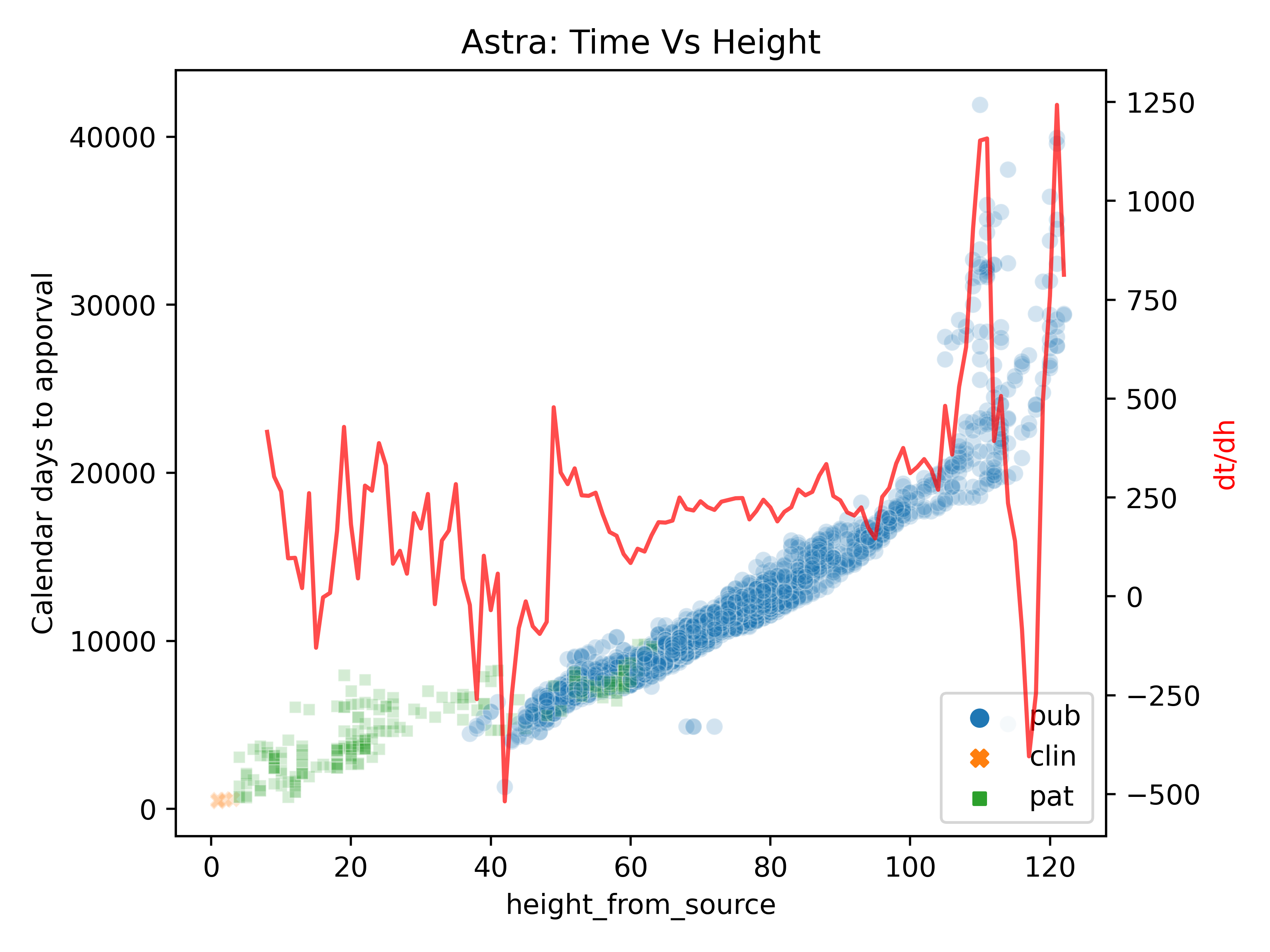}
         \caption{Vaxzevria, COVID-19, AstraZeneca, 2020, AVV}
     \end{subfigure}
     \hfill
     \begin{subfigure}[b]{0.24\textwidth}
         \centering
         \includegraphics[width=\textwidth]{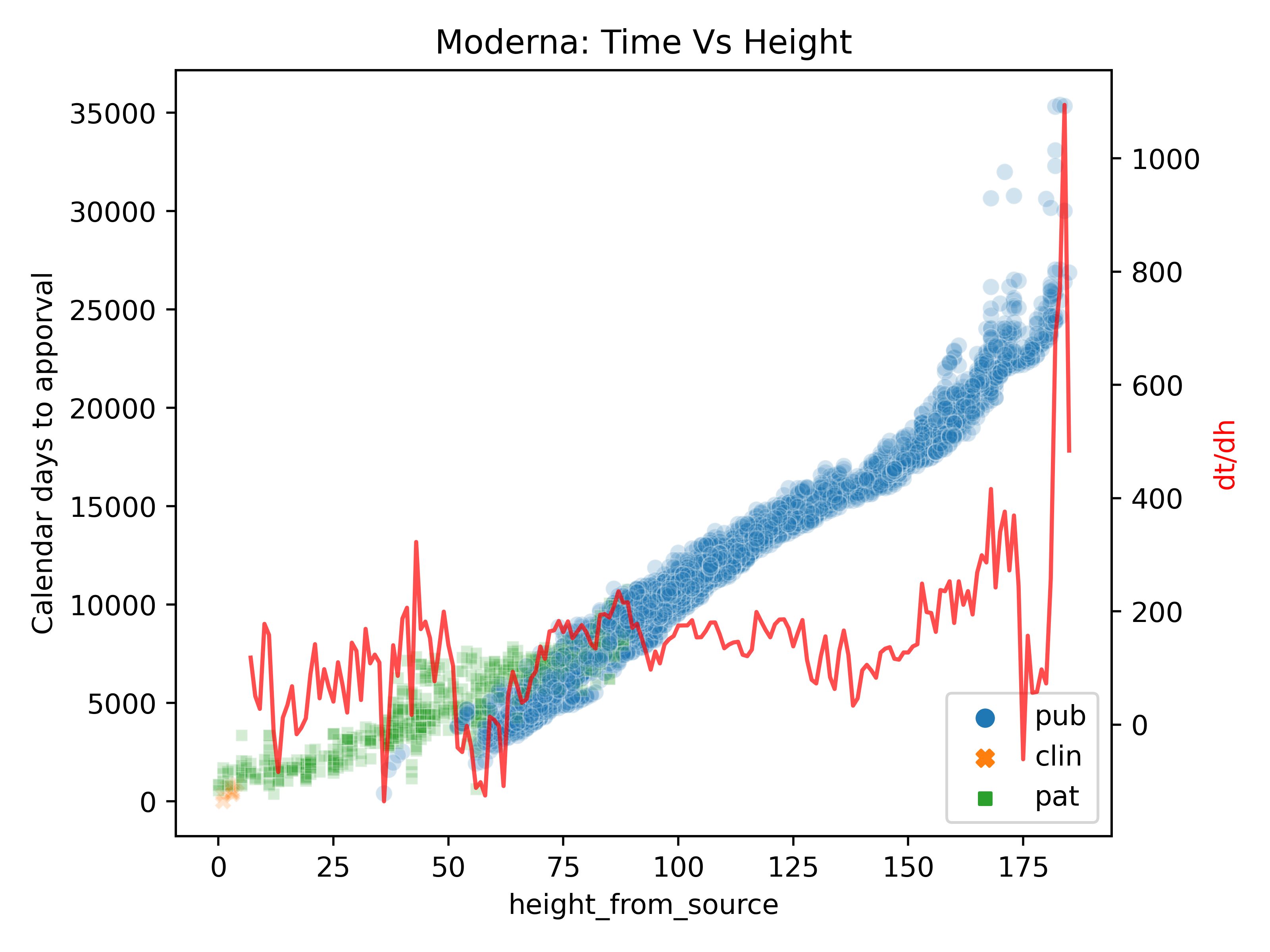}
         \caption{Spikevax, COVID-19, Moderna, 2020, mRNA}
     \end{subfigure}
     \hfill
     \begin{subfigure}[b]{0.24\textwidth}
         \centering
         \includegraphics[width=\textwidth]{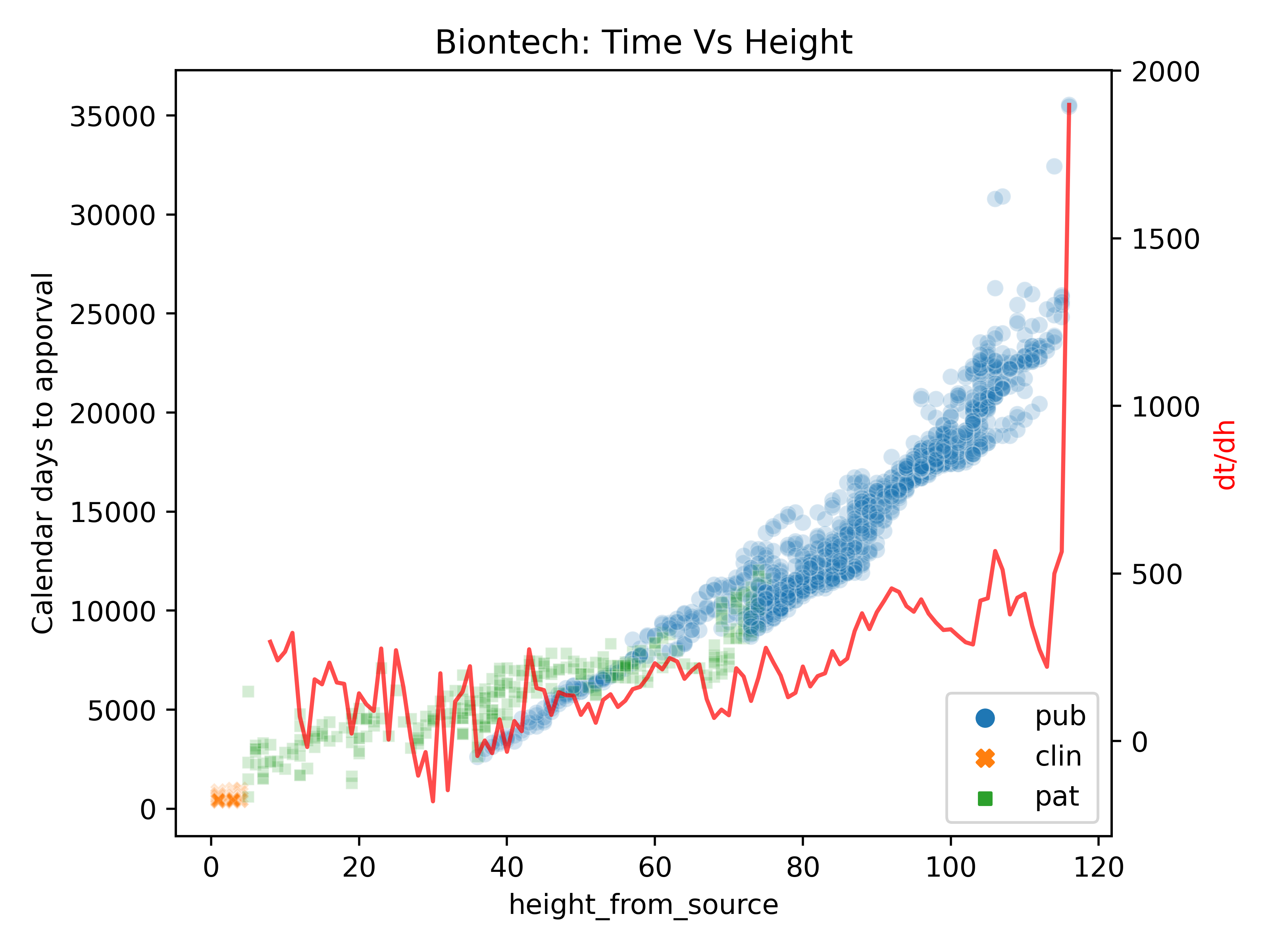}
         \caption{Comirnaty, COVID-19, BioNTech/Pfizer, 2020, mRNA}
     \end{subfigure}
     \hfill
     \begin{subfigure}[b]{0.24\textwidth}
         \centering
         \includegraphics[width=\textwidth]{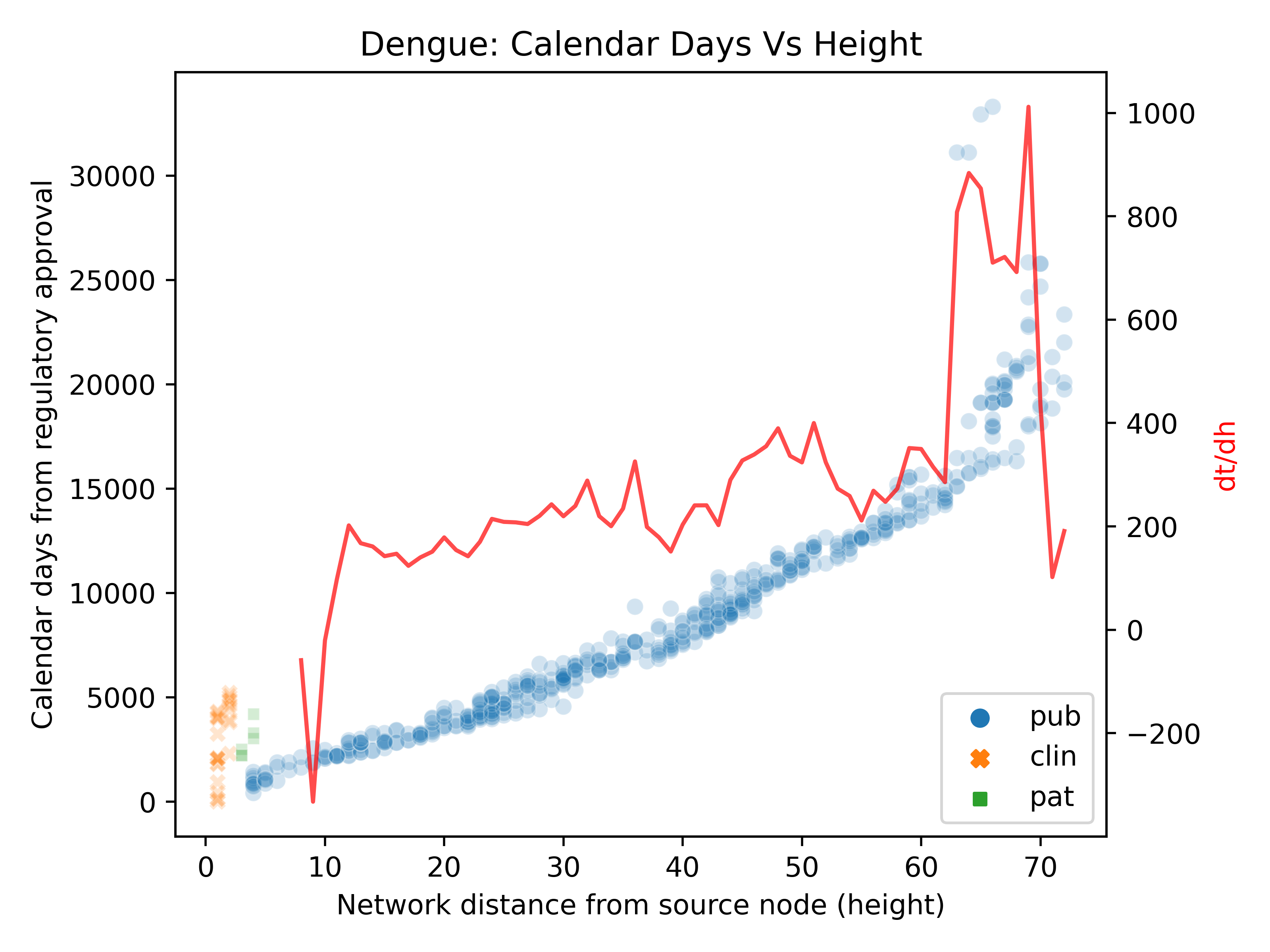}
         \caption{Dengvaxia, Dengue, Sanofi, 2019,\\ WPV}
     \end{subfigure}
     \hfill
     \begin{subfigure}[b]{0.24\textwidth}
         \centering
         \includegraphics[width=\textwidth]{height_vs_date_imvanex_rolling_diff.png}
         \caption{Imvanex, Smallpox, Bavarian Nordic, 2013, WPV}
     \end{subfigure}
     \hfill
     \begin{subfigure}[b]{0.24\textwidth}
         \centering
         \includegraphics[width=\textwidth]{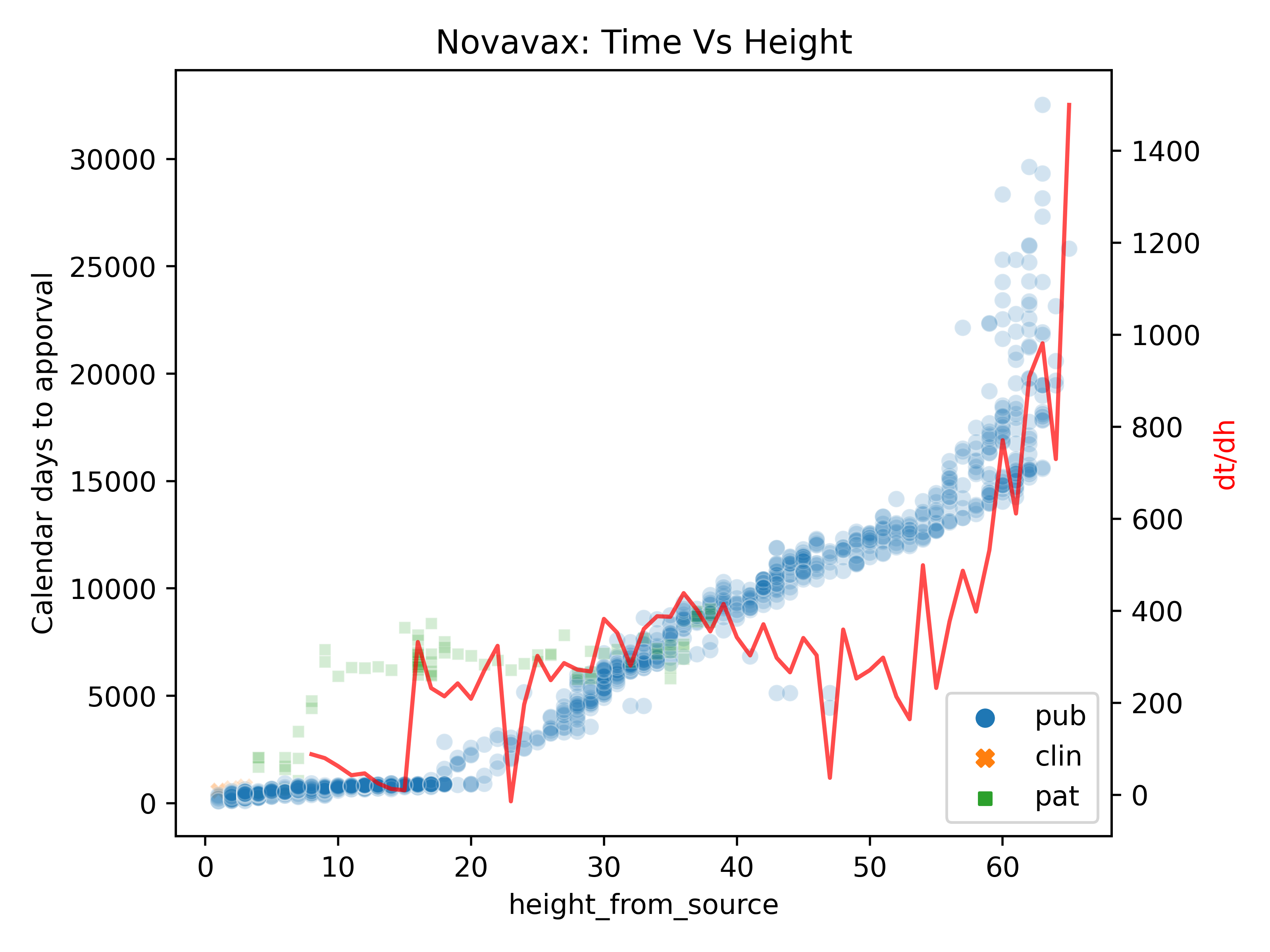}
         \caption{Nuvaxovid, COVID-19, Novavax, 2022, subunits}
     \end{subfigure}
     \hfill
     \begin{subfigure}[b]{0.24\textwidth}
         \centering
         \includegraphics[width=\textwidth]{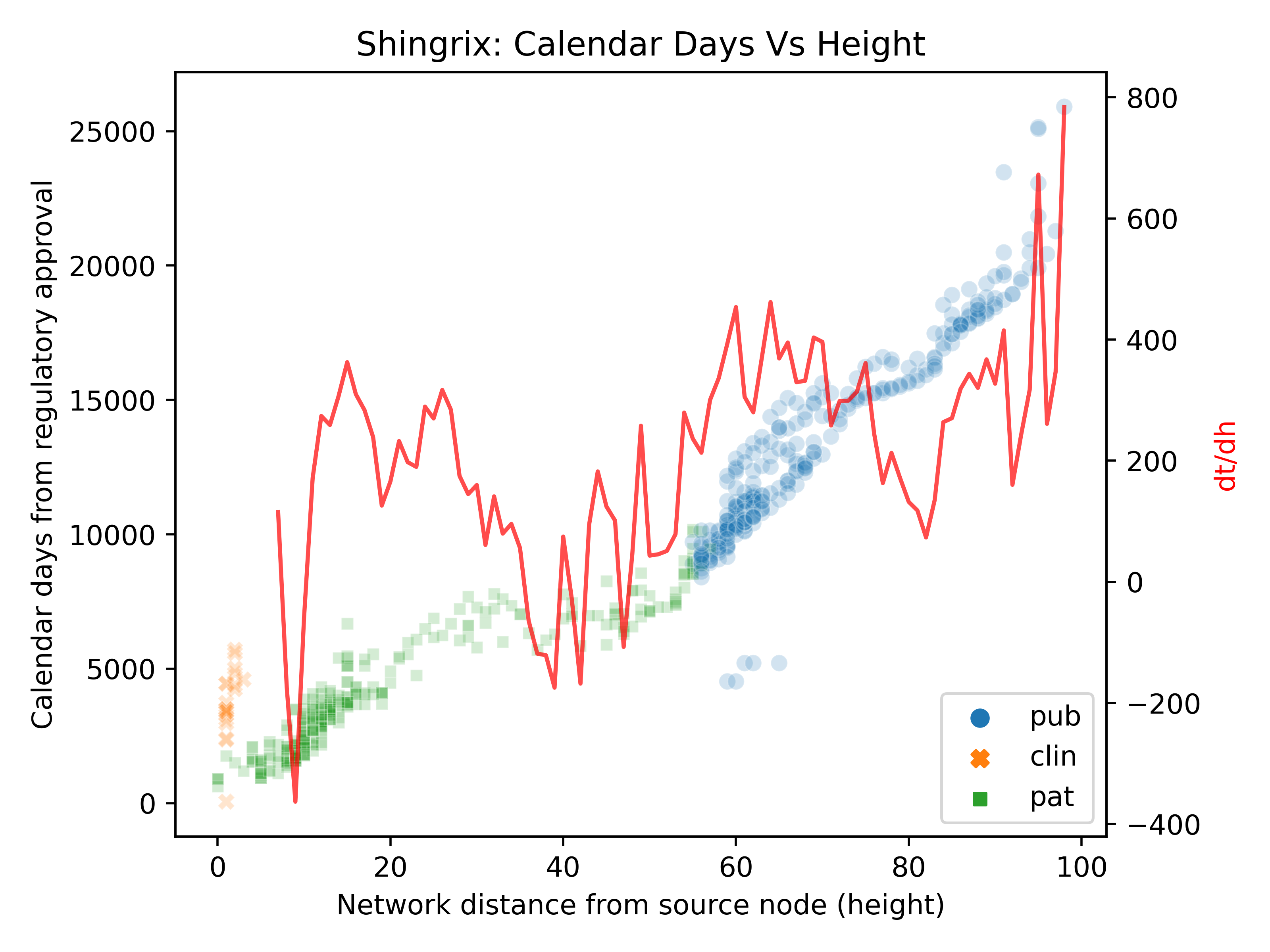}
         \caption{Shringrix, Shingles, GSK, 2017,\\ subunits}
     \end{subfigure}
     \hfill
        \caption{\textbf{Calendar days per height as a function of height.}}
        \label{fig:all_day_per_height_v_height}
\end{figure}

\clearpage
\section{Further methodological discussions}\label{validity}

\subsection{Natural experiment} \label{natural_experiment}

Randomized controlled trials
aim to eliminate selection bias, but are mostly only feasible in clinical trials. The estimation of average treatment effect comes with the important assumptions that the treatment of any participant does not have an effect on other participants and that all confounding covariates are accounted for. To mimic randomization in the economy where regions and stakeholders cannot be randomly assigned into groups, economists perform natural experiments where exogenous events ``as if'' randomly assign subjects into treatment and control groups. Notable examples of natural experiments include difference-in-differences, which estimates the average treatment effect by comparing treatment and control groups, who would otherwise move in parallel without the treatment, in multiple time panels~\cite{RN1014,RN1634}; and regression discontinuity, which exploits abrupt changes, cutoffs, or thresholds to ``as if'' randomly assign samples into treatment and control groups~\cite{RN1628,RN1629,RN1631,RN1630}. One validity requirement of natural experiments rests on the assumption that assignment is sufficiently  ``as if'' random.

\subsection{Main Path Analysis} \label{main_path_analysis}

Here we will give a brief summary of ``main path analysis''. This was proposed by Hummon and Doreian~\cite{RN997} and, with variations, has been implemented in some popular analysis packages, for example see \cite{RN1245}.

Main path analysis starts by looking at all possible paths between a specified set of nodes, a set which varies between various implementations of main path analysis. The number of these paths which pass through a given edge is used to assign the ``edge weight'' of that edge, i.e.\  a value assigned to that edge. Then the length of a path is defined to be the sum of the edges weights traversed in that path. Finally, the ``main path'' is defined using a greedy algorithm to find paths of high length as defined using the edge weights. That is to construct a main path, the current main path is extended by adding one more node to the end of the path such that the length increases by the largest amount. The main path will end when it reaches a sink node, a node with no outgoing edges.

The first main path analysis described the scientific advances that eventually led to the discovery of the DNA structure by Watsons and Crick \cite{RN997}. Main path analysis is later applied to other networks of academic publications~\cite{HC93,RN1246,RN361,RN1261,RN1248,RN1253,RN1254,RN1260,RN1255,RN1011,A23} and, more recently, patents~\cite{RN361,RN1247,RN1249,RN1093,RN1252,RN1001,RN1251,RN1231,RN1259,RN1092}. Case studies of main path analysis span from understanding the emergence of engineered products such as battery, nanotubes, automobiles, and semiconductors; to academic theories such as bioinformatics, social network analysis, absorptive capacity, Hirsch index, and peer review. Most of these case studies employ one of the four ``out-of-the-box" indices developed by Hummon and Doreian~\cite{RN997} and Batagelj~\cite{RN1245} with the objective to reduce the number of nodes in a citation network to a single chain of events to enable qualitative interpretation. However, these studies do not consider the longest path beyond a simplification device.

\subsection{Edge weighting}

We choose to assign a weighting of 1 to all edges because we assume each edge in the citation network, including in the longest path, represents a minimum viable increment of novelty. A weight of 1 allows equivalence in increments of innovation. We believe this assumption is valid because being published in a journal, accepted as a patent, approved to run a clinical trial, or authorised to market a therapeutic represents a minimum normalised threshold of originality from peer-review\footnote{Anecdotally also known as the LPI or Least Publishable Unit \url{https://en.wikipedia.org/wiki/Least_publishable_unit}}. i.e a group would not be able to publish any earlier and would not delay publication as they would seek to publish an increment as soon as possible. Conversely, any reweighting, such as in main path analysis, impedes interpretability. Future works can use funding amount as weight if data becomes more complete; whereas we do not recommend using time as edge weights because time is already implied in network height and an edge that consumes a long duration of time does not mean it is more novel.

\subsection{Citation behaviours}
\label{citation_behaviours}

When studying citation networks, it is important to note that citation practices vary. Different types of document may have different goals, and publishers set their own constraints on the bibliographies. The citation tradition in various fields can be very different, while individual authors add another source of variability. For instance, patent applicants need to strike the balance between minimising citations to demonstrate novelty and citing enough to not infringe prior arts~\cite{CE16, RN1710}. Based on \figref{fig:funder_network_order}, we can believe privately-funded publications are likelier to have end-uses in mind and may bias citations towards applied research; whereas publicly-funded publications may be preoccupied with phenomenological questions. In addition, funders enforce grant acknowledgements in publications and patents differently.
For example, the US \href{https://en.wikipedia.org/wiki/Bayh-Dole_Act}{Bayh-Dole Act} requires that all recipients of federal research funds report to the funding agency any patent they file and acknowledge on patent documents the existence of federal funding, while many other countries do not have similar requirements. The different citation behaviours are likely more pronounced in the multilayer citation network we use as it assumes publications, patents, clinical trials, and regulatory approvals cite in the same way.

\subsection{Patent family}
\label{patent_family}

A patent family is a collection of patent applications covering the same or similar technical content. Patent families usually arise from a single invention being filed in multiple countries (``simple patent family") and when an applicant files new applications for similar existing technical contents (``extended patent family").  \secref{patent_prosecution} below explains the importance of considering patent families in our network.

\subsection{Patent prosecution}
\label{patent_prosecution}

 Patent application often spans several years. Four key dates in chronological order are:

 \begin{enumerate}
     \item Priority date: date used to establish the novelty of an invention
     \item Filing date: when a patent application is first filed at a patent office
     \item Publication date: when a patent application is published
     \item Grant date: when a patent office grants a patent
 \end{enumerate}

 Patent prosecution is the interaction among patent applicants, patent offices including examiners, and other interested parties. Patent prosecution usually spans between (2) filing date and (4) grant date, but can extend after grant if there is opposition, corrections, or other post-grant proceedings. 
 
 Due to patent prosecution, the bibliography of almost every patent is updated with new references. Almost every patent gets citations added during their prosecution time. These can be added by the examiner, by the applicant, assignee, or the inventor. What occurs less frequently is for citations to be added after grant. These usually happen for more limited reasons, e.g. post grant opposition, corrections, reissues, etc.
 
 We use the initial patent submission date as our patent publication date. A year or two into the patent process, a recent paper can be added to the application, one that was published after the patent was submitted. As a result a patent may cite forward in time as well as the logically acceptable backwards in time. We could use the patent award date as our patent publication date which would solve the problem with the example just given.  However, now we run into problems with documents that cite a patent that is not yet approved yet is a critical part of the innovation process. This illustrates why our using the height of a node in our citation network can be a more consistent record of the logical order in the innovation process compared to calendar time. We also address this issue by considering patent families rather than single patents when possible to capture references added to a patent during patent prosecution.

\subsection{Critical path hit rate}
\label{critical_path_hit_rate}

The criticality of a given node is clearly defined in this paper by equation \eqref{eq:criticality}. Critical innovation path, on the other hand, depends on a threshold -- nodes with a criticality below an arbitrary value would be considered residing on the critical innovation path. To provide a fair comparison across funders and across vaccines in \tabref{tab:funder_criticality}, we define the critical path as nodes whose criticality is below the maximum height in a DAG multiplied by a criticality threshold $x$. To determine the value of $x$, we conducted a robustness check (\figref{fig:robustness_shingrix}) to determine that the criticality threshold for the Shingrix network would be 0.35 as this is when most funders become present on the critical path. In \tabref{tab:funder_criticality}, we only include funders who funded more than one node on the critical path and more than ten nodes in the entire network for meaningful comparison.

\begin{figure}
  \centering
    \includegraphics[width=1\textwidth]{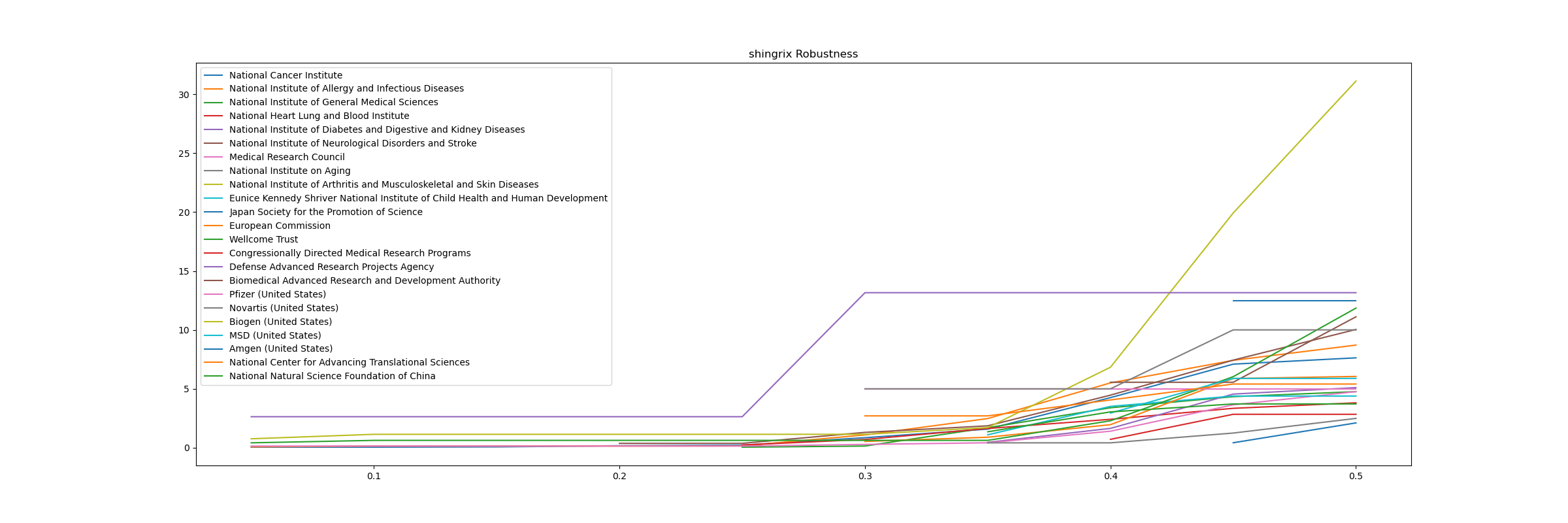}
  \caption{\textbf{Number of critical path nodes funded by funders as a function of criticality threshold.} Illustrative data from Shingrix vaccine. x-axis: critical path threshold, y-axis: number of critical path documents for funders.}
  \label{fig:robustness_shingrix}
\end{figure}

\subsection{Network density} The density of nodes reflects ambiguity in the networks' local and global order. \figref{fig:all_density} shows the citation networks are densest at low height and sparsest at high height. The latter is due to dangling nodes, potentially due to incomplete citation data in early years, meaning these regions are sensitive to change. On the other hand, observations drawn from other heights are more stable.

\begin{figure}
     \centering
     \begin{subfigure}[b]{0.24\textwidth}
         \centering
         \includegraphics[width=\textwidth]{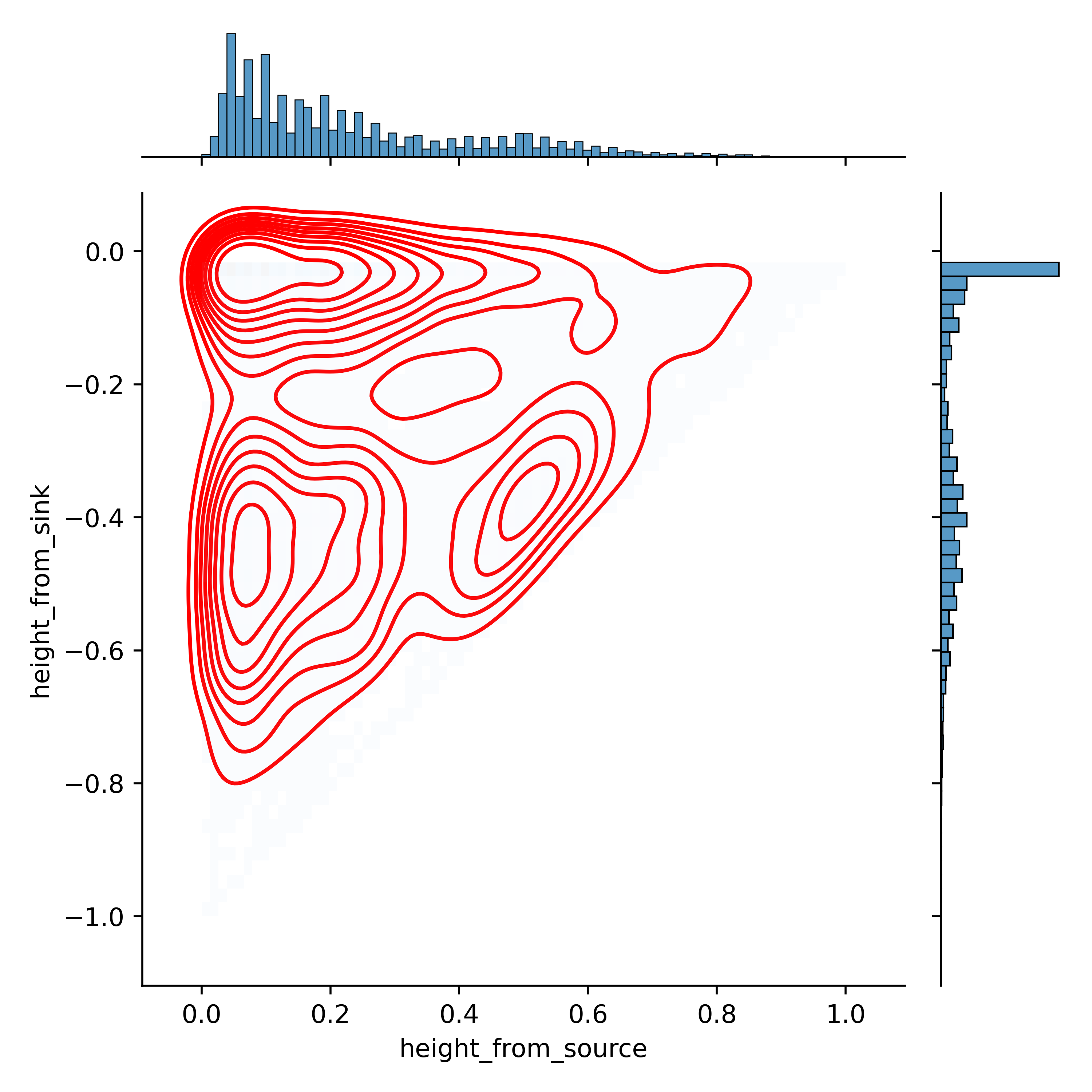}
         \caption{Zabdeno, \\Ebola, Janssen, 2020,\\AVV}
     \end{subfigure}
     \hfill
     \begin{subfigure}[b]{0.24\textwidth}
         \centering
         \includegraphics[width=\textwidth]{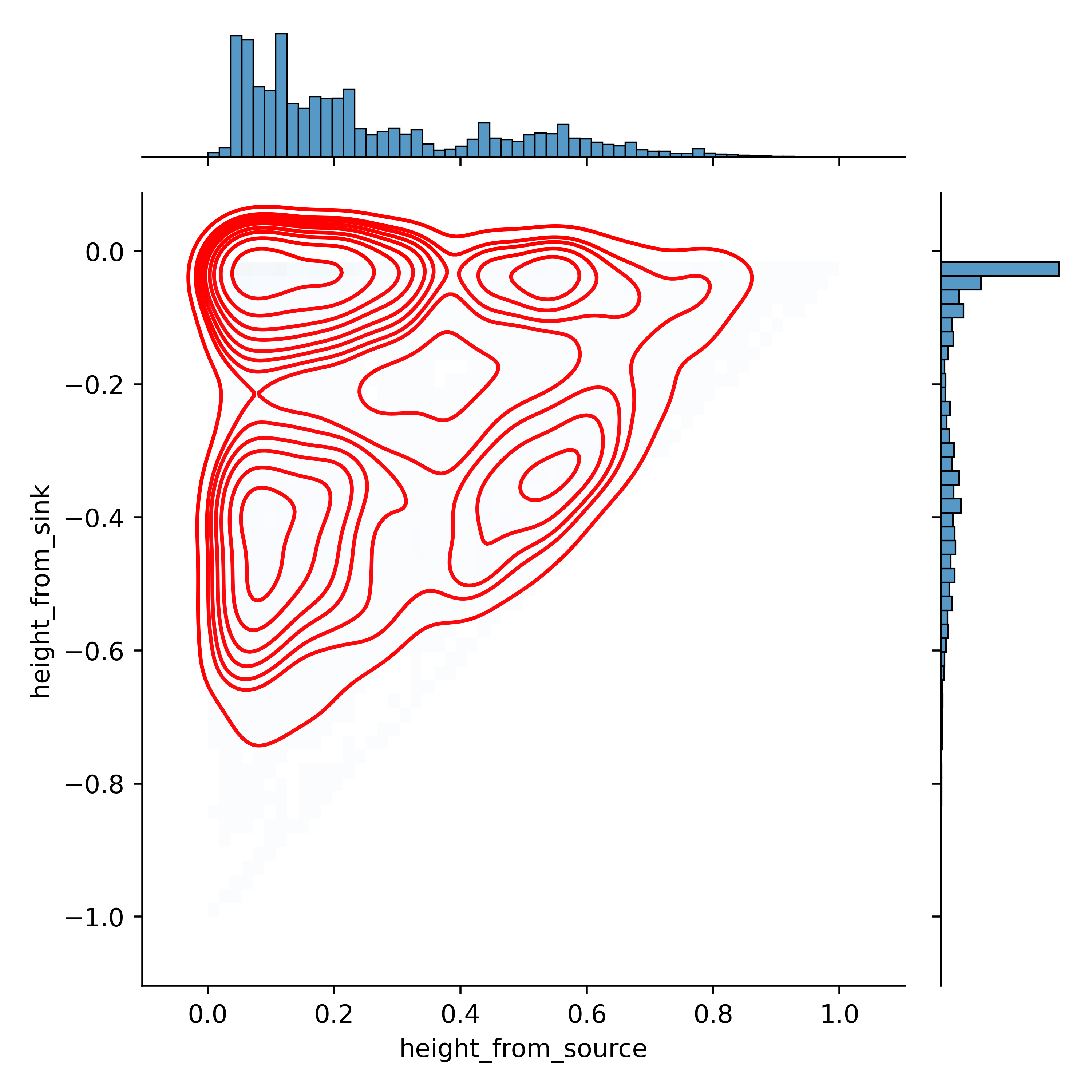}
         \caption{Vaxzevria, COVID-19, AstraZeneca, 2020, AVV}
     \end{subfigure}
     \hfill
     \begin{subfigure}[b]{0.24\textwidth}
         \centering
         \includegraphics[width=\textwidth]{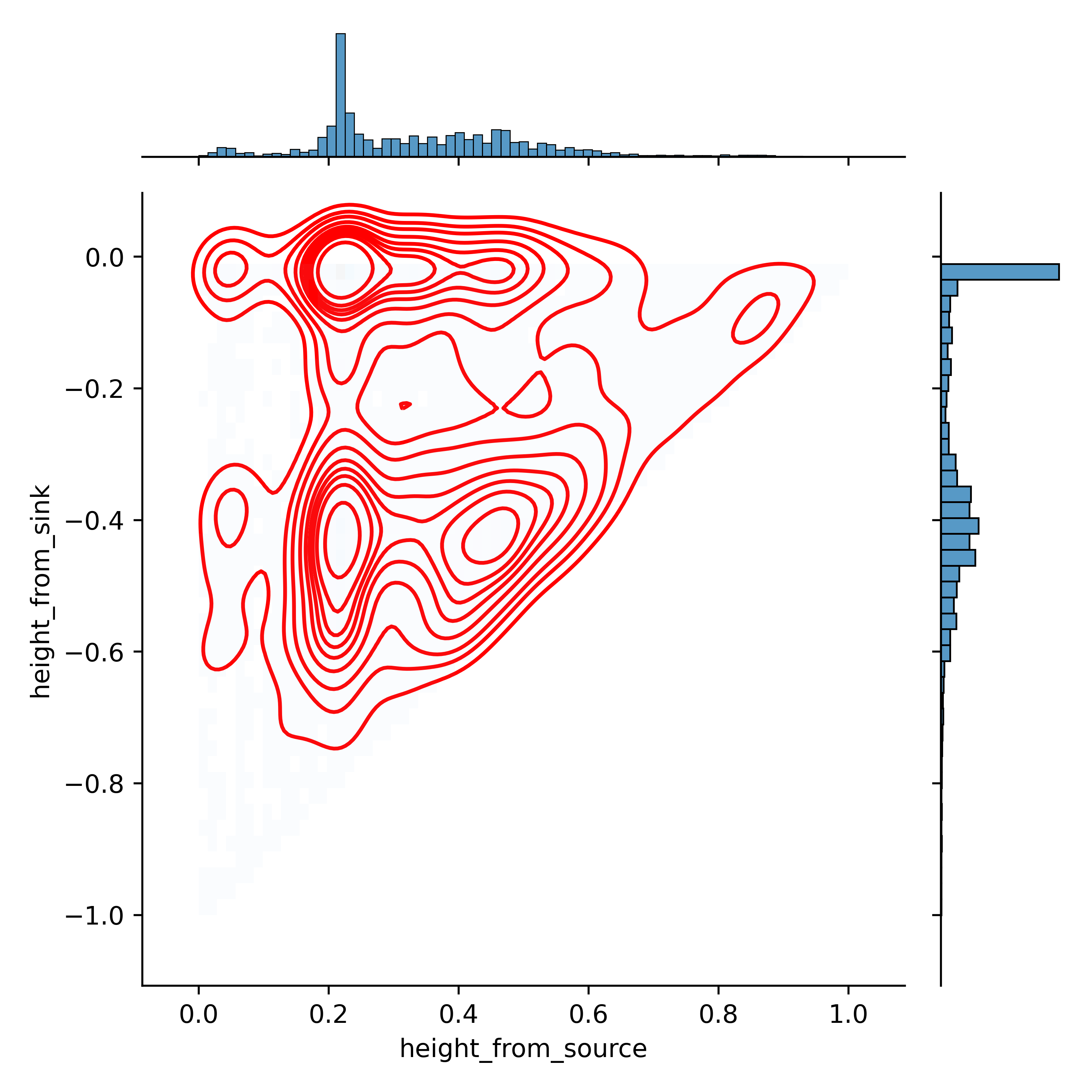}
         \caption{Spikevax, COVID-19, Moderna,\\ 2020, mRNA}
     \end{subfigure}
     \hfill
     \begin{subfigure}[b]{0.24\textwidth}
         \centering
         \includegraphics[width=\textwidth]{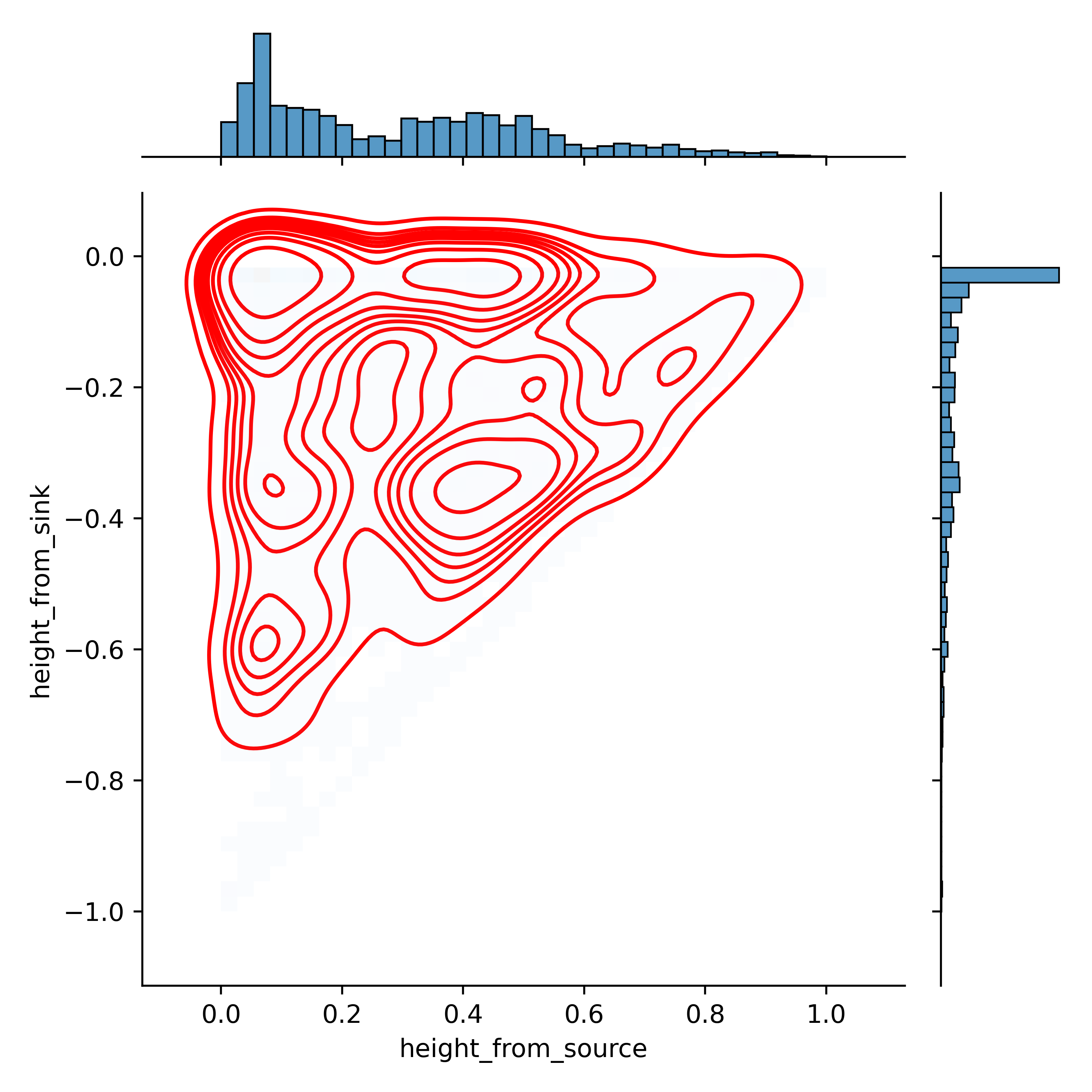}
         \caption{Comirnaty, COVID-19, BioNTech/Pfizer, 2020, mRNA}
     \end{subfigure}
     \hfill
     \begin{subfigure}[b]{0.24\textwidth}
         \centering
         \includegraphics[width=\textwidth]{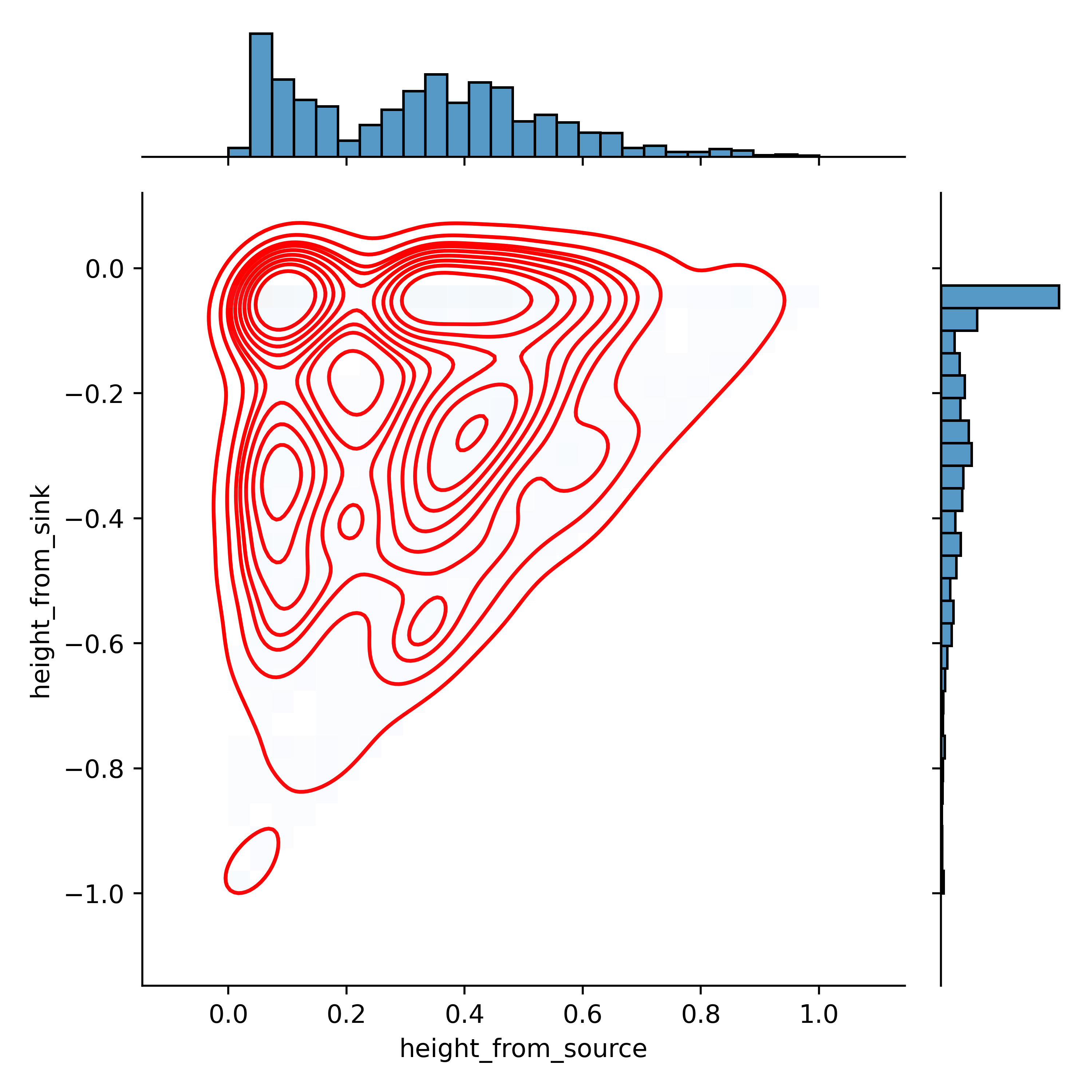}
         \caption{Dengvaxia, Dengue, Sanofi, 2019,\\ WPV}
     \end{subfigure}
     \hfill
     \begin{subfigure}[b]{0.24\textwidth}
         \centering
         \includegraphics[width=\textwidth]{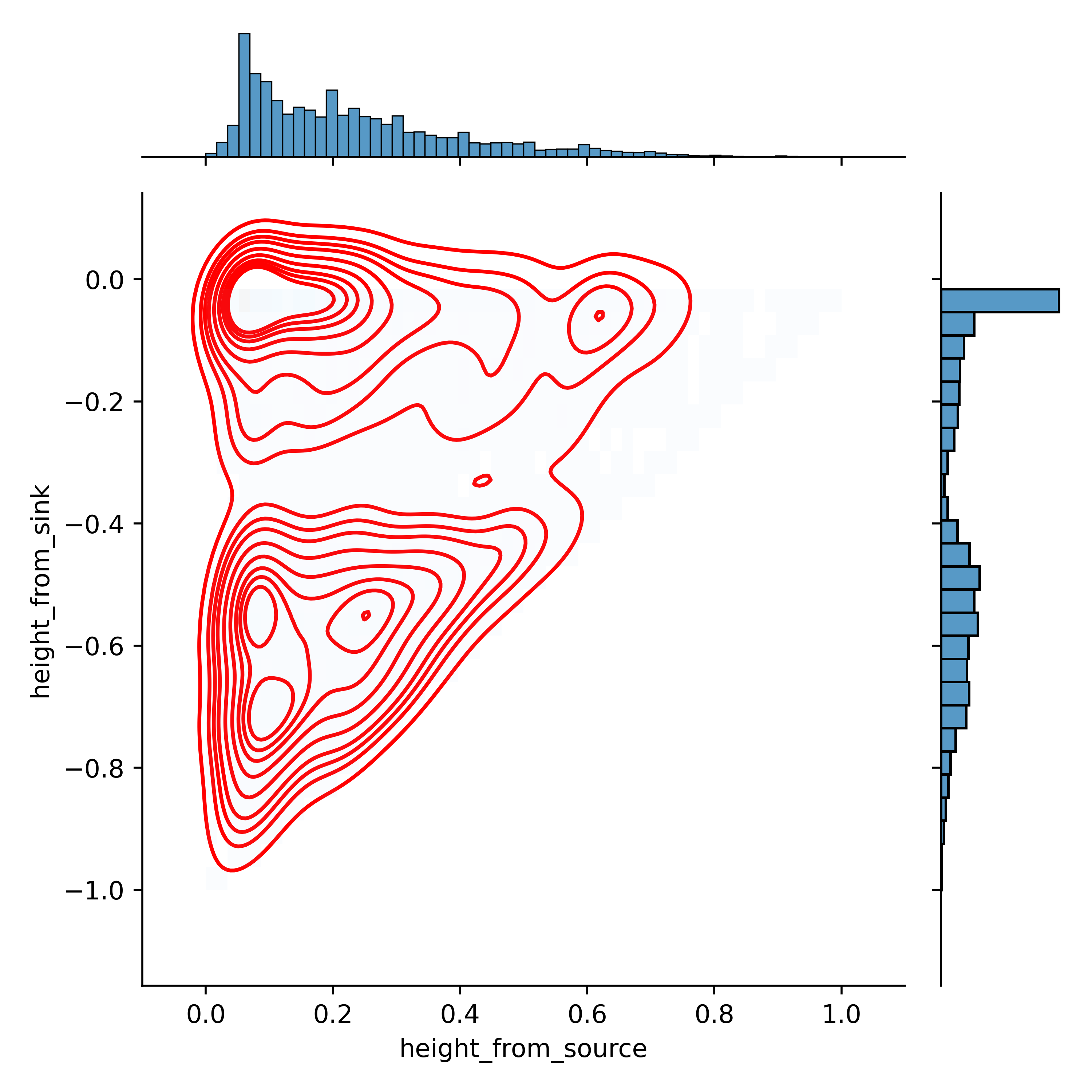}
         \caption{Imvanex, Smallpox, Bavarian Nordic, 2013, WPV}
     \end{subfigure}
     \hfill
     \begin{subfigure}[b]{0.24\textwidth}
         \centering
         \includegraphics[width=\textwidth]{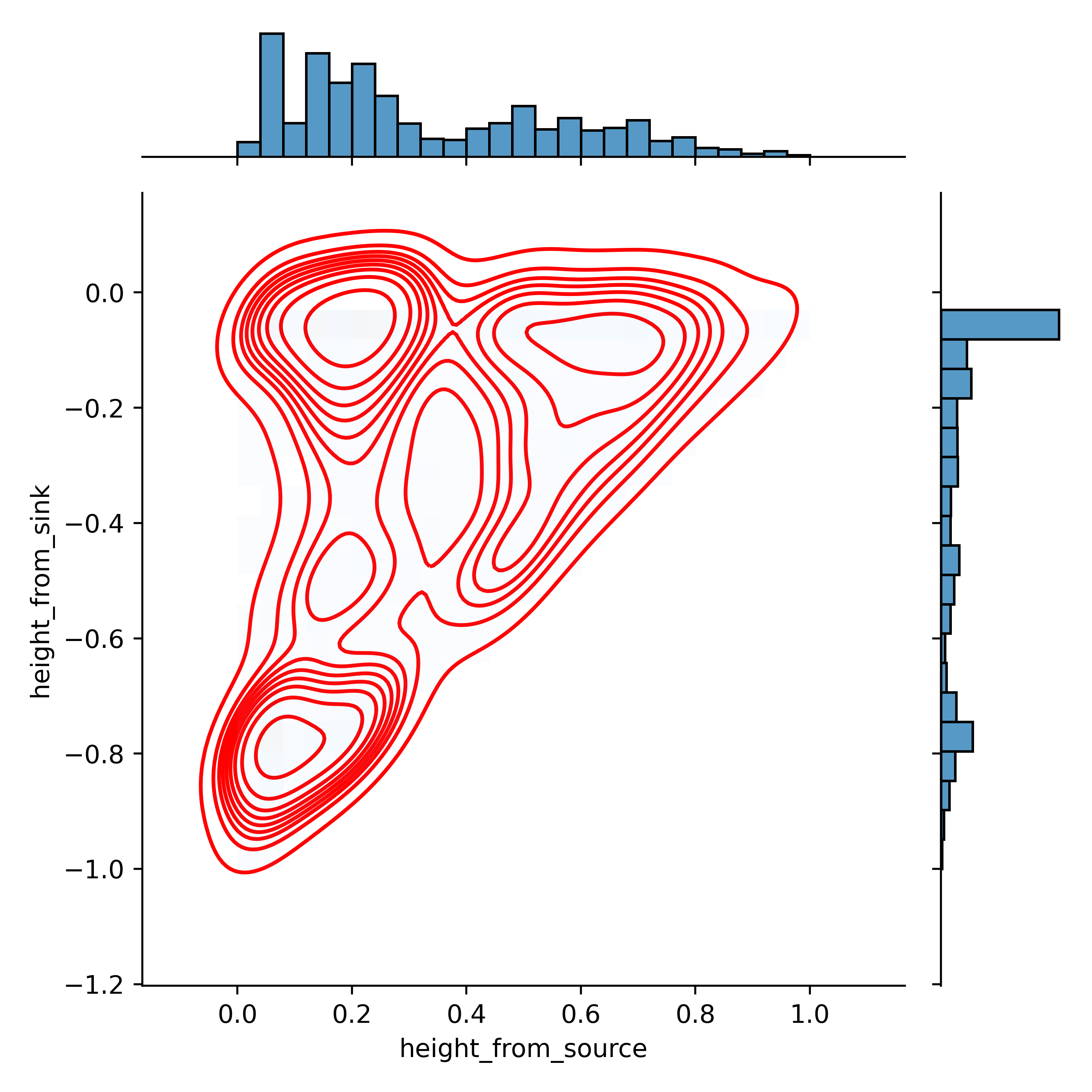}
         \caption{Nuvaxovid, COVID-19, Novavax, 2022, subunits}
     \end{subfigure}
     \hfill
     \begin{subfigure}[b]{0.24\textwidth}
         \centering
         \includegraphics[width=\textwidth]{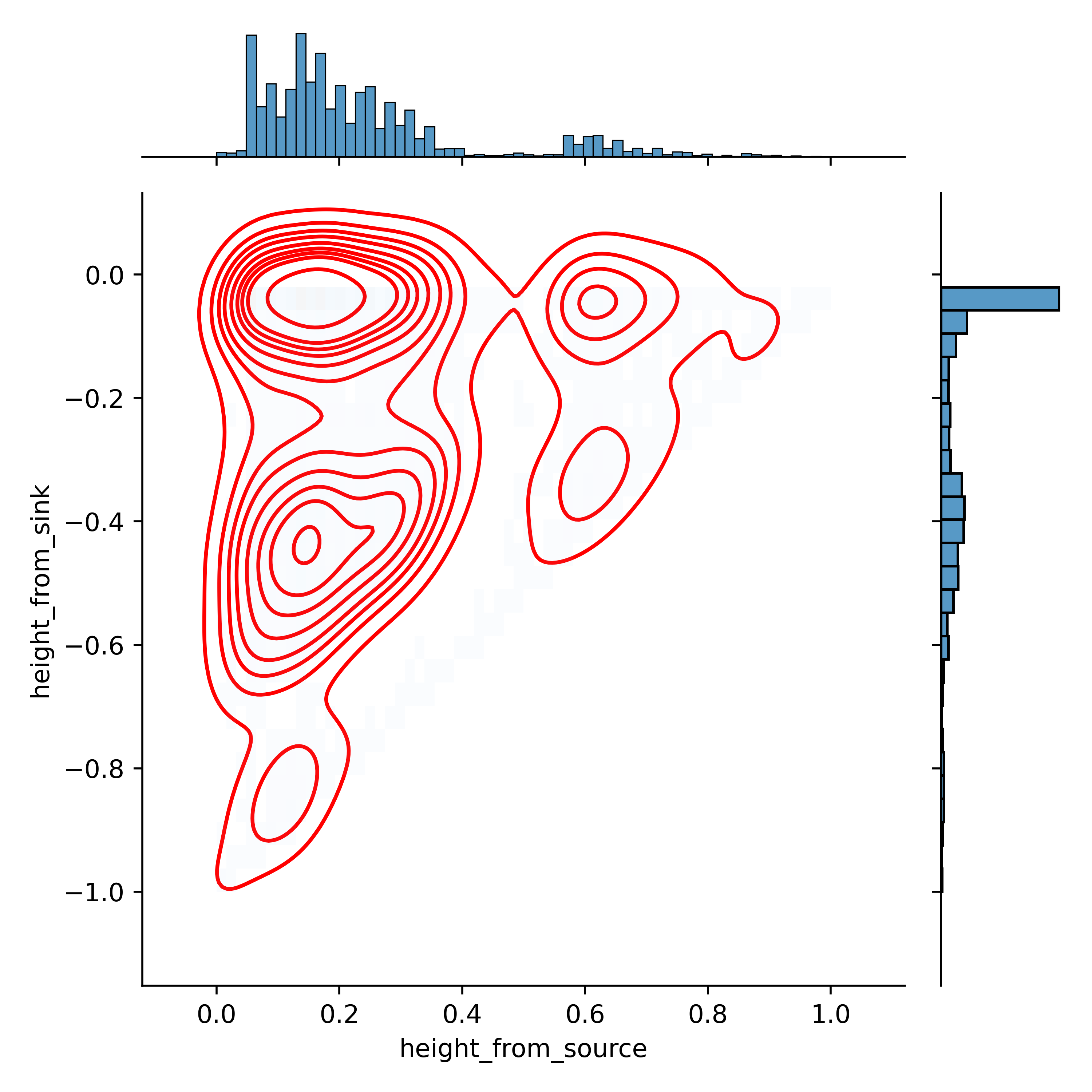}
         \caption{Shringrix, Shingles, GSK, 2017,\\ subunits}
     \end{subfigure}
     \hfill
        \caption{\textbf{Kernel density estimation of height and depth.}}
        \label{fig:all_density}
\end{figure}

\end{document}